\documentclass[12pt]{article}

\usepackage[inner=2.3cm,outer=2.3cm,bottom=2.5cm]{geometry}
\usepackage{amssymb,amsthm, amsmath, caption} 
\usepackage{graphicx,psfrag,epsf,epstopdf}
\usepackage{enumerate}
\usepackage{siunitx}
\usepackage{natbib} 
\usepackage{amsmath,amsfonts,amssymb,amsthm,algorithm}
\usepackage{bm, bbm}    
\usepackage{threeparttable}   
\usepackage{tcolorbox}    
\usepackage{color,float}    
\usepackage{multirow}   
\usepackage{caption}
\usepackage{subcaption}
\usepackage[toc,page]{appendix}        
\usepackage{hyperref}
\hypersetup{     
	colorlinks=true, 
	linkcolor=blue,
	citecolor=blue 
} 
\usepackage{setspace}
\usepackage{algorithm, algorithmic}

\usepackage[figuresright]{rotating}

\newtheorem{theorem}{Theorem}

\newtheorem{definition}{Definition}
\newtheorem{proposition}{Proposition} 
\newtheorem{assumption}{Assumption}  
\newtheorem{lemma}{Lemma}

\newcounter{examplecounter}

\graphicspath{{./figure/}} 

\def \K {\mathcal{K}}
\def \Y {\mathcal{Y}}

\def \R {\mathbb{R}}
\def \B {\mathbb{B}} 
\def \F {\mathbb{F}}  
\def \P {\mathbb{P}}   
\def \GG {\mathbb{G}} 
\def \R {\mathbb{R}}
\def \Z {\mathbb{B}}
\def \E {\mathbb{E}}

\def\1{1\!{\rm l}}

\newcommand\Var {\ensuremath{\mathrm{Var}}}
\newcommand\Cov {\ensuremath{\mathrm{Cov}}}

\newcommand{\indep}{\perp \!\!\! \perp} 

\title{\Large
  \textbf{ 
    Regression Adjustment for Estimating Distributional Treatment Effects in Randomized Controlled Trials}\footnote{
    We would like to express our appreciation to Editor Yuya Sasaki, the Associate Editor, and two anonymous referees for their valuable comments that greatly improved the paper.
      We are also grateful to Yu-Chin Hsu, Chu-An Liu, Yoshimasa Uematsu, Tomu Hirata, Shuting Wu and conference and seminar participants in Kansai Econometric Research Group, 
      Keio Economic Research Workshop, and Academia Sinica for their helpful comments and discussions. 
      Oka gratefully acknowledges financial support by the Japan Society of the Promotion of Science under
      Grant-in-Aid for Scientific Research (C)~24K04821.  
    }
    \vspace{1.0cm}
}

\author{ 
 Tatsushi Oka\footnote{
    Department of Economics, Keio University 
    (\href{mailto:tatsushi.oka@keio.jp}{tatsushi.oka@keio.jp})}  
  \hspace{-0.2cm} \and
  Shota Yasui\footnote{ 
    CyberAgent, Inc.
    (\href{mailto:yasui_shota@cyberagent.co.jp}{yasui\_shota@cyberagent.co.jp})
  } 
  \hspace{-0.2cm} \and 
  Yuta Hayakawa\footnote{ 
    CyberAgent, Inc.
    (\href{mailto:hayakawa_yuta@cyberagent.co.jp}{hayakawa\_yuta@cyberagent.co.jp})
  } 
  \hspace{-0.2cm} \and 
  Undral Byambadalai\footnote{
    CyberAgent, Inc.
    (\href{mailto:undral_byambadalai@cyberagent.co.jp}{undral\_byambadalai@cyberagent.co.jp})} 
  \vspace{0.5cm}
}
 
\begin{document}

\maketitle
\thispagestyle{empty}

\begin{abstract}
In this paper, we address the issue of estimating and inferring distributional treatment effects in randomized experiments. The distributional treatment effect provides a more comprehensive understanding of treatment heterogeneity compared to average treatment effects. We propose a regression adjustment method that utilizes distributional regression and pre-treatment information, establishing theoretical efficiency gains without imposing restrictive distributional assumptions. We develop a practical inferential framework and demonstrate its advantages through extensive simulations. Analyzing water conservation policies, our method reveals that behavioral nudges systematically shift consumption from high to moderate levels. Examining health insurance coverage, we show the treatment reduces the probability of zero doctor visits by 6.6 percentage points while increasing the likelihood of 3-6 visits. In both applications, our regression adjustment method substantially improves precision and identifies treatment effects that were statistically insignificant under conventional approaches.
\end{abstract}
 
\vspace{0.5cm}
\noindent%
{\it Keywords:}
Randomized experiment, A/B testing, Distributional treatment effect, Heterogeneous treatment effect, Regression adjustment, Distributional regression
\vspace{0.2cm} \\
\noindent 
{\it JEL Codes:}
C14, C32, C53, E17, E44

\clearpage
\setcounter{page}{1}
\setstretch{1.5}
\section{Introduction}

Since the foundational works of \cite{fisher1925statistical, fisher1935design} and \cite{Neyman1923}, randomized controlled trials (RCTs) have become essential in evaluating the effectiveness of treatments or interventions. Over the last few decades, RCTs, also known as A/B testing, have emerged as key tools for causal inference across a wide range of scientific and engineering disciplines. These fields include medicine \citep{rubin1997estimating}, economics \citep{duflo2007using, athey2017econometrics}, political science \citep{imai2005get, horiuchi2007designing}, and sociology \citep{baldassarri2017field}, among others. For a more thorough exploration, refer to the modern textbook treatments in \cite{imbens2015causal}.

In RCTs, researchers often encounter the challenge of low detection power, where the variance of the treatment effect estimator is large relative to the actual treatment effect, leading to ambiguous results \citep[e.g.,][]{lewis2015unfavorable}. Contrary to the common belief that large sample sizes alone can resolve this issue, research indicates this is not always the case \citep{Deng2013}. Increasing sample size or experiment duration is costly and not always feasible. Therefore, using estimators with smaller variance can improve precision and lead to more reliable estimates. Additionally, examining distributional treatment effects, not just average treatment effects, is of practical interest to researchers \citep{bitler2006mean}.

In this paper, we introduce a regression adjustment method to improve the precision of estimating the distributional treatment effect (DTE) in RCTs. The proposed method uses pre-treatment information through distributional regression (DR). Our approach addresses the challenge of low detection power, employing a semiparametric approach in that the underlying distributions are left unspecified. This method accommodates a wide range of outcome distributions, including discrete, continuous, and mixed variables, and is useful for documenting heterogeneous treatment effects.

We theoretically demonstrate that regression adjustment can reduce the variance of distribution estimators and enhance the precision of the DTE estimator. 
We establish the uniform asymptotic properties of the DTE estimator.
We also introduce a practical resampling method 
for statistical inference
and establish its theoretical validity. 
Through Monte Carlo simulations, we evaluate our method and show that regression adjustment significantly improves the precision of the DTE estimator and outperforms existing methods. 
Additionally, we apply our method in two empirical studies:
a randomized experiment on the impact of nudges on water usage
and 
the Oregon Health Insurance Experiment.
Our results highlight the importance of distributional treatment estimation with regression adjustment.

Our paper provides several important contributions to the existing literature.
First, our paper contributes to a growing body of literature on heterogeneous treatment effects in econometrics and statistics \citep[e.g.,][]{Athey2016, Wager2018}.
The quantile treatment effect (QTE) was first introduced by \cite{doksum1974empirical} and \cite{lehmann1975nonparametrics} for estimating treatment effects due to unobserved heterogeneity. Since then, various estimation and inference methods for the distributional and quantile treatment effect have been developed and applied, including
\cite{heckman1997making},
\cite{abadie2002instrumental}, \cite{athey2006identification}, \cite{bitler2006mean},
\cite{chernozhukov2005iv},
\cite{djebbari2008heterogeneous},
\cite{donald2014estimation},
\cite{callaway2018quantile}, \cite{callaway2019quantile}, 
\cite{chernozhukov2019generic}, and \cite{kallus2024localized}, among others. 
\cite{jiang2023regression} study the estimation of regression-adjusted QTE under covariate adaptive randomization schemes. To the best of our knowledge, however, techniques for estimating DTE has not been explored in the context of regression-adjustment method.\footnote{
To estimate heterogeneous treatment effects, researchers commonly use two mututally not exclusive approaches: conditional average treatment effect (CATE) and distributional treatment effect (DTE). CATE estimates the average treatment effect for subgroups based on observed variables \citep{Athey2016, Imai2013, Metalearner2019, Nie2020, Wager2018}. DTE, on the other hand, compares outcome distributions between treatment and control groups across the entire spectrum, and thus it can capture effects that CATE might miss, particularly those stemming from unobserved factors. As such, DTE remains valuable even after accounting for observed variables in CATE analyses.
}\textsuperscript{,}\footnote{\cite{jiang2023regression} and our paper complement each other. While \cite{jiang2023regression} considers our experimental scheme as a special case and employs doubly robust moment conditions for quantile regression, our proposed method applies to non-continuous outcome variables and also employs doubly robust moment condition with distribution regression framework.
}

Next, the literature has extensively investigated the use of pretreatment covariates to estimate ATE
with the regression adjustment 
\citep[e.g.][]{fisher1925statistical, cochran1957analysis, cox1982biometrics, frison1992repeated}.
\cite{freedman2008regression,freedman2008regressionA}
and 
\cite{lin2013agnostic}
have noted the potential issue of the adjustment due to finite-sample bias and misspecification. 
Recent works by 
\cite{Yang2001}, \cite{Tsiatis2008} and \cite{Ansel2018}
apply linear regression adjustment, while 
\cite{rosenblum2010simple}
and 
\cite{bartlett2018covariate}
explore nonlinear regression adjustment. 
\cite{list2022using} use machine learning methods for regression adjustment. 
\cite{imbens2015causal}
and 
\cite{negi2021revisiting, negi2020robust}
consider regression adjustment for both linear and nonlinear regression models. 
These existing works focus on average treatment effect (ATE), 
whereas our research broadens the scope of regression adjustment to estimate the impact of treatment on the entire distribution in randomized experiments.

Thirdly, 
for our regression adjustment,
we utilize the DR approach, which can be considered a semiparametric approach for estimating conditional distributions,
in that both marginal and conditional distributions are left unspecified.
The semiparametric property is essential in practice since we rarely know
the true underlying distributions in many applications. 
\cite{williams1972analysis}
introduce DR to analyze ordered categorical outcomes by using multiple binary regressions.
\cite{foresi1995conditional}
extend the approach to characterize a continuous, conditional distribution
and 
\cite{hall1999methods} propose a local version of distributional regression.
\cite{chernozhukov2013inference} establish its uniform validity for estimating the entire conditional distribution when an asymptotical continuum of binary regressions is used. 
See also 
\cite{chernozhukov2020network},
\cite{chernozhukov2023distribution},
\cite{wang2023bivariate, wang2023distributional}
for later developments.

The rest of the paper is organized as follows.
Section \ref{sec:setup} introduces the potential outcomes framework and defines the treatment effect parameters of interest within our study.  Section \ref{sec:theory} presents theoretical results. In Section \ref{sec:simulation}, we provide finite-sample properties through  Monte Carlo simulations. Section \ref{sec:empirical} offers two empirical case studies. Finally, we conclude in Section \ref{sec:con}. 
All proofs and additional simulation results can be found in Supplemental Material.

\section{Treatment Effect and Estimation Method} \label{sec:setup}

This section first introduces the setup and explain a model
for regression adjustment. 
In what follows, 
we use $\E[\cdot]$ and $\Var(\cdot)$
as the expectation and variance operator, respectively. 
All small-$o$ notation and asymptotic results, including $o_p(1)$ and $O_p(1)$, in the paper are relative to sample size $n$.
Let
$\ell^{\infty}(T)$ be the collection of all bounded real-valued functions
defined on an arbitrary index set $T$.
We use $\rightsquigarrow$ to denote convergence in law. 
That is, given $X_{n} \in \ell^{\infty}(T)$
or $X_{n}(t)$ for $t \in T$,
$X_{n} \rightsquigarrow X$ in $\ell^{\infty}(T)$
if $X_{n}$ converge in law to $X$ in $\ell^{\infty}(T)$.

\subsection{Setup and Treatment Effects}

Consider a randomized controlled trial with $K$ treatments, where
each individual is randomly assigned to one of the treatments with pre-specified probabilities
$\{\pi_{k}\}_{k=1}^{K}$ with
$\sum_{k=1}^{K}\pi_{k} =1 $
and 
$ \pi_{k}> 0$
for all $k \in \mathcal{K}:=  \{1, \dots, K \}$.
We consider the potential outcome framework \citep{Neyman1923, rubin1974estimating}.
Let $Y$ denote the (observed) outcome
of interest
with its support $\mathcal{Y} \subset \mathbb{R}$
and 
$Y(k) $ the potential outcome under treatment $k \in \mathcal{K}$.
Then the following relation holds:
$
Y = \sum_{k=1}^{K}W_{k} Y(k) 
$,
where
$W_{k}$ is set to 1 if an individual is assigned to treatment $k \in \mathcal{K}$ and 0 otherwise
and
satisfies
$\Pr(W_{k}= 1) = \pi_{k}$
and 
$\sum_{k=1}^{K} W_{k} = 1$.

Suppose we observe a random sample
$\{(\bm{W}_{i}, \bm{X}_{i}, Y_{i})\}_{i=1}^{n}$
drawn independently from the distribution of $(\bm{W}, \bm{X}, Y)$ with a sample size of $n$.
Here,
$\bm{W} := (W_{1}, \dots, W_{K})^{\top}$ denotes the vector of treatment assignment indicators for $K$ treatments
with its support $\mathcal{W}:= \{0,1\}^{K}$
and 
$\bm{X}$ 
a $p \times 1$ vector of pre-treatment covariates
with its support $\mathcal{X} \subset \mathbb{R}^{p}$.  
We assume that $\bm{X}$ realizes before the randomized controlled trials
and is
independently and identically distributed
even across treatment statuses.
Given 
$\bm{W}_{i}= (W_{1, i}, \dots, W_{K, i})^{\top}$,
the number of observations assigned to treatment $k \in \mathcal{K}$
is defined as $n_{k} := \sum_{i=1}^{n} W_{k,i}$
and the empirical probability of treatment assignment is given by 
$\hat{\pi}_{k}:=n_{k}/n$ for $k= 1, \dots, K$.

One prevalent approach to measure the treatment effect is 
the average treatment effect (ATE), defined as
$
ATE_{k,k'} := \E[Y(k)] - \E[Y(k')] 
$
for $k, k' \in \mathcal{K}$.
ATE is simple to  estimate and interpret,
while it is incapable of estimating the treatment effect on the entire distribution. 
To consider the distributional features,
let
$F_{Y(k)}(y)$ be the distribution function
of $Y(k)$ for treatment status $k \in \mathcal{K}$
and 
$\gamma:=
(
\F_{Y(1)},
\dots, 
\F_{Y(K)}
)^{\top}
\in
\Gamma
\subseteq
\ell^{\infty}(\mathcal{Y})^{K}
$
a vector of the potential outcome distributions.
We define the distributional treatment effect (DTE)
between treatments $k, k' \in \mathcal{K}$ as 
\begin{eqnarray*}
  \Delta^{DTE}_{k,k'}
  :=
  \Psi_{k, k'}^{DTE} (\gamma),
\end{eqnarray*}
where the functional  
$\Psi_{k, k'}^{DTE}: \Gamma \to \ell^{\infty}(\mathcal{Y})$
is defined as 
$
\Psi_{k, k'}^{DTE}(\gamma)
=
F_{Y(k)}
- 
F_{Y(k')}$.
One of its advantages is that 
DTE is well-defined for any type of outcome, including discrete variables.\footnote{
Our approach measures the distributional effect as the difference between the distributions of potential outcomes under distinct treatments. While this paper focuses on DTE, we note that under the rank invariance assumption, our results are also informative about the distribution of individual treatment effects (ITE). Without additional conditions or rank invariance, researchers may consider adopting a partial-identification approach for estimating the distribution of ITE. See \cite{heckman1997making} among others.   }

Moreover,
the distributional information
is an essential building block
for other types of treatment
effect parameters. 
For example,
fixing some constant $h>0$,
we can write the probability 
$\Pr\{y < Y(k) \le y+h \} = F_{Y(k)}(y+h) - F_{Y(k)}(y)$,
thereby defining 
the probability treatment effect (PTE) as 
\begin{align*}
  \Delta^{PTE}_{k,k', h}
  :=
  \Psi_{k, k', h}^{PTE}(\gamma),
\end{align*}
where
the functional 
$\Psi_{k, k', h}^{PTE}: \Gamma \to \ell^{\infty}(\mathcal{Y})$
is defined as 
\begin{eqnarray*}
  \Psi_{k, k', h}^{PTE}(\gamma)(y)
  &=&
      (
      F_{Y(k)}
      -
      F_{Y(k')}
      )(y+h)
      -
      (
      F_{Y(k)}
      - 
      F_{Y(k')}
      )(y).
\end{eqnarray*}
PTE measures changes in the probability
that the outcome variable
falls in interval $(y, y+h]$
and 
can apply for
any type of outcome variables.
For an integer outcomes and $h=1$,
$\Delta_{k, k',1}^{PTE}(y)$
quantifies the probability
change at location $y$,
or 
$\Pr\{Y(k) = y\} - \Pr\{Y(k') = y\}$.
Alternatively, for continuous outcome variables, we can consider the quantile treatment effect (QTE), which is defined through a simple inversion of the distribution function. For $u \in (0,1)$, let the quantile function of $Y(k)$ be defined as the map $u \mapsto F_{Y(k)}^{-1}(u)$.
The QTE between treatments $k, k' \in \mathcal{K}$  
 is defined as 
$
  \Delta^{QTE}_{k, k'} :=  \Psi_{k, k'}^{QTE}(\gamma)(u)
$
where 
$
  \Psi_{k, k'}^{QTE}(\gamma)(u) :=  F_{Y(k)}^{-1}(u) - F_{Y(k')}^{-1}(u).
$
Since standard analysis based on quantile functions presumes the existence of the outcome density, QTE is not readily applicable in the presence of discreteness in the outcome distributions.

\subsection{Distributional Regression}

In this paper, we extend the idea of a regression-adjustment approach
for estimating the DTE and PTE, 
using the distributional regression (DR) framework
with pre-treatment covariates. 
The true conditional distribution of $Y(k)$ given $\bm{X}$
can be written as 
\begin{eqnarray*}
  F_{Y(k)|\bm{X}}(y|\bm{X}) =
  \E
  [\1_{ \{Y(k) \le y\} }
  |\bm{X}
  ],    
\end{eqnarray*}
where $\1_{\{\cdot\} }$
represents the indicator function taking the value of 1 if the condition within
$\{\cdot\}$ is satisfied, and 0 otherwise.

In many practical applications, the true conditional distribution $F_{Y(k)|\bm{X}}(y|\bm{X})$ remains unknown. It is crucial not to impose overly restrictive global parametric constraints. To address this issue, we employ the distribution regression framework, wherein a parametric linear-index model is used to describe the binary outcomes $\1_{ \{Y(k) \le y\} }$ for each $y \in \mathcal{Y}$. 
More specifically, let $\Lambda: \mathbb{R} \to [0,1]$ be a known inverse link function. For each $y \in \mathcal{Y}$, 
we define the following conditional distribution based on a linear index model:
\begin{eqnarray} 
  \label{eq:model}
  G_{Y(k)| \bm{X}}(y| \bm{X}) 
  :=
  \Lambda\big(
  T(\bm{X})^{\top}\beta_{k}(y)
  \big),
\end{eqnarray}
where 
$T: \mathcal{X} \to \R^{d}$
is some transformation, such as polynomial or pair-wise interaction
of regressors, specified by researchers,
and 
$\beta_{k}(y) {\in} \mathbb{R}^{d}$
are  unknown parameters.
Setting $\Lambda(\cdot)$ as either the normal or logistic distribution function
yeilds probit or logit models, respectively,
for binary outcome $\1_{ \{Y(k) \le y\} }$.

It is worth noting that 
the unknown parameters $\beta_{k}(y)$
are specific to point $y \in \mathcal{Y}$,
and can be considered 
as pseudo-parameters to characterize the
true conditional distribution at each point $y$.
These pseudo-parameters 
capture local information of the distributions of interest, and can be estimated at the parametric rate  under certain regularity conditions, as detailed later.
Thus,
we view the conditional distribution model 
$G_{Y(k)| \bm{X}}(\cdot| \bm{X})$
in (\ref{eq:model}) as an approximation model for the true conditional distribution $F_{Y(k)|\bm{X}}(\cdot|\bm{X})$, and our subsequent analysis accommodates potential model misspecifications.

We estimate the DR model in (\ref{eq:model}) using the quasi-maximum likelihood framework \citep[see][]{White1982}, separately for observations exposed to different treatments.
For each $(k, y) \in \mathcal{K}\times\mathcal{Y}$,
the estimator for $\beta_{k}(y)$
is defined as the maximizer 
of the log-likelihood function:
\begin{eqnarray}
  \label{eq:mle} 
  \widehat{\beta}_{k}(y)
  =
  \arg
  \max_{\beta \in \mathcal{B}}
  \widehat{\ell}_{k}(\beta; y),  
\end{eqnarray}
where
the log-likelihood function
for treatment $k$
is denoted by 
$\widehat{\ell}_{k}(\beta; y)  
:=
n_{k}^{-1}
\sum_{i=1}^{n}
W_{k,i}
\cdot
\ell_{i}(\beta;y)$
with
$\ell_{i} (\beta:y) =  \1_{  \{Y_{i} \le y \} } 
\log \Lambda \big(
T(\bm{X}_i)^{\top} \beta \big )
+ 
(1 - \1_{ \{Y_{i} \le y \} } ) \log  
\big \{ 1 -  \Lambda\big(
T(\bm{X}_i)^{\top} \beta 
\big) 
\big \}$
and
$\mathcal{B} \subset \mathbb{R}^{d}$
is the parameter space.
Given the  estimator $\widehat{\beta}_{k}(y)$,
we define the conditional distribution function estimator as 
\begin{eqnarray} 
  \label{eq:dr-est}
  \widehat{G}_{Y(k)|\bm{X}}(y|\bm{x}) 
  :=
  \Lambda\big(
  T(\bm{x})^{\top}\widehat{\beta}_{k}(y)
  \big).
\end{eqnarray}
We can estimate the conditional distribution above for each value of $y \in \mathcal{Y}$ when the outcome variable $Y$ takes finite discrete values, and we select a sufficiently large set of discrete points
$\{y_j \in \mathcal{Y} \}_{j=1}^{J}$
for estimation when the outcome variable is a continuous random variable.\footnote{
One important property
of the conditional distribution function 
is non-decreasing by definition,
while the estimator
$\widehat{G}_{Y(k)|\bm{X}}(y|\bm{x})$
does not necessarily
satisfy monotonicity in finite samples due to estimation errors. 
We can monotonize the conditional distribution estimators using the rearrangement method, for example, as proposed by \cite{chernozhukov2009improving}.
}

\subsection{Regression-Adjusted Treatment Effect Estimator}
By randomization, the potential outcome distributions $F_{Y(k)}$ for $k\in\mathcal{K}$ are identified as the outcome distributions under treatment $k$. Hence, it can be estimated from the data $\{(\bm{W}_{i}, \bm{X}_{i}, Y_{i})\}_{i=1}^{n}$.
The simple DTE and PTE  estimators are
based on the empirical distribution functions
$\widehat{\F}_{Y(k)}^{simple}(y)
:=
n_{k}^{-1}
\sum_{i=1}^{n}
W_{k,i}
\cdot 
\1_{ \{Y_{i} \le y \} }
$
for $k \in \mathcal{K}$.
We propose a new estimator
based on the distributional regression adjustment 
with pre-treatment covariates. 
Define the discrete probability measures
$
\widehat{\P}_{\bm{X}} :=   
n^{-1}\sum_{i=1}^{n} 
\delta_{\bm{X}_i}
$
and 
$
\widehat{\P}_{\bm{X}}^{(k)} :=   
n_{k}^{-1}\sum_{i=1}^{n}
W_{k,i}  \cdot 
\delta_{\bm{X}_i} 
$
for covariates $\bm{X}$
and for $k \in \K$,
where $\delta_{\bm{x}}$ is the measure that assigns mass 1 at $\bm{x} \in \mathcal{X}$. 
For an arbitrary real-valued function $f: \mathcal{X} \to \R$, we denote by 
\begin{align*}
    \widehat{\P}_{\bm{X}} f 
    = \frac{1}{n} \sum_{i=1}^{n}f(\bm{X}_{i}) 
    \ \ \ \ \mathrm{and} \ \ \ \
    \widehat{\P}_{\bm{X}}^{(k)} f 
    = \frac{1}{n_{k}} \sum_{i=1}^{n} W_{k,i} \cdot f(\bm{X}_{i}) . 
\end{align*}
Then,
for treatment $k \in \mathcal{K}$,
we define the regression-adjusted distribution function
as,
\begin{align}
  \label{eq:RA}
  \widehat{\F}_{Y(k)}
  := 
  \widehat{\P}_{\bm{X}}
  \widehat{G}_{Y(k)|\bm{X}}
  \ \ \ 
  \mathrm{in} \  \ell^{\infty}(\mathcal{Y}).
\end{align}
The regression-adjusted estimator is the sample average of conditional distribution functions over pre-treatment covariates utilizing all observations, where these conditional distributions are estimated using only observations pertaining to specific treatments.

Letting
$\widehat{\gamma}
:=
(
\widehat{\F}_{Y(1)},
\dots, 
\widehat{\F}_{Y(K)}
)^{\top}
$, 
we define the regression-adjusted DTE and PTE estimators:
\begin{eqnarray}
  \label{eq:ra-est}
  \widehat{\Delta}^{DTE}_{k,k'}:=
  \Psi_{k, k'}^{DTE} (\, \widehat{\gamma} \, )
  \ \ \ \mathrm{and}  \ \ \
  \widehat{\Delta}^{PTE}_{k,k',h}:=       
  \Psi_{k, k', h}^{PTE} (\, \widehat{\gamma} \,).
\end{eqnarray}
Similarly, the QTE estimator is defined as 
$  \widehat{\Delta}^{QTE}_{k,k'}:=
  \Psi_{k, k'}^{QTE} (\, \widehat{\gamma} \, )
$.

 The linear index model in (\ref{eq:dr-est}) can be viewed as a special case of generalized linear models (GLMs). Within the GLM framework, using the canonical link function yields unbiased estimation of the unconditional mean in finite samples (see Section~\ref{sec:glm-canonical} for more details and \citealp{mccullagh1989binary}). Specifically, the canonical link function must satisfy
\begin{align}
\label{eq:canonical}
  \frac{1}{  n_{k}}
  \sum_{i=1}^{n}
  W_{k, i} \cdot
  \big (
  \1_{ \{Y_{i} \le y \} } 
  -
  \widehat{G}_{Y(k)|\bm{X}}(y|\bm{X}_{i}) 
  \big ) = 0. 
\end{align}
At the population level, this relation can be expressed as:
$
\E
\big[
\1_{ \{Y(k) \le y \} } 
-
G_{Y(k)|\bm{X}}(y|\bm{X})
\big ] = 0. 
$
This property is important 
for the regression adjustment
to guarantee that 
the regression adjustment does not induce
bias to improve the precision of the treatment
effect estimator.
\cite{negi2021revisiting} pointed this out 
for a nonlinear regression adjustment method 
and derive the asymptotic property  
when the true distribution is known. 
In what follows, we focus
the GLMs with the canonical link functions
for binary outcome, which includes
OLS and logit, while we do not assume 
any global parametric assumptions
and allow for model mis-specification.

\section{Theoretical Results} \label{sec:theory}

In this section, we first derive the oracle efficiency of the proposed estimator. 
We subsequently present asymptotic properties and
a resampling method.

\subsection{Oracle Estimator and Efficiency}

In this section, we investigate
possible efficiency gain
from the regression adjustment,
when the population counterpart
related to the adjustment is available.  
This analysis will
allow us to establish
the efficiency property.

Leveraging the finite-sample unbiasedness property in (\ref{eq:canonical}), we can
rewrite (\ref{eq:RA}) as 
\begin{align*}
  \widehat{\F}_{Y(k)}
  := 
  \widehat{\F}_{Y(k)}^{simple}
  +
  \big (
  \widehat{\P}_{\bm{X}}
  -
  \widehat{\P}_{\bm{X}}^{(k)}
  \big ) 
  \widehat{G}_{Y(k)|\bm{X}},
\end{align*}
where
the second term of the right-hand side
involves the DR estimator
for the adjustment. 
In the above equations,
we meanwhile consider
the situation where 
the estimator $\widehat{G}_{Y(k)|\bm{X}}$
is replaced by the population model 
$G_{Y(k)|\bm{X}}(y|\bm{X})$
in (\ref{eq:model}). 
Then, define the oracle distribution function with
the regression adjustment 
\begin{align*}
    \widetilde{\F}_{Y(k)}
    &:= 
    \widehat{\F}_{Y(k)}^{simple}
    +
      \big (
      \widehat{\P}_{\bm{X}}
      -
      \widehat{\P}_{\bm{X}}^{(k)}
      \big ) 
        G_{Y(k)|\bm{X}} .
\end{align*}
Theorem \ref{theorem:distribution} below presents the efficiency gain resulting from the application of our regression adjustment method, given the oracle estimator above. Although this oracle estimator is infeasible,  
the subsequent section (specifically, in the proof of Theorem \ref{theorem:ate-ra-fclt}) establishes an asymptotic equivalence between the feasible and oracle estimators.

For Theorem \ref{theorem:distribution}, we make the following assumptions regarding the experimental setup and sampling scheme:

\vspace{0.3cm}
\begin{assumption} 
  \label{as:as1}  \ \\ 
  \noindent 
  (a)
  The observations 
  $\{(\bm{W}_{i}, \bm{X}_{i}, Y_{i})\in \{0,1\}^{K}{\times}\mathcal{X}{\times}\mathcal{Y}\}_{i=1}^{n}$
  are $n$ independent copies of the random variables
  $(\bm{W}, \bm{X}, Y)$. \\
  (b)
  Treatment assignment is independent of both potential outcomes and pre-treatment covariates:
$  \bm{W} \indep \big(Y(1), \dots, Y(K), \bm{X})$. 
The pre-specified treatment probabilities $\{\pi_{k}\}_{k=1}^{K}$ satisfy $\pi_{k}>0$ for every $k \in \mathcal{K}$ and $\sum_{k=1}^{K} \pi_{k}= 1$. \\
\end{assumption}
\vspace{0.01cm}

\noindent
Assumption \ref{as:as1} (a) requires random sampling and
Assumption \ref{as:as1} (b) is standard under a randomized controlled trial with multiple treatments.

\vspace{0.2cm}
\begin{theorem}
  \label{theorem:distribution}
  (a)
  Suppose that Assumption
  \ref{as:as1}  
  holds
  and that $\hat{\pi}_{k} = \pi_{k} + o(1)$ as $n\to \infty$ for every $k \in \mathcal{K}$.
  Then,
  for each $k \in \mathcal{K}$
  and
  for any $y \in \mathcal{Y}$,
  we have in $\ell^{\infty}(\mathcal{Y})$, 
  \begin{align*}
    \Var \big(
      \hat{\F}_{Y(k)}^{simple}
    \big)
    \ge 
    \Var \big(
    \widetilde{\F}_{Y(k)}
    \big)
    + o(n^{-1}),
  \end{align*}
  provided that 
  $
    \E
    [
    (
      \1_{ \{Y(k) \le \cdot\} }      
      -
      F_{Y(k)}
    )^2
    ]
    \ge  
    \E
    [
     (
      \1_{ \{Y(k) \le \cdot\} }      
      -
      G_{Y(k)|\bm{X}}
     )^2
     ]
  $. 

 \noindent
   (b) 
   Define 
   $\widehat{\gamma}^{simple}
    :=
    (
    \widehat{\F}_{Y(1)}^{simple},
    \dots, 
    \widehat{\F}_{Y(K)}^{simple}
    )^{\top}
    $ 
  and 
     $\widetilde{\gamma}
    :=
    (
    \widetilde{\F}_{Y(1)},
    \dots, 
    \widetilde{\F}_{Y(K)}
    )^{\top}
    $. 
    Additionally, we assume that the model specification for distributional regression in (\ref{eq:model}) is correct.
    Then,
        \begin{align*}
            \Var
            \big(
            \widehat{\gamma}^{simple}
            \big)
            \succeq
            \Var
            \big(
            \widetilde{\gamma}
            \big) + o(n^{-1}),
        \end{align*}
        where $\succeq$ denotes the positive semi-definiteness. 
        When
        $
               \Var 
     \big (
         F_{Y(k)}(y|\bm{X})
         -
         r \cdot 
         F_{Y(k')}(y|\bm{X})
         \big )
         > 0 
        $
 for any distinct
$k, k' \in \K$ and $r \in \R$, 
        the positive definite result holds.   

\end{theorem}
\vspace{0.2cm}

Theorem \ref{theorem:distribution}(a) shows the efficiency gains achieved by applying regression adjustment to estimate unconditional distribution functions. More specifically, the variance of the unconditional distribution estimators decreases if $G_{Y(k)|\bm{X}}(y)$ has some predictive power for $\1_{\{Y(k) \le y\}}$ in the $L^2$ sense. 
When $\Lambda(\cdot)$ is the identity function or in the case of a linear prediction model, this property holds automatically.
In the case of logistic regression, the level of model misspecification must be sufficiently small relative to the predictive signal in the data.\footnote{
    A simple algebra shows that 
    $
    \E
    [
    (
      \1_{ \{Y(k) \le \cdot\} }      
      -
      F_{Y(k)}
    )^2
    ]
    -   
    \E
    [
     (
      \1_{ \{Y(k) \le \cdot\} }      
      -
      G_{Y(k)|\bm{X}}
     )^2
     ]
     = 
         \E
    [
    (
      F_{Y(k)|\bm{X}}      
      -
      F_{Y(k)}
    )^2
    ]
    -   
    \E
    [
     (
      F_{Y(k)|\bm{X}}      
      -
      G_{Y(k)|\bm{X}}
     )^2
     ]
  $. 
  Here, the first term measures the signal in $\bm{X}$ for predicting $\1_{ \{Y(k) \le \cdot\} }$. 
  The second term quantifies the model misspecification, that is,   
  the $L^2$ distance between 
  the true conditional distribution 
  $F_{Y(k)|\bm{X}}$
  and 
  the specified model 
  $      
  G_{Y(k)|\bm{X}}
  $.  
   When the model misspecification (second term) is sufficiently small relative to the signal in the data (quantified by the variance of \( F_{Y(k)|\bm{X}} \)), the predictor \( G_{Y(k)|\bm{X}} \) achieves variance reduction for the unconditional estimator.
}

Furthermore, Theorem \ref{theorem:distribution}(b) elaborates on these gains in a positive semi-definite sense for a vector of regression-adjusted estimators. This result encompasses any linear combination of the unconditional distribution estimators and implies potential efficiency gains for both DTE and PTE estimators. While a similar positive-definiteness result for non-linear regression adjustment is reported in Theorem 6 of \cite{negi2020robust}, our result provides an additional insight that the efficiency gains occur when the regression models are not linearly dependent across treatments. Theorem \ref{theorem:distribution}(b) requires correct specification, which is analogous to the requirement in \cite{negi2020robust} for non-linear regression adjustment. 
We also show the point-wise semiparametric efficiency of the DTE estimator under correct specification in Section \ref{subsection:Semiparametric Efficiency Bound} for the sake of completeness.

\subsection{Asymptotic Distribution of DTE Estimator} \label{sec:asymptotic}

In this subsection, we  derive the asymptotic distribution
of the DTE and PTE estimators,
which is useful for statistical inference. 
We also introduce a resampling technique
and show its asymptotic validity. 
We obtain the results in this section under the following assumptions:

\vspace{0.2cm}
\begin{assumption} 
  \label{as:as2}  
  For any $y \in \mathcal{Y}$ and $k \in \mathcal{K}$, the following conditions hold. \\ 
 \noindent
  (a)
  The log-likelihood functions 
  $\beta \mapsto
  \ell_{k}(\beta;y)$
  is a concave function.
  The link function 
  is  
  a canonical link function, 
  and
  its inverse map 
  $\Lambda(\cdot)$
  is 
  twice continuously differentiable
  with its first derivative $\lambda(\cdot)$.

  \noindent 
  (b)
  The true parameters
  $\beta_{k}(y)$
  uniquely
  solve 
  the maximization problem in
  $\max_{\beta}\E[\widehat{\ell}_{k}(\beta;y)]$
  and
 are contained in
  the interior
  of 
  the compact parameter space $\mathcal{B}$.
  
  \noindent 
  (c)
  The maximum eigenvalues of 
  $H_{k}(y)$
  is strictly negative
  uniformly
  over $y \in \mathcal{Y}$.
  
  \noindent 
  (d)
  If $Y$ is a continuous random variable,
  then 
  the conditional density function  
  $f_{Y|\bm{X} }(y| \bm{x})$ exists, 
  is uniformly bounded, and is uniformly continuous
  in $(y, x) \in \mathcal{Y} \times \mathcal{X}$.
\end{assumption}
\vspace{0.2cm}

Assumption \ref{as:as2} (a)-(c) are needed for deriving
the asymptotic property of the DR estimators
and are standard for the likelihood estimation of binary-choice models in the literature.
Assumption \ref{as:as2} (d) is required only when
the outcome variable is continuous.

In the below theorem, we derive the limiting distribution of
the regression-adjusted DTE and PTE estimators. 

\vspace{0.2cm}
\begin{theorem}
  \label{theorem:ate-ra-fclt}
  Suppose that 
  Assumptions \ref{as:as1} and \ref{as:as2} hold. \\
  (a) Then, we have, 
    for every $k, k' \in \mathcal{K}$ and some constant $h>0$,
    in
  $ \ell^{\infty}(\mathcal{Y})$,
    \begin{align*}
         \sqrt{n}
        (
        \widehat{\Delta}^{DTE}_{k,k'}
        -
        \Delta^{DTE}_{k,k'}
        )
      \rightsquigarrow 
      \Psi^{DTE}_{k,k'}(\Z)
      \ \ \ \mathrm{and} \ \ \ 
         \sqrt{n}
        (
        \widehat{\Delta}^{PTE}_{k,k',h}
        -
        \Delta^{PTE}_{k,k',h}
        )
      \rightsquigarrow 
      \Psi^{PTE}_{k,k',h}(\Z),
    \end{align*}
    where $\B$
  is a mean-zero Gaussian process
  defined in 
  equation (\ref{eq:B-def}) of Supplemental Material.

  \noindent
  (b) 
  Additionally, if the potential outcome $Y_{k}$ has a continuously differentiable distribution function with strictly positive density $f_{Y_k}$ on its support for every $k \in \mathcal{K}$, then
  \begin{align*}
         \sqrt{n}
        (
        \widehat{\Delta}^{QTE}_{k,k'}
        -
        \Delta^{QTE}_{k,k'}
        )
      \rightsquigarrow 
      \psi^{DTE}_{k,k', \gamma}(\Z),
    \end{align*}
    where 
    $\psi^{DTE}_{k,k', \gamma}(h)(u)
    :=
    h_{k} \circ F_{Y_{k}}^{-1}(u) / f_{Y_{k}} \circ F_{Y_{k}}^{-1}(u)
    - 
    h_{k'} \circ F_{Y_{k'}}^{-1}(u) / f_{Y_{k'}} \circ F_{Y_{k'}}^{-1}(u)$ 
    is the Hadamard derivative of $\Psi^{DTE}_{k,k'}(\cdot)$
    at $\gamma$ in direction $h:=(h_{1}, \dots, h_{K})^{\top}$.

\end{theorem}
\vspace{0.2cm}

The limiting processes in the above theorem rely solely on the sampling error of the original data, $\{(\bm{W}_{i}, \bm{X}_{i}, Y_{i})\}_{i=1}^{n}$, and are free from estimation errors in DR. This is ensured through the use of the stochastic equicontinuity property in the theorem's proof, which is in line with the argument presented in \cite{newey1994asymptotic}. This result highlights the significance of our DR-based regression adjustment, as it does not add additional variance to the DTE and PTE estimators for large sample sizes.


Although it is simplified,  
the limiting processes of the DTE and PTE estimators depend on unknown nuisance parameters
and may complicate inference in finite samples.
To address the challenge of non-pivotal limiting processes in finite samples, we turn to the exchangeable bootstrap \citep[as described in][]{praestgaard1993Exchangeably, van1996weak}. This approach consistently estimates limit laws of relevant empirical distributions, allowing us to consistently estimate the limit process of the estimator via the functional delta method.
 

For the resampling scheme, we introduce a vector of random weights 
$(S_{1}, \dots, S_{n})$.
To establish the validity of the bootstrap, 
we assume that the random weights satisfy the 
following assumption.

\vspace{0.1cm}
\begin{assumption}[Bootstrap]
  \label{as:b}
Let
$(S_{1}, \dots, S_{n})$
be $n$ scalar, nonnegative exchangeable  
random variables that are independent of the original sample.
The random weights 
$\{S_{i}\}_{i=1}^{n}$
satisfy the followings:
$
\E|S_{i}|^{2+\epsilon}
< \infty$
for some $\epsilon>0$,
$\bar{S}_{n}
:=
n^{-1}
\sum_{i=1}^{n}
S_{i}
\to^p
0
$
and 
$
n^{-1} \sum_{i=1}^{n}
(
S_{i}
- 
\bar{S}_{n}
)^2
\to^p 1
$.
\end{assumption}
\vspace{0.1cm}

This resampling scheme encompasses 
a variety of bootstrap methods, such as 
the empirical bootstrap, subsampling, wild bootstrap and so on 
\citep[See][]{van1996weak}.

Under treatment $k \in \mathcal{K}$,
the empirical distribution of the outcome variable with the random weights 
and the empirical probability measure of $\bm{X}$ with the random weights
are respectively given by 
\begin{align*}
  \widehat{\F}_{Y(k)}^{simple  \ast}(y)
   :=
  \frac{1}{n_{k}}
  \sum_{i=1}^{n}
  S_{i}
  \cdot 
  W_{k,i}
  \cdot 
  \1_{ \{Y_{i} \le y \} } 
  \ \ \ \ \mathrm{and} \ \ \ \ 
  \widehat{\P}_{\bm{X}}^{\ast} :=   
  \frac{1}{n}
  \sum_{i=1}^{n}
  S_{i}
  \cdot  \delta_{\bm{X}_i}.
\end{align*}
Then,
for treatment $k  \in \mathcal{K}$,
we define regression-adjusted distribution functions
in $\ell^{\infty}(\mathcal{Y})$
as, 
\begin{align}
  \label{eq:RA-B}
  \widehat{\F}_{Y(k)}^{\ast}
  := 
  \widehat{\P}_{\bm{X}}^{\ast}
  \widehat{G}_{Y(k)|\bm{X}}^{\ast}.
\end{align}
In the above equation, we use 
$\widehat{G}_{Y(k)|\bm{X}}^{\ast}$
the conditional distribution estimator 
based on the original sample, 
instead of the bootstrap sample. This improves the computational efficiency of our bootstrap inference by eliminating the need to reestimate the parameters for each bootstrap repetition.
According to Theorem \ref{theorem:ate-ra-fclt}, the conditional distribution estimator 
has no impact on the limit distribution 
and we keep the estimator throughout
the resampling procedure. The 
validity is proven in the theorem below. 

Letting 
$\widehat{\gamma}^{\ast} := ( \widehat{\F}_{Y(1)}^{\ast}, \dots, \widehat{\F}_{Y(K)}^{\ast})^{\top} $, we define 
the bootstrapped treatment effect estimator
with regression adjustment as
\begin{eqnarray*}
  \label{eq:ra-est-boot}
  \widehat{\Delta}^{DTE \ast}_{k,k'}:=
  \Psi_{k, k'}^{DTE} (\, \widehat{\gamma}^{\ast} \, ), 
  \ \ \
  \widehat{\Delta}^{PTE \ast}_{k,k',h}:=       
  \Psi_{k, k', h}^{PTE} (\, \widehat{\gamma}^{\ast}  \,) 
     \ \ \mathrm{and}  \ \ 
  \widehat{\Delta}^{QTE \ast}_{k,k',h}:=       
  \Psi_{k, k', h}^{QTE} (\, \widehat{\gamma}^{\ast}  \,).  
\end{eqnarray*}
In the theorem provided below, we show that 
the exchangeable bootstrap provides a method to consistently estimate 
the limit process of the regression-adjusted DTE, PTE and QTE estimators, following \cite{van1996weak}.

\vspace{0.2cm}
\begin{theorem}
  \label{theorem:ate-ra-fclt-bootstrap}
  Suppose that Assumptions \ref{as:as1}-\ref{as:b} hold. \\
  (a) Then, we have, in
  $ \ell^{\infty}(\mathcal{Y})$,
   \begin{align*}
     \sqrt{n}
    (
    \widehat{\Delta}^{DTE \ast}_{k,k'}
    -
    \widehat{\Delta}^{DTE}_{k,k'}
    )
      \overset{p}{\rightsquigarrow} 
  \Psi^{DTE}_{k,k'}(\Z)
  \ \ \mathrm{and}  \ \
    \sqrt{n}
    (
    \widehat{\Delta}^{PTE \ast}_{k,k',h}
    -
    \widehat{\Delta}^{PTE}_{k,k',h}
    )
      \overset{p}{\rightsquigarrow} 
  \Psi^{PTE}_{k,k',h}(\Z),
  \end{align*}
    for every $k, k' \in \mathcal{K}$ and some constant $h>0$, 
  where  
  $\overset{p}{\rightsquigarrow}$ denotes  
  the conditional weak convergence in probability, 
  which is explained in more detail in 
  Supplemental Material \ref{subsec:appendix-bootstrap} 
    and $\B$
  is a mean-zero Gaussian process defined in Lemma
  \ref{lemma:donsker2} of Supplemental Material.

  \noindent
  (b) 
  Additionally, if the potential outcome $Y_{k}$ has a continuously differentiable distribution function with strictly positive density $f_{Y_k}$ on its support for every $k \in \mathcal{K}$, then
  \begin{align*}
         \sqrt{n}
        (
        \widehat{\Delta}^{QTE \ast}_{k,k'}
        -
        \widehat{\Delta}^{QTE}_{k,k'}
        )
      \rightsquigarrow 
      \psi^{DTE}_{k,k', \gamma}(\Z),
    \end{align*}
    where 
    $\psi^{DTE}_{k,k', \gamma}(\cdot)$ is defined in Theorem \ref{theorem:ate-ra-fclt}(b).

\end{theorem}
\vspace{0.2cm}

Theorems \ref{theorem:ate-ra-fclt} and \ref{theorem:ate-ra-fclt-bootstrap} present a methodologically rigorous and empirically useful approach for estimating DTE and PTE. Moreover, these theoretical results establish the asymptotic validity of the exchangeable bootstrap procedure for constructing confidence intervals and conducting hypothesis tests for both DTE and PTE, as well as QTE under additional conditions.

The following algorithm describes the construction of confidence intervals for regression-adjusted DTE using empirical bootstrap. The same procedure can be applied to obtain confidence intervals for regression-adjusted PTE, QTE and simple estimators.


\begin{algorithm}[H] 
\caption{Bootstrap confidence intervals for regression-adjusted DTE} \label{algorithm:bootstrap}
\begin{algorithmic}
\item[\textbf{Input:}] Original sample $\{(\bm{W}_{i}, \bm{X}_{i}, Y_{i})\}_{i=1}^{n}$
\item[\textbf{Output:}] $(1-\alpha)\times100\%$ confidence intervals for the regression-adjusted DTE
\vspace{0.5em}
\item[\textbf{1.}] For each bootstrap iteration $b=1,\ldots,B$:
\vspace{0.3em}
\item[\textbf{2.}] \quad Draw a bootstrap sample of size $n$ with replacement:\\
      \qquad $\{(\bm{W}_{i}^b, \bm{X}_{i}^b, Y_{i}^b)\}_{i=1}^{n}$ from $\{(\bm{W}_{i}, \bm{X}_{i}, Y_{i})\}_{i=1}^{n}$
\vspace{0.3em}
\item[\textbf{3.}] \quad Compute regression-adjusted DTE $\widehat{\Delta}^{DTE, b}_{k,k'}$ given the conditional\\
      \qquad distribution estimator based on the original sample
\vspace{0.3em}
\item[\textbf{4.}] Calculate standard errors $\widehat{\Sigma}^{DTE}_{k, k'}(y)$ using the bootstrap iterations
\vspace{0.3em}
\item[\textbf{5.}] Construct the confidence band:
      $$\big \{\widehat{\Delta}^{DTE}_{k,k'}(y) \pm \Phi^{-1}(1-\alpha/2)\cdot \widehat\Sigma^{DTE}_{k,k'}(y): y\in\Y \big \},$$
      \quad where $\Phi$ is the standard normal distribution function
\end{algorithmic}
\end{algorithm}

In the above algorithm, the standard errors $\widehat{\Sigma}^{DTE}_{k,k'}(y)$ can be computed 
by 
(a) using the bootstrap standard deviation of the bootstrapped DTEs $\{\widehat{\Delta}^{DTE,b}_{k,k'}(y)\}_{b=1}^B$, or 
(b) using the rescaled bootstrap interquartile range 
$\big ( q_{0.75}(y) - q_{0.25}(y) \big ) / (z_{0.75} - z_{0.25})$, 
where $q_p(y)$ is the $p$-th quantile of the bootstrapped DTE at location $y$ and $z_p$ is the $p$-th quantile of the standard normal distribution.

\section{Monte-Carlo Simulation Results} \label{sec:simulation}

For our Monte Carlo simulation, 
we consider various setups to cover 
continuous/discrete and symmetric/skewed outcome distributions.
We report our results from four types of data generating processes (DGPs): DGP1-DGP4. 
In all DGPs, 
we use the same setup for pre-treatment regressors 
$\bm{X}= (1, X_{1}, X_{2})^{\top}$
with a uniform random variable 
$X_{1} \sim U(0.5, 1.5)$
and 
a normal random variable
$X_{2} \sim N(0,1)$. 
The treatment status 
$W_{1}$ follows the binomial distribution 
with $\Pr\{W_{1}=1\} = \pi_{1} \in \{0.3, 0.5\}$.

Under DGP1 and DGP2,
we generate continuous outcomes 
from the following model:
 \begin{align*}
     Y = X_{1} + (X_{1} + X_{2}) \cdot W_{1}
     + |X_{1} + X_{2}| \cdot U,  
 \end{align*}
where the unobserved component $U$ is specified as
$ U \sim N(0,1)$  under DGP1 
and 
$U \sim \chi_{3}^2 $
under DGP2, respectively.
Consequently, the unobserved variable is symmetrically distributed under DGP1, while it exhibits a skewed distribution under DGP2.
Under DGP3 and DGP4, 
we consider discrete outcome distributions
satisfying the following conditional mean:
\begin{align*}
    \E[Y|\bm{X},\bm{W}] = \exp(W_{1} + X_1 + X_2 /2).
\end{align*}
Given the conditional mean model above, we set the outcome distribution conditional on $\bm{X}$ and $\bm{W}$ to follow the Poisson distribution under DGP3 and to follow the negative binomial distribution with the dispersion parameter $r=5$ under DGP4, respectively.

\begin{figure}[!h]
    \vskip 0.2in
    \begin{center}
    \caption{Performance metrics of simple and regression-adjusted DTE estimators}
    (DGP1, continuous outcome, $\pi_{1}=0.5$)
    
    \includegraphics[width=0.9\columnwidth]{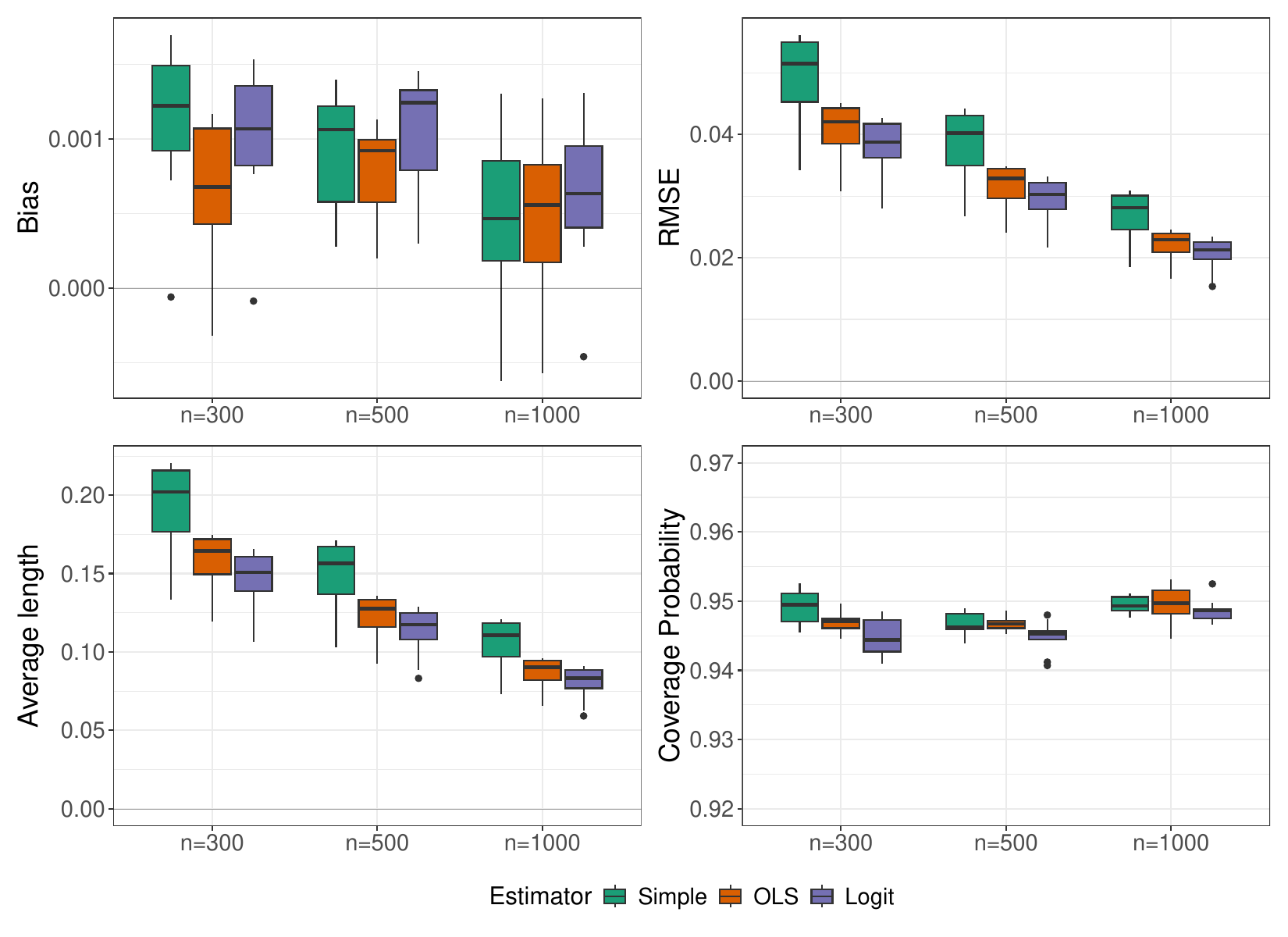}
      \begin{minipage}{0.92\textwidth}
      \small 
      \textit{Note}:
      Bias, RMSE, 95\% CI length and coverage probability calculated over 10,000 simulations. Each boxplot represents the distribution across locations $y$ for a specific sample size.
      \end{minipage}    
      \label{fig:dgp-1-rho0.5}
        \end{center}
    \vskip -0.2in
    \end{figure}

\begin{figure}[!h]
\vskip 0.2in
\begin{center}
\caption{Performance metrics of simple and regression-adjusted DTE estimators}
(DGP3, discrete outcome, $\pi_{1}=0.5$): 
\includegraphics[width=0.9\columnwidth]{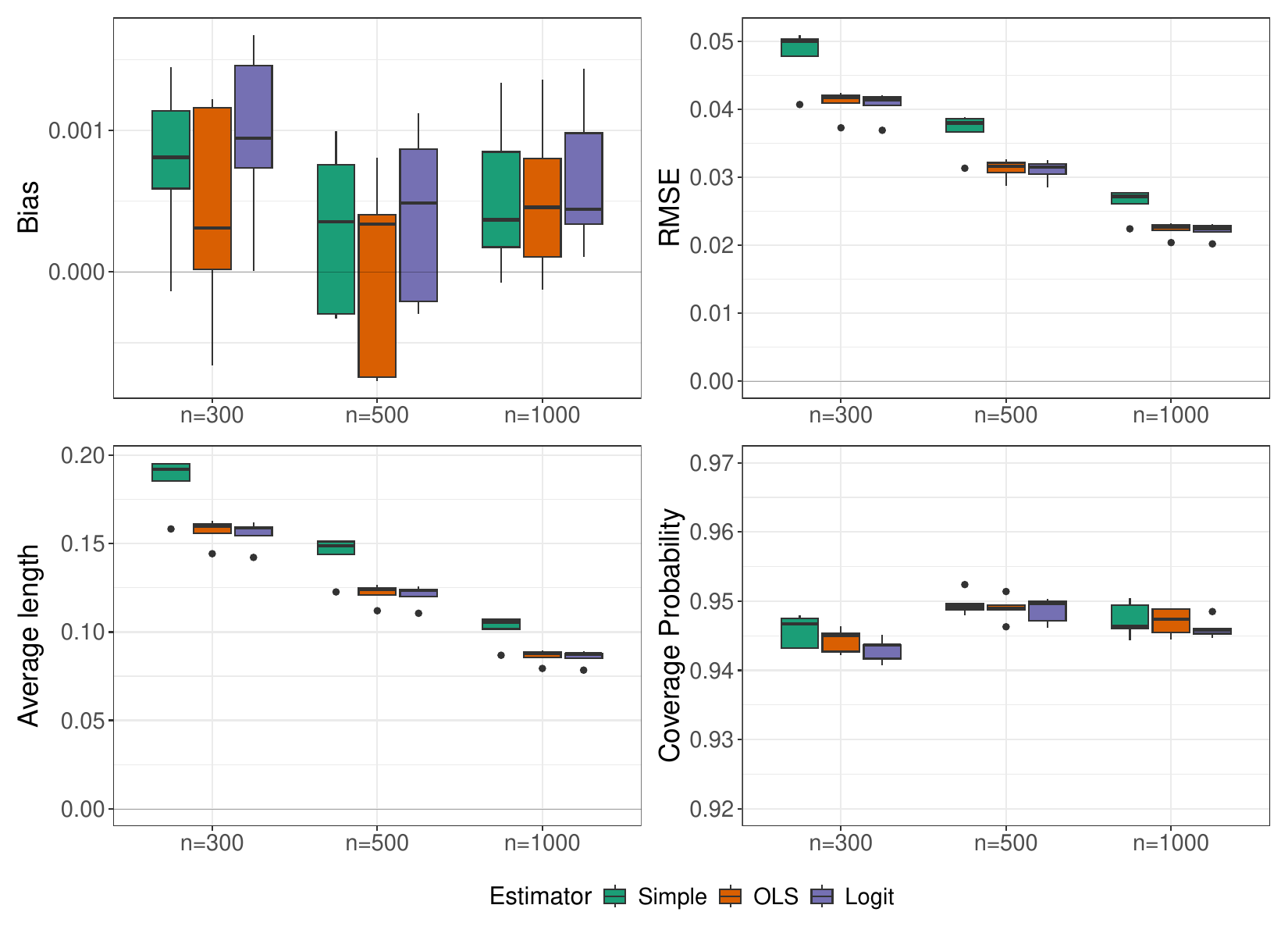}
      \begin{minipage}{0.92\textwidth}
      \small 
      \textit{Note}:
      Bias, RMSE, 95\% CI length and coverage probability calculated over 10,000 simulations. Each boxplot represents the distribution across locations $y$ for a specific sample size.
      \end{minipage}    

\label{fig:dgp-3-rho0.5}
\end{center}
\vskip -0.2in
\end{figure}

We consider three estimators: the simple DTE (Simple) and the DTE with regression adjustment based on the linear regression (OLS) and logit model (Logit). To evaluate the finite sample performance of these DTE estimators, we compute their bias and root mean squared error (RMSE) across 10,000 simulations. Additionally, we assess the average length and coverage probability of their 95\% pointwise confidence intervals, based on analytic standard errors derived from the sample analogs. In the main text, we present results for DGP1 and DGP3 under a treatment assignment probability of $\pi_{1} = 0.5$. Figures displaying results for all other combinations of DGPs and treatment assignment probabilities can be found in Supplemental Material \ref{app:simulation}. Additional results using bootstrap confidence intervals and polynomial transformations of covariates are presented in Sections \ref{app:bootstrap} and \ref{app:cov_transformation}, respectively.

Figure \ref{fig:dgp-1-rho0.5} summarizes the performance metrics for DGP1, while Figure \ref{fig:dgp-3-rho0.5} presents results for DGP3 under a treatment assignment probability of $\pi_{1} = 0.5$. For continuous outcomes in DGP1 and DGP2, we consider thresholds $y$ at quantiles $\{0.1, \dots, 0.9\}$ of the true outcome distribution. For discrete outcomes in DGP3 and DGP4, we consider $y \in\{1, \dots, 5\}$. 
The results are displayed using boxplots over these threshold locations for a concise visual representation. To approximate true distributions, we generate datasets with 1,000,000 observations and compute the DTE at each threshold $y$.

As illustrated in Figure \ref{fig:dgp-1-rho0.5}, both the RMSE and the average length of the confidence intervals decrease as the sample size $n$ increases, as expected. The regression-adjusted estimators improve upon the simple estimator by reducing the RMSE and the length of confidence intervals by approximately 13\%-25\% across all sample sizes. However, for smaller sample sizes ($n \in \{300, 500\}$), the regression-adjusted estimators exhibit a slight undercoverage, with a coverage probability around 0.945. Similar trends are observed with the discrete outcome depicted in Figure \ref{fig:dgp-3-rho0.5}. In this scenario, the regression-adjusted estimators shorten the confidence intervals by about 9\%-20\% compared to the simple estimator.

In summary, these findings suggest that while regression adjustment does not affect the bias in DTE estimation, it significantly enhances precision. Additionally, we confirm that regression adjustment reduces the length of confidence intervals, thereby enhancing precision, while maintaining their 95\% coverage probability.

\section{Empirical Applications} \label{sec:empirical}

\subsection{Nudge and Water Consumption} \label{sec:water}
In this section, we reanalyze data from a randomized experiment conducted by \citet{ferraro2013using} in 2007 to examine the impact of norm-based messages, or nudges, on water usage in Cobb County, Atlanta, Georgia.\footnote{The dataset is publicly available as \cite{DVN1/22633_2013} at \href{{https://doi.org/10.7910/DVN1/22633}}{https://doi.org/10.7910/DVN1/22633}.} The experiment implemented and compared three different interventions aimed at reducing water usage against a control group (no nudge). For more details on the experiment and the results of the average treatment effect analysis, refer to \cite{ferraro2013using}. For simplicity, following \citet{list2022using}, we focus exclusively on the nudge that demonstrated the strongest average effect over the control group, which combined prosocial appeal with social comparison. We estimate the regression-adjusted DTE and PTE of this intervention over the control group, using the monthly water consumption in the year prior as pre-treatment covariates $\bm X$. The outcome variable $Y$ represents the level of water consumption from June to September 2007, measured in thousands of gallons and discretely distributed. It is important to note that, while the measure in gallons appears approximately continuous, the presence of subtle discreteness can complicate both theoretical and practical statistical inference. Consequently, QTE is not applicable in this context.

\begin{figure}[!h]
\vskip 0.2in
\captionsetup{justification=centering, width=0.95\textwidth}
\begin{center}
\caption{Nudge Effect on Water Consumption (thousands of gallons) \\
Distributional and Probability Treatment Effect 
}
\includegraphics[width=0.4\columnwidth]{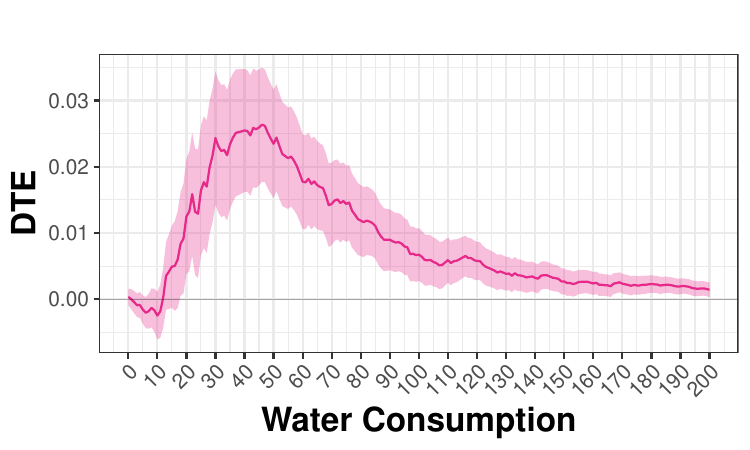}
\includegraphics[width=0.4\columnwidth]{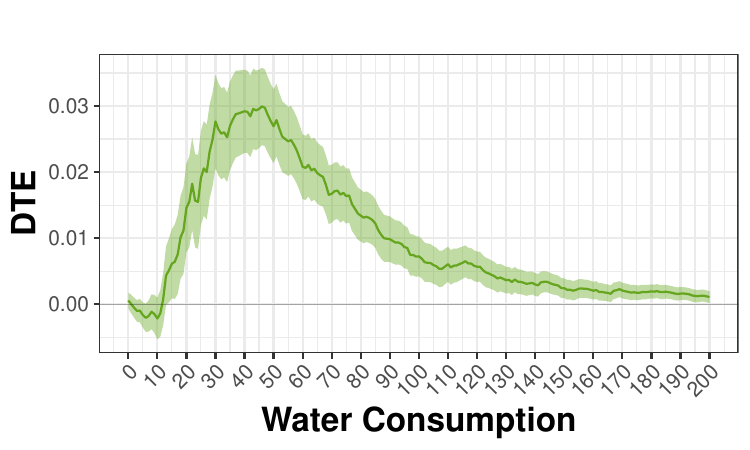}
\includegraphics[width=0.4\columnwidth]{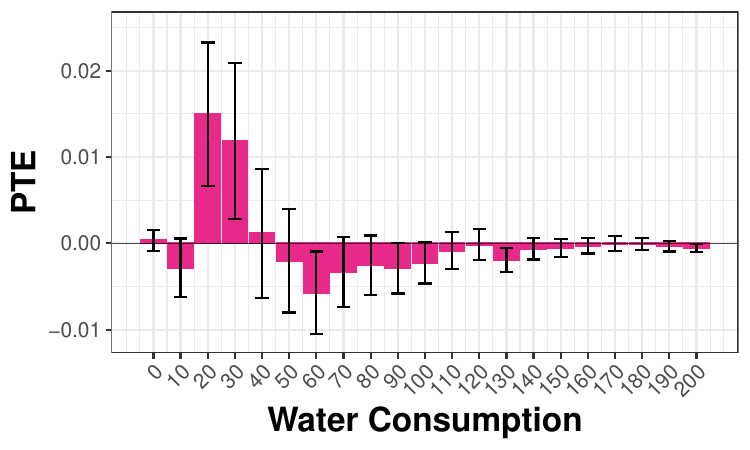}
\includegraphics[width=0.4\columnwidth]{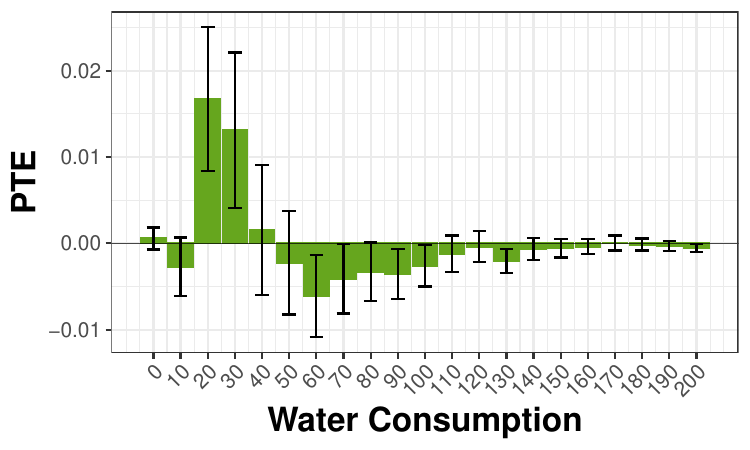}
      \begin{minipage}{0.90\textwidth}
      \small 
      \textit{Note}:
 Top left: simple DTE; top right: regression-adjusted DTE (logit model). Bottom left: simple PTE; bottom right: regression-adjusted PTE (logit model). Shaded areas and error bars: 95\% pointwise confidence intervals. $n=78,500$.
      \end{minipage}    

\label{fig:water-dte}
\end{center}
\vskip -0.2in
\end{figure}

Figure \ref{fig:water-dte} presents the DTE and PTE of the intervention in comparison to the control group. We compute the DTE for $y\in\{0,1,2,\dots, 200\}$. The top left panel of Figure \ref{fig:water-dte} displays the simple estimate of the DTE, whereas the top right panel illustrates the regression-adjusted estimate of the DTE obtained with logit model. The shaded areas denote the 95\% pointwise confidence intervals, based on analytic standard errors derived from the sample analogs. The results indicate that the regression-adjusted DTE has around 4\% to 32\% smaller standard errors and consequently tighter confidence intervals compared to the simple DTE. 

The bottom left panel of Figure \ref{fig:water-dte} depicts the simple PTE, while the bottom right panel shows the regression-adjusted PTE, both with 95\% pointwise confidence intervals. The PTE is presented in increments of $h=10$ for $y\in\{0, 10, 20,\dots, 200\}$. The findings suggest that the treatment is effective in reducing the probability of high water consumption and increasing the probability of low water consumption. Specifically, the regression-adjusted PTE results indicate an probability increase in water usage in the range of (20,30] and (30, 40] by 1.67 percentage points (pp) with standard errors (s.e.) of 0.42 pp and 1.31 pp (s.e.~= 0.46 pp), respectively. On the other hand, it suggests a probability decrease in water usage in the range of (60,70], (70, 80], (90, 100] and (100, 110] by 0.61 pp (s.e.~= 0.24 pp), 0.41 pp (s.e.~= 0.20 pp), 0.35 pp (s.e.~= 0.15 pp), and 0.26 pp (s.e.~= 0.12 pp), respectively. The variance reduction is particularly notable in the range between 70 and 110. The probability change for the range (70,80], (90,100], and (100,110] are not significant under simple estimates, but are significantly negative under regression adjustment.\footnote{Some significance levels are difficult to discern in Figure \ref{fig:water-dte} due to the scale. The probability of decreased water consumption in the range $(80,90]$ is 0.33 pp with a standard error of 0.17 pp under regression-adjusted PTE, despite appearing significantly negative. Additionally, for the ranges $(90,100]$ and $(100,110]$, the probabilities of decreased water consumption are 0.28 pp (s.e.~= 0.15 pp) and 0.22 pp (s.e.~= 0.12 pp), respectively, under simple PTE.}

\subsection{Insurance Coverage and Healthcare Utilization} \label{sec:oregon}
In this subsection, we investigate the impact of insurance coverage on healthcare utilization using data from the Oregon Health Insurance Experiment.\footnote{The dataset is publicly available at \href{https://www.nber.org/research/data/oregon-health-insurance-experiment-data}{https://www.nber.org/research/data/oregon-health-insurance-experiment-data}.} In 2008, the state of Oregon offered insurance coverage to a group of uninsured low-income adults through a lottery. Twelve months after the lottery, a mail survey was conducted to collect data on various outcomes. The treatment in this experiment was randomly assigned based on the number of people in the household. Therefore, we focus on individuals who reported being in single-person households, excluding those who included one or two additional people in their lottery application. Refer to \cite{finkelstein2012oregon} for more details about the experiment and the treatment effect results on various outcomes.

We examine healthcare utilization by focusing on the number of primary care (or doctor) visits over a six-month period, denoted as our outcome variable $Y$, collected from survey data. Our set of covariates $\bm X$ includes the applicant's age, gender, race, and income. Not everyone who was selected by the lottery to apply for insurance coverage ended up obtaining it. Previous research by \cite{finkelstein2012oregon} has documented a positive intention-to-treat (ITT) effect and a local average treatment effect (LATE) of insurance coverage on the number of primary care visits. In this illustrative example, we use the lottery assignment (rather than actual insurance enrollment) as our treatment variable and focus solely on estimating the ITT effect. To obtain the local distributional treatment effect, one can use the representation provided in \citet{abadie2002bootstrap}, applying the ITT effect at each location $y$.

\begin{figure}[!h]
\vskip 0.2in
\captionsetup{justification=centering, width=0.95\textwidth}
\begin{center}
\caption{Effect of Health Insurance Coverage on Doctor Visits \\
Distributional and Probability Treatment Effect
}
\vspace{-0.5cm}
\includegraphics[width=0.4\columnwidth]{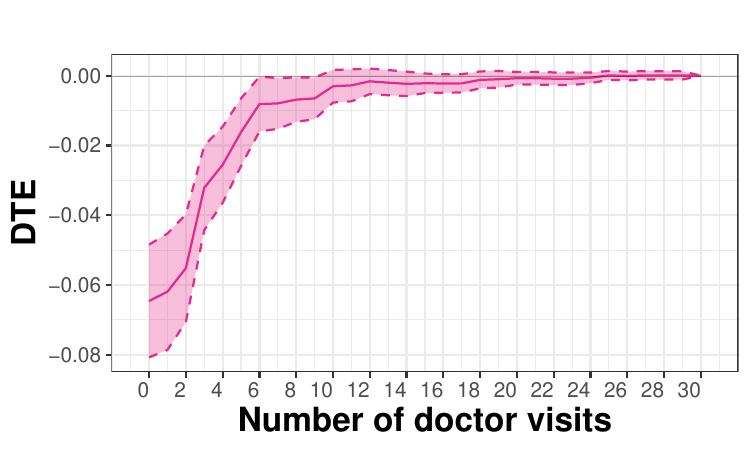}
\includegraphics[width=0.4\columnwidth]{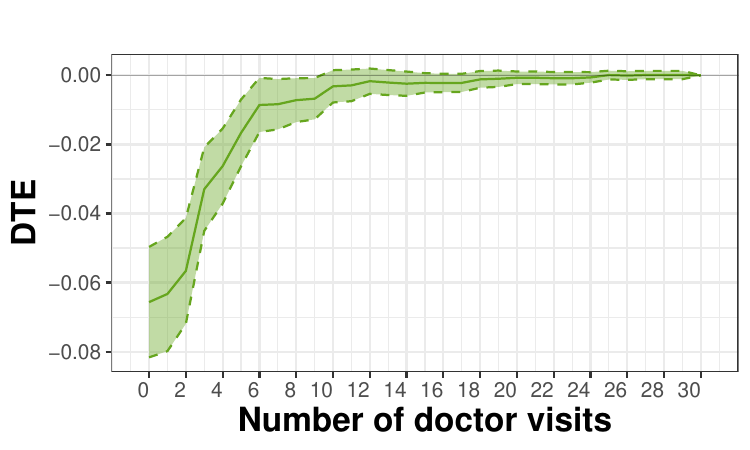}
\includegraphics[width=0.4\columnwidth]{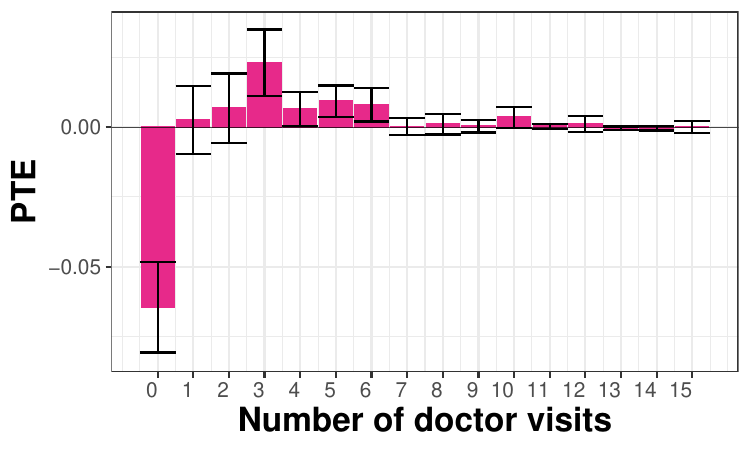}
\includegraphics[width=0.4\columnwidth]{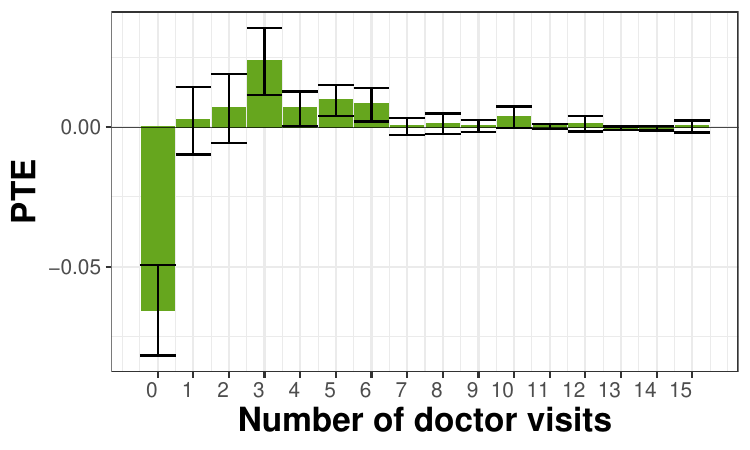}
      \begin{minipage}{0.90\textwidth}
      \small 
      \textit{Note}:
 Top left: simple DTE; top right: regression-adjusted DTE (logit model). Bottom left: simple PTE; bottom right: regression-adjusted PTE (logit model). Shaded areas and error bars: 95\%  pointwise confidence intervals. n=13,759.
      \end{minipage}    
    \label{fig:oregon}
\end{center}
\vspace{-0.7cm}
\end{figure}

Figure \ref{fig:oregon} shows the distributional effect of insurance coverage on the number of primary care visits. The DTE estimates (top left) indicate a statistically significant effect on visits ranging from 0 to 6. The regression-adjusted DTE estimates obtained with logit model are nearly identical to the simple DTE estimates and reduce the standard errors by 0.1\%-1.1\% across values of 
$y{\in}\{0, 1, {\dots}, 30\}$. The gains from regression adjustment are modest in this case.

In the bottom panel of Figure \ref{fig:oregon}, we present the estimated PTE with 95\% confidence intervals for each integer value between 0 and 15 (i.e., $h = 1$). The regression-adjusted PTE estimates suggest that insurance coverage reduces the probability of not seeing a doctor by approximately 6.56 pp (s.e.~= 0.82 pp) and increases the probability of seeing a doctor three, four, five, and six times by 2.36 pp (s.e.~= 0.61 pp), 0.66 pp (s.e.~= 0.31 pp), 0.96 pp (s.e.~=0.28 pp), and 0.81 pp (s.e.~= 0.31 pp), respectively.

\section{Conclusion}\label{sec:con}
In this paper, we introduce a regression adjustment method to improve the precision of estimating distributional treatment effects through the use of distributional regression. We analyze the theoretical properties of our proposed method and provide a practical inference method that has been proven to be theoretically valid. Our simulation results demonstrate the remarkable performance of our method. Additionally, we apply our method to real-world data in two empirical applications and uncover heterogeneous treatment effects. In both applications, we find clearer conclusions compared to existing methods.
  
\clearpage
\setstretch{0.1}
\bibliographystyle{chicago}    
\bibliography{ATE_RA}

\begin{thebibliography}{}

\bibitem[\protect\citeauthoryear{Abadie}{Abadie}{2002}]{abadie2002bootstrap}
Abadie, A. (2002).
\newblock Bootstrap tests for distributional treatment effects in instrumental variable models.
\newblock {\em Journal of the American statistical Association\/}~{\em 97\/}(457), 284--292.

\bibitem[\protect\citeauthoryear{Abadie, Angrist, and Imbens}{Abadie et~al.}{2002}]{abadie2002instrumental}
Abadie, A., J.~Angrist, and G.~Imbens (2002).
\newblock Instrumental variables estimates of the effect of subsidized training on the quantiles of trainee earnings.
\newblock {\em Econometrica\/}~{\em 70\/}(1), 91--117.

\bibitem[\protect\citeauthoryear{Ansel, Powell, and van~der Laan}{Ansel et~al.}{2018}]{Ansel2018}
Ansel, J., J.~Powell, and M.~van~der Laan (2018).
\newblock Covariate balancing propensity score: Implementation and evaluation.
\newblock {\em The International Journal of Biostatistics\/}~{\em 14\/}(1), 59--71.

\bibitem[\protect\citeauthoryear{Athey and Imbens}{Athey and Imbens}{2016}]{Athey2016}
Athey, S. and G.~Imbens (2016).
\newblock Recursive partitioning for heterogeneous causal effects.
\newblock {\em Proceedings of the National Academy of Sciences\/}~{\em 113\/}(27), 7353--7360.

\bibitem[\protect\citeauthoryear{Athey and Imbens}{Athey and Imbens}{2006}]{athey2006identification}
Athey, S. and G.~W. Imbens (2006).
\newblock Identification and inference in nonlinear difference-in-differences models.
\newblock {\em Econometrica\/}~{\em 74\/}(2), 431--497.

\bibitem[\protect\citeauthoryear{Athey and Imbens}{Athey and Imbens}{2017}]{athey2017econometrics}
Athey, S. and G.~W. Imbens (2017).
\newblock The econometrics of randomized experiments.
\newblock In {\em Handbook of Economic Field Experiments}, Volume~1, pp.\  73--140. Elsevier.

\bibitem[\protect\citeauthoryear{Baldassarri and Abascal}{Baldassarri and Abascal}{2017}]{baldassarri2017field}
Baldassarri, D. and M.~Abascal (2017).
\newblock Field experiments across the social sciences.
\newblock {\em Annual Review of Sociology\/}~{\em 43\/}(1), 41--73.

\bibitem[\protect\citeauthoryear{Bartlett}{Bartlett}{2018}]{bartlett2018covariate}
Bartlett, J.~W. (2018).
\newblock Covariate adjustment and estimation of mean response in randomised trials.
\newblock {\em Pharmaceutical Statistics\/}~{\em 17\/}(5), 648--666.

\bibitem[\protect\citeauthoryear{Bitler, Gelbach, and Hoynes}{Bitler et~al.}{2006}]{bitler2006mean}
Bitler, M.~P., J.~B. Gelbach, and H.~W. Hoynes (2006).
\newblock What mean impacts miss: Distributional effects of welfare reform experiments.
\newblock {\em American Economic Review\/}~{\em 96\/}(4), 988--1012.

\bibitem[\protect\citeauthoryear{Callaway and Li}{Callaway and Li}{2019}]{callaway2019quantile}
Callaway, B. and T.~Li (2019).
\newblock Quantile treatment effects in difference in differences models with panel data.
\newblock {\em Quantitative Economics\/}~{\em 10\/}(4), 1579--1618.

\bibitem[\protect\citeauthoryear{Callaway, Li, and Oka}{Callaway et~al.}{2018}]{callaway2018quantile}
Callaway, B., T.~Li, and T.~Oka (2018).
\newblock Quantile treatment effects in difference in differences models under dependence restrictions and with only two time periods.
\newblock {\em Journal of Econometrics\/}~{\em 206\/}(2), 395--413.

\bibitem[\protect\citeauthoryear{Chernozhukov, Fernandez-Val, and Galichon}{Chernozhukov et~al.}{2009}]{chernozhukov2009improving}
Chernozhukov, V., I.~Fernandez-Val, and A.~Galichon (2009).
\newblock Improving point and interval estimators of monotone functions by rearrangement.
\newblock {\em Biometrika\/}~{\em 96\/}(3), 559--575.

\bibitem[\protect\citeauthoryear{Chernozhukov, Fern{\'a}ndez-Val, and Luo}{Chernozhukov et~al.}{2023}]{chernozhukov2023distribution}
Chernozhukov, V., I.~Fern{\'a}ndez-Val, and S.~Luo (2023).
\newblock Distribution regression with sample selection and uk wage decomposition.
\newblock Technical report, cemmap working paper.

\bibitem[\protect\citeauthoryear{Chernozhukov, Fern{\'a}ndez-Val, and Melly}{Chernozhukov et~al.}{2013}]{chernozhukov2013inference}
Chernozhukov, V., I.~Fern{\'a}ndez-Val, and B.~Melly (2013).
\newblock Inference on counterfactual distributions.
\newblock {\em Econometrica\/}~{\em 81\/}(6), 2205--2268.

\bibitem[\protect\citeauthoryear{Chernozhukov, Fernandez-Val, Melly, and W{\"u}thrich}{Chernozhukov et~al.}{2020}]{chernozhukov2019generic}
Chernozhukov, V., I.~Fernandez-Val, B.~Melly, and K.~W{\"u}thrich (2020).
\newblock Generic inference on quantile and quantile effect functions for discrete outcomes.
\newblock {\em Journal of the American Statistical Association\/}.

\bibitem[\protect\citeauthoryear{Chernozhukov, Fernandez-Val, and Weidner}{Chernozhukov et~al.}{2020}]{chernozhukov2020network}
Chernozhukov, V., I.~Fernandez-Val, and M.~Weidner (2020).
\newblock Network and panel quantile effects via distribution regression.
\newblock {\em Journal of econometrics\/}, 105009.

\bibitem[\protect\citeauthoryear{Chernozhukov and Hansen}{Chernozhukov and Hansen}{2005}]{chernozhukov2005iv}
Chernozhukov, V. and C.~Hansen (2005).
\newblock An iv model of quantile treatment effects.
\newblock {\em Econometrica\/}~{\em 73\/}(1), 245--261.

\bibitem[\protect\citeauthoryear{Cochran}{Cochran}{1957}]{cochran1957analysis}
Cochran, W.~G. (1957).
\newblock Analysis of covariance: its nature and uses.
\newblock {\em Biometrics\/}~{\em 13\/}(3), 261--281.

\bibitem[\protect\citeauthoryear{Cox and McCullagh}{Cox and McCullagh}{1982}]{cox1982biometrics}
Cox, D.~R. and P.~McCullagh (1982).
\newblock A biometrics invited paper with discussion. some aspects of analysis of covariance.
\newblock {\em Biometrics\/}, 541--561.

\bibitem[\protect\citeauthoryear{Deng, Xu, Kohavi, and Walker}{Deng et~al.}{2013}]{Deng2013}
Deng, A., Y.~Xu, R.~Kohavi, and T.~Walker (2013).
\newblock Improving the sensitivity of online controlled experiments by utilizing pre-experiment data.
\newblock WSDM '13. Association for Computing Machinery.

\bibitem[\protect\citeauthoryear{Djebbari and Smith}{Djebbari and Smith}{2008}]{djebbari2008heterogeneous}
Djebbari, H. and J.~Smith (2008).
\newblock Heterogeneous impacts in progresa.
\newblock {\em Journal of Econometrics\/}~{\em 145\/}(1-2), 64--80.

\bibitem[\protect\citeauthoryear{Doksum}{Doksum}{1974}]{doksum1974empirical}
Doksum, K. (1974).
\newblock Empirical probability plots and statistical inference for nonlinear models in the two-sample case.
\newblock {\em The Annals of Statistics\/}, 267--277.

\bibitem[\protect\citeauthoryear{Donald and Hsu}{Donald and Hsu}{2014}]{donald2014estimation}
Donald, S.~G. and Y.-C. Hsu (2014).
\newblock Estimation and inference for distribution functions and quantile functions in treatment effect models.
\newblock {\em Journal of Econometrics\/}~{\em 178}, 383--397.

\bibitem[\protect\citeauthoryear{Duflo, Glennerster, and Kremer}{Duflo et~al.}{2007}]{duflo2007using}
Duflo, E., R.~Glennerster, and M.~Kremer (2007).
\newblock Using randomization in development economics research: A toolkit.
\newblock {\em Handbook of Development Economics\/}~{\em 4}, 3895--3962.

\bibitem[\protect\citeauthoryear{Ferraro and Price}{Ferraro and Price}{2013a}]{DVN1/22633_2013}
Ferraro, P. and M.~Price (2013a).
\newblock {Replication data for: Using Nonpecuniary Strategies to Influence Behavior: Evidence from a Large-Scale Field Experiment}.

\bibitem[\protect\citeauthoryear{Ferraro and Price}{Ferraro and Price}{2013b}]{ferraro2013using}
Ferraro, P.~J. and M.~K. Price (2013b).
\newblock Using nonpecuniary strategies to influence behavior: Evidence from a large-scale field experiment.
\newblock {\em Review of Economics and Statistics\/}~{\em 95\/}(1), 64--73.

\bibitem[\protect\citeauthoryear{Finkelstein, Taubman, Wright, Bernstein, Gruber, Newhouse, Allen, Baicker, and Group}{Finkelstein et~al.}{2012}]{finkelstein2012oregon}
Finkelstein, A., S.~Taubman, B.~Wright, M.~Bernstein, J.~Gruber, J.~P. Newhouse, H.~Allen, K.~Baicker, and O.~H.~S. Group (2012).
\newblock The oregon health insurance experiment: evidence from the first year.
\newblock {\em The Quarterly Journal of Economics\/}~{\em 127\/}(3), 1057--1106.

\bibitem[\protect\citeauthoryear{Fisher}{Fisher}{1925}]{fisher1925statistical}
Fisher, R.~A. (1925).
\newblock {\em Statistical methods for research workers}.
\newblock Oliver and Boyd, Edinburgh, Scotland.

\bibitem[\protect\citeauthoryear{Fisher}{Fisher}{1935}]{fisher1935design}
Fisher, R.~A. (1935).
\newblock {\em The Design of Experiments}.
\newblock Edinburgh and London: Oliver and Boyd.

\bibitem[\protect\citeauthoryear{Foresi and Peracchi}{Foresi and Peracchi}{1995}]{foresi1995conditional}
Foresi, S. and F.~Peracchi (1995).
\newblock The conditional distribution of excess returns: An empirical analysis.
\newblock {\em Journal of the American Statistical Association\/}~{\em 90\/}(430), 451--466.

\bibitem[\protect\citeauthoryear{Freedman}{Freedman}{2008a}]{freedman2008regressionA}
Freedman, D.~A. (2008a).
\newblock On regression adjustments in experiments with several treatments.
\newblock {\em The annals of Applied Statistics\/}~{\em 2\/}(1), 176--196.

\bibitem[\protect\citeauthoryear{Freedman}{Freedman}{2008b}]{freedman2008regression}
Freedman, D.~A. (2008b).
\newblock On regression adjustments to experimental data.
\newblock {\em Advances in Applied Mathematics\/}~{\em 40\/}(2), 180--193.

\bibitem[\protect\citeauthoryear{Frison and Pocock}{Frison and Pocock}{1992}]{frison1992repeated}
Frison, L. and S.~J. Pocock (1992).
\newblock Repeated measures in clinical trials: analysis using mean summary statistics and its implications for design.
\newblock {\em Statistics in medicine\/}~{\em 11\/}(13), 1685--1704.

\bibitem[\protect\citeauthoryear{Hahn}{Hahn}{1998}]{hahn1998role}
Hahn, J. (1998).
\newblock On the role of the propensity score in efficient semiparametric estimation of average treatment effects.
\newblock {\em Econometrica\/}, 315--331.

\bibitem[\protect\citeauthoryear{Hall, Wolff, and Yao}{Hall et~al.}{1999}]{hall1999methods}
Hall, P., R.~C. Wolff, and Q.~Yao (1999).
\newblock Methods for estimating a conditional distribution function.
\newblock {\em Journal of the American Statistical Association\/}~{\em 94\/}(445), 154--163.

\bibitem[\protect\citeauthoryear{Heckman, Smith, and Clements}{Heckman et~al.}{1997}]{heckman1997making}
Heckman, J.~J., J.~Smith, and N.~Clements (1997).
\newblock Making the most out of programme evaluations and social experiments: Accounting for heterogeneity in programme impacts.
\newblock {\em The Review of Economic Studies\/}~{\em 64\/}(4), 487--535.

\bibitem[\protect\citeauthoryear{Horiuchi, Imai, and Taniguchi}{Horiuchi et~al.}{2007}]{horiuchi2007designing}
Horiuchi, Y., K.~Imai, and N.~Taniguchi (2007).
\newblock Designing and analyzing randomized experiments: Application to a japanese election survey experiment.
\newblock {\em American Journal of Political Science\/}~{\em 51\/}(3), 669--687.

\bibitem[\protect\citeauthoryear{Imai}{Imai}{2005}]{imai2005get}
Imai, K. (2005).
\newblock Do get-out-the-vote calls reduce turnout? the importance of statistical methods for field experiments.
\newblock {\em American Political Science Review\/}~{\em 99\/}(2), 283--300.

\bibitem[\protect\citeauthoryear{Imai and Ratkovic}{Imai and Ratkovic}{2013}]{Imai2013}
Imai, K. and M.~Ratkovic (2013).
\newblock {Estimating treatment effect heterogeneity in randomized program evaluation}.
\newblock {\em The Annals of Applied Statistics\/}~{\em 7\/}(1), 443 -- 470.

\bibitem[\protect\citeauthoryear{Imbens and Rubin}{Imbens and Rubin}{2015}]{imbens2015causal}
Imbens, G.~W. and D.~B. Rubin (2015).
\newblock {\em Causal inference for statistics, social, and biomedical sciences: An introduction}.
\newblock Taylor \& Francis.

\bibitem[\protect\citeauthoryear{Jiang, Phillips, Tao, and Zhang}{Jiang et~al.}{2023}]{jiang2023regression}
Jiang, L., P.~C. Phillips, Y.~Tao, and Y.~Zhang (2023).
\newblock Regression-adjusted estimation of quantile treatment effects under covariate-adaptive randomizations.
\newblock {\em Journal of Econometrics\/}~{\em 234\/}(2), 758--776.

\bibitem[\protect\citeauthoryear{Kallus, Mao, and Uehara}{Kallus et~al.}{2024}]{kallus2024localized}
Kallus, N., X.~Mao, and M.~Uehara (2024).
\newblock Localized debiased machine learning: Efficient inference on quantile treatment effects and beyond.
\newblock {\em Journal of Machine Learning Research\/}~{\em 25\/}(16), 1--59.

\bibitem[\protect\citeauthoryear{Kato}{Kato}{2009}]{kato2009asymptotics}
Kato, K. (2009).
\newblock Asymptotics for argmin processes: Convexity arguments.
\newblock {\em Journal of Multivariate Analysis\/}~{\em 100\/}(8), 1816--1829.

\bibitem[\protect\citeauthoryear{Künzel, Sekhon, Bickel, and Yu}{Künzel et~al.}{2019}]{Metalearner2019}
Künzel, S.~R., J.~S. Sekhon, P.~J. Bickel, and B.~Yu (2019).
\newblock Metalearners for estimating heterogeneous treatment effects using machine learning.
\newblock {\em Proceedings of the National Academy of Sciences\/}~{\em 116\/}(10), 4156--4165.

\bibitem[\protect\citeauthoryear{Lehmann}{Lehmann}{1975}]{lehmann1975nonparametrics}
Lehmann, E.~L. (1975).
\newblock {\em Nonparametrics: statistical methods based on ranks.}
\newblock Holden-day.

\bibitem[\protect\citeauthoryear{Lewis and Rao}{Lewis and Rao}{2015}]{lewis2015unfavorable}
Lewis, R.~A. and J.~M. Rao (2015).
\newblock The unfavorable economics of measuring the returns to advertising.
\newblock {\em The Quarterly Journal of Economics\/}~{\em 130\/}(4), 1941--1973.

\bibitem[\protect\citeauthoryear{Lin}{Lin}{2013}]{lin2013agnostic}
Lin, W. (2013).
\newblock Agnostic notes on regression adjustments to experimental data: Reexamining freedman’s critique.
\newblock {\em The Annals of Applied Statistics\/}~{\em 7\/}(1), 295--318.

\bibitem[\protect\citeauthoryear{List, Muir, and Sun}{List et~al.}{2022}]{list2022using}
List, J.~A., I.~Muir, and G.~K. Sun (2022).
\newblock Using machine learning for efficient flexible regression adjustment in economic experiments.
\newblock Technical report, National Bureau of Economic Research.

\bibitem[\protect\citeauthoryear{McCullagh and Nelder}{McCullagh and Nelder}{1989}]{mccullagh1989binary}
McCullagh, P. and J.~Nelder (1989).
\newblock Binary data.
\newblock In {\em Generalized linear models}, pp.\  98--148. Springer.

\bibitem[\protect\citeauthoryear{Negi and Wooldridge}{Negi and Wooldridge}{2020}]{negi2020robust}
Negi, A. and J.~M. Wooldridge (2020).
\newblock Robust and efficient estimation of potential outcome means under random assignment.
\newblock {\em arXiv preprint arXiv:2010.01800\/}.

\bibitem[\protect\citeauthoryear{Negi and Wooldridge}{Negi and Wooldridge}{2021}]{negi2021revisiting}
Negi, A. and J.~M. Wooldridge (2021).
\newblock Revisiting regression adjustment in experiments with heterogeneous treatment effects.
\newblock {\em Econometric Reviews\/}~{\em 40\/}(5), 504--534.

\bibitem[\protect\citeauthoryear{Newey}{Newey}{1994}]{newey1994asymptotic}
Newey, W.~K. (1994).
\newblock The asymptotic variance of semiparametric estimators.
\newblock {\em Econometrica: Journal of the Econometric Society\/}, 1349--1382.

\bibitem[\protect\citeauthoryear{Neyman}{Neyman}{1923}]{Neyman1923}
Neyman, J. (1923).
\newblock On the application of probability theory to agricultural experiments. essay on principles. section 9. trans. dorota m. dabrowska and terence p. speed. 1990.
\newblock {\em Statistical Science\/}, 465--472.

\bibitem[\protect\citeauthoryear{Nie and Wager}{Nie and Wager}{2020}]{Nie2020}
Nie, X. and S.~Wager (2020, 09).
\newblock {Quasi-oracle estimation of heterogeneous treatment effects}.
\newblock {\em Biometrika\/}~{\em 108\/}(2), 299--319.

\bibitem[\protect\citeauthoryear{Pollard}{Pollard}{1991}]{pollard1991asymptotics}
Pollard, D. (1991).
\newblock Asymptotics for least absolute deviation regression estimators.
\newblock {\em Econometric Theory\/}~{\em 7\/}(2), 186--199.

\bibitem[\protect\citeauthoryear{Praestgaard and Wellner}{Praestgaard and Wellner}{1993}]{praestgaard1993Exchangeably}
Praestgaard, J. and J.~A. Wellner (1993).
\newblock Exchangeably weighted bootstraps of the general empirical process.
\newblock {\em Annals of Probability\/}~{\em 21}, 2053--2086.

\bibitem[\protect\citeauthoryear{Rosenblum and Van Der~Laan}{Rosenblum and Van Der~Laan}{2010}]{rosenblum2010simple}
Rosenblum, M. and M.~J. Van Der~Laan (2010).
\newblock Simple, efficient estimators of treatment effects in randomized trials using generalized linear models to leverage baseline variables.
\newblock {\em The International Journal of Biostatistics\/}~{\em 6\/}(1).

\bibitem[\protect\citeauthoryear{Rubin}{Rubin}{1974}]{rubin1974estimating}
Rubin, D.~B. (1974).
\newblock Estimating causal effects of treatments in randomized and nonrandomized studies.
\newblock {\em Journal of Educational Psychology\/}~{\em 66\/}(5), 688.

\bibitem[\protect\citeauthoryear{Rubin}{Rubin}{1997}]{rubin1997estimating}
Rubin, D.~B. (1997).
\newblock Estimating causal effects from large data sets using propensity scores.
\newblock {\em Annals of internal medicine\/}~{\em 127\/}(8\_Part\_2), 757--763.

\bibitem[\protect\citeauthoryear{Tsiatis, Davidian, and Gneyou}{Tsiatis et~al.}{2008}]{Tsiatis2008}
Tsiatis, A., M.~Davidian, and M.~Gneyou (2008).
\newblock Semiparametric estimation in censored survival models with time-dependent covariates.
\newblock {\em Biometrika\/}~{\em 95\/}(3), 611--623.

\bibitem[\protect\citeauthoryear{van~der Vaart and Wellner}{van~der Vaart and Wellner}{1996}]{van1996weak}
van~der Vaart, A. and J.~Wellner (1996).
\newblock {\em Weak convergence and empirical processes: with applications to statistics}.
\newblock Springer Science \& Business Media.

\bibitem[\protect\citeauthoryear{van~der Vaart}{van~der Vaart}{2000}]{van2000asymptotic}
van~der Vaart, A.~W. (2000).
\newblock {\em Asymptotic statistics}, Volume~3.
\newblock Cambridge university press.

\bibitem[\protect\citeauthoryear{Wager and Athey}{Wager and Athey}{2018}]{Wager2018}
Wager, S. and S.~Athey (2018).
\newblock Estimation and inference of heterogeneous treatment effects using random forests.
\newblock {\em Journal of the American Statistical Association\/}~{\em 113\/}(523), 1228--1242.

\bibitem[\protect\citeauthoryear{Wang, Oka, and Zhu}{Wang et~al.}{2023}]{wang2023bivariate}
Wang, Y., T.~Oka, and D.~Zhu (2023).
\newblock Bivariate distribution regression with application to insurance data.
\newblock {\em Insurance: Mathematics and Economics\/}~{\em 113}, 215--232.

\bibitem[\protect\citeauthoryear{Wang, Oka, and Zhu}{Wang et~al.}{2024}]{wang2023distributional}
Wang, Y., T.~Oka, and D.~Zhu (2024).
\newblock Distributional vector autoregression: Eliciting macro and financial dependence.
\newblock {\em arXiv preprint arXiv:2303.04994\/}.

\bibitem[\protect\citeauthoryear{White}{White}{1982}]{White1982}
White, H. (1982).
\newblock {Maximum likelihood estimation of misspecified models}.
\newblock {\em Econometrica\/}~{\em 50\/}(1), 1--25.

\bibitem[\protect\citeauthoryear{Williams and Grizzle}{Williams and Grizzle}{1972}]{williams1972analysis}
Williams, O.~D. and J.~E. Grizzle (1972).
\newblock Analysis of contingency tables having ordered response categories.
\newblock {\em Journal of the American Statistical Association\/}~{\em 67\/}(337), 55--63.

\bibitem[\protect\citeauthoryear{Yang and Tsiatis}{Yang and Tsiatis}{2001}]{Yang2001}
Yang, Y. and A.~Tsiatis (2001).
\newblock Semiparametric estimation of causal effects of treatments in randomized studies with noncompliance.
\newblock {\em Biometrika\/}~{\em 88\/}(3), 809--819.

\end{thebibliography}

\clearpage
\appendix

\setcounter{section}{0} 
\setcounter{equation}{0} 
\setcounter{lemma}{0}\setcounter{page}{1}
\setcounter{proposition}{0} %
\renewcommand{\thepage}{S-\arabic{page}}
\renewcommand{\theequation}{A.%
  \arabic{equation}}\renewcommand{\thelemma}{A.\arabic{lemma}} 
\renewcommand{\theproposition}{A.\arabic{proposition}}

\setstretch{1.3}

\begin{center}
  \section*{Supplementary Material}  
\end{center}

In this supplementary material, we complete the proofs in Section A, and provide additional results from our simulation experiment in Section B. The relevant R code is available online.

\section{Theoretical Results} \label{app:theory}

\subsection{Proof of  Theorem \ref{theorem:distribution}}

To ensure the completeness of the proof, we begin by presenting a variant of
Lagrange's identity and Bergstr\"{o}m's inequality in the below lemma, which is useful to prove the efficiency gain of the regression adjustment. 

\vspace{0.2cm}
\begin{lemma}
    \label{lemma:L-ind}
    For any 
    $(a_{1}, \dots, a_{K}) \in \R^{K}$
    and 
    $(b_{1}, \dots, b_{K}) \in \R^{K}$
    with $b_{k} > 0 $ for all $k=1, \dots, K$,
    we can show that 
    \begin{align*}
     \sum_{k=1}^{K} \frac{a_{k}^{2}}{b_{k}}   
     - 
     \frac{
     \big (
     \sum_{k=1}^{K} a_{k}
     \big )^2 
    }{
     \sum_{k=1}^{K} b_{k}
    }
    = 
     \frac{
     1
    }{
     \sum_{k=1}^{K} b_{k}
    }
    \cdot 
    \frac{1}{2}
    \sum_{k=1}^{K}
    \sum_{\substack{\ell =1 \\ \ell \neq k}}^K 
    \frac{
    (a_{k} b_{\ell} 
    - 
    a_{\ell} b_{k}
    )^2      
    }{
    b_{k} b_{\ell}
    } ,
    \end{align*}
    which implies 
    Bergstr\"{o}m's inequality, given by 
  \begin{align*}
     \sum_{k=1}^{K} \frac{a_{k}^{2}}{b_{k}}   
     \geq 
     \frac{
     \big (
     \sum_{k=1}^{K} a_{k}
     \big )^2 
    }{ 
     \sum_{k=1}^{K} b_{k}
    }
    . 
    \end{align*}

\end{lemma}
\begin{proof}
    Lagrange's identity is that, 
    for any 
    $(c_{1}, \dots, c_{K}) \in \R^{K}$
    and 
    $(d_{1}, \dots, d_{K}) \in \R^{K}$, 
    \begin{align}
    \label{eq:L-ind}
        \bigg (\sum_{k=1}^{K} c_{k}^2
        \bigg)
        \bigg(\sum_{k=1}^{K} d_{k}^2
        \bigg)
        - 
        \bigg(\sum_{k=1}^{K} c_{k}d_{k}
        \bigg)^2
        =
    \frac{1}{2}
    \sum_{k=1}^{K}
    \sum_{\substack{\ell =1 \\ \ell \neq k}}^K 
     (c_{k} d_{\ell} 
    - 
    c_{\ell} d_{k}
    )^2   .
    \end{align}
    Fix arbitrary 
    $(a_{1}, \dots, a_{K}) \in \R^{K}$
    and 
    $(b_{1}, \dots, b_{K}) \in \R^{K}$
    with $b_{k} > 0 $ for all $k=1, \dots, K$. 
    Then, 
    taking 
    $c_{k} = a_{k} / \sqrt{b_{k}}$ 
    and 
    $d_{k} = \sqrt{b_{k}}$
    for all $k=1, \dots, K$
    in (\ref{eq:L-ind}), we can show that 
    \begin{align*}
        \bigg (
        \sum_{k=1}^{K} 
        \frac{a_{k}^2}{b_{k}}
        \bigg)
        \bigg(\sum_{k=1}^{K} b_{k}
        \bigg)
        - 
        \bigg(\sum_{k=1}^{K} a_{k}
        \bigg)^2
     &   =
    \frac{1}{2}
    \sum_{k=1}^{K}
    \sum_{\substack{\ell =1 \\ \ell \neq k}}^K 
     \bigg (
     \frac{a_{k}}{ \sqrt{b_{k}}}
     \sqrt{b_{\ell}} 
     - 
     \frac{a_{\ell}}{ \sqrt{b_{\ell}}}
     \sqrt{ b_{k}} 
     \bigg )^2   \\ 
     & = 
         \frac{1}{2}
    \sum_{k=1}^{K}
    \sum_{\substack{\ell =1 \\ \ell \neq k}}^K 
    \frac{
    (a_{k} b_{\ell} 
    - 
    a_{\ell} b_{k}
    )^2      
    }{
    b_{k} b_{\ell}
    },
    \end{align*}
    which leads to the desired equality. 
    Also, the last expression in the math display above 
    is non-negative, which leads to Bergstr\"{o}m's inequality. 
\end{proof}
\vspace{0.5cm}

\begin{proof}[\textbf{Proof of Theorem \ref{theorem:distribution}}]
\textbf{Part (a)}
  Choose any arbitrary $k \in \K$.  
  The quadratic expansion of  
  $
  \widetilde{\F}_{Y(k)}
  = 
  \widehat{\F}_{Y(k)}^{simple}
  -
  (\widehat{\P}_{X}^{(k)} - \widehat{\P}_{X})  
  G_{Y(k)|\bm{X}} 
  $ yields 
  \begin{align}
    & \Var
    \big(
      \widetilde{\F}_{Y(k)}
    \big)
    =
    \Var
    \big(
    \widehat{\F}_{Y(k)}^{simple}
    \big)  
    -
    2
    \Cov
    \Big (
    \widehat{\F}_{Y(k)}^{simple},
    (\widehat{\P}_{\bm{X}}^{(k)} - \widehat{\P}_{\bm{X}}) 
    G_{Y(k)|\bm{X}}
    \Big ) \notag \\ 
    & \hspace{3cm} +
    \Var
    \Big (
    (\widehat{\P}_{\bm{X}}^{(k)} - \widehat{\P}_{\bm{X}}) 
    G_{Y(k)|\bm{X}}
    \Big ).     \label{eq:var-1}
  \end{align}
  We can write 
  $\widehat{\P}_{\bm{X}} = \sum_{k \in \K} 
  \hat{\pi}_{k} \widehat{\P}_{\bm{X}}^{(k)}
  $. 
  It is assumed that 
  observations are a random sample and $n_{k} /n = \pi_{k} + o(1)$ for every $k \in \K$
  as $n \to \infty$. 
  Furthermore, all unconditional and conditional functions are 
  bounded. 
  By applying the dominated convergence theorem, we can show 
  \begin{align}
    n \Cov
    \Big (
    \widehat{\F}_{Y(k)}^{simple},
    (\widehat{\P}_{\bm{X}}^{(k)} - \widehat{\P}_{\bm{X}}) 
    G_{Y(k)|\bm{X}}
    \Big ) \notag
    &
      = n
    \Cov    \Big (
    \widehat{\F}_{Y(k)}^{simple},
    (1 - \hat{\pi}_{k})\widehat{\P}_{\bm{X}}^{(k)} 
    G_{Y(k)|\bm{X}}
    \Big ) \notag \\
    &
      = 
    \frac{1 - \pi_{k}}{ \pi_{k}}
    \Cov
    \big (
      \1_{ \{Y(k) \leq \cdot \} },
      G_{Y(k)|\bm{X}}
    \big ) + o(1). 
    \label{eq:cov-1}
  \end{align}
   Similarly, we can show that  
  \begin{align}
    n \Var
    \big (
    (\widehat{\P}_{\bm{X}}^{(k)} - \widehat{\P}_{\bm{X}}) 
    G_{Y(k)|\bm{X}}
    \big ) \notag 
    & 
      =
    n \Var
    \big (
    (1 -   \hat{\pi}_{k})
    \widehat{\P}_{\bm{X}}^{(k)} 
    G_{Y(k)|\bm{X}}
    \big )
    + 
    n
    \sum_{\ell : \ell \not = k }
     \Var
    \big (
    \hat{\pi}_{\ell}
    \widehat{\P}_{\bm{X}}^{(\ell)} 
    G_{Y(k)|\bm{X}}
    \big )
 \notag \\ 
    & 
      =
 \frac{  (1 - \pi_{w})^2}{ \pi_{w}} 
    \Var
    \big (
    G_{Y(k)|\bm{X}}
    \big )
    +
      \sum_{\ell : \ell \not = k }
      \frac{    \pi_{\ell}^{2}}{\pi_{\ell}}
    \Var
    \big (
    G_{Y(k)|\bm{X}}
    \big ) + o(1)
    \notag \\ 
    & 
      =
      \bigg (
           \frac{  (1 -   \pi_{k})^2}{ \pi_{k}} 
        +
      \sum_{\ell : \ell \not = k }
      \pi_{\ell} 
      \bigg )
    \Var
    \big (
    G_{Y(k)|\bm{X}}
    \big ) + o(1)\notag \\ 
    & 
  =
  \frac{1-\pi_{k}}{\pi_{k}}
    \Var
    \big (
    G_{Y(k)|\bm{X}}
    \big ) + o(1). 
    \label{eq:cov-2}
  \end{align}
  It follows from 
  (\ref{eq:var-1})-(\ref{eq:cov-2}) that
  \begin{align}
    \label{eq:bb2}
    &
    n
    \big \{
    \Var \big(
    \widehat{\F}_{Y(k)}^{simple}
    \big)
    -
    \Var \big(
      \widetilde{\F}_{Y(k)}
    \big)
      \big \} \notag \\ 
    &\ \ \ \ \ \ = 
    \frac{1-\pi_{k}}{\pi_{k}}
    \big \{
     2
    \Cov
    \big (
      \1_{ \{Y(k) \leq \cdot \} },
      G_{Y(k)|\bm{X}}
    \big )
    -
    \Var
    \big (
    G_{Y(k)|\bm{X}}  
    \big )
    \big \} + o(1).
  \end{align}
  Applying addition and subtraction, we obtain 
  \begin{align*}
     2
    \Cov
    \big (
      \1_{ \{Y(k) \leq \cdot \} },
      G_{Y(k)|\bm{X}}
    \big )
    -
    \Var
    \big (
    G_{Y(k)|\bm{X}}  
    \big )
    & =
    \Var \big (     \1_{ \{Y(k) \leq \cdot \} }     \big )
    - 
    \Var \big (     \1_{ \{Y(k) \leq \cdot \} }  - G_{Y(k)|\bm{X}}    \big ),
  \end{align*}
  which together with (\ref{eq:bb2}) yields  
   \begin{align*}
    n
    \big \{
    \Var \big(
    \widehat{\F}_{Y(k)}^{simple}
    \big)
    -
    \Var \big(
      \widetilde{\F}_{Y(k)}
    \big) 
      \big \}
    &= 
    \frac{1-\pi_{k}}{\pi_{k}}
    \Big \{
    \Var
      \big (
      \1_{ \{Y(k) \le \cdot\} }       
      \big ) 
      -
      \Var \big (     \1_{ \{Y(k) \leq \cdot \} }  - G_{Y(k)|\bm{X}}    \big )
    \Big \}
      + o(1). 
  \end{align*}
  Under the assumption that GLMs use a canonical link function,  
  the conditional distribution is 
  an unbiased estimator of the unconditional distribution
  or 
  $\E[G_{Y(k)|\bm{X}}] = F_{Y(k)}$,
  as explained in Section \ref{sec:glm-canonical}. 
  Thus we can show that 
  $
    \Var
    \big (
      \1_{ \{Y(k) \le \cdot\} }      
      -
      G_{Y(k)|\bm{X}}
    \big )
    =
    \E
    \big [
    \big (
      \1_{ \{Y(k) \le \cdot\} }      
      -
      G_{Y(k)|\bm{X}}
    \big )^2
    \big ]
  $. 
  Also, we have that 
  $    \Var
      \big (
      \1_{ \{Y(k) \le \cdot\} }       
      \big ) 
      = 
    \E
    \big [
    \big (
      \1_{ \{Y(k) \le \cdot\} }      
      -
      F_{Y(k)}
    \big )^2
    \big ]
  $. 
  Combining these results with the mathematical display above leads to that 
    \begin{align*}
    \Var \big(
    \widehat{\F}_{Y(k)}^{simple}
    \big)
    \ge 
    \Var \big(
      \widetilde{\F}_{Y(k)}
    \big) 
       + o(n^{-1}),  
    \end{align*}
  provided that 
  $
    \E
    \big [
    \big (
      \1_{ \{Y(k) \le \cdot\} }      
      -
      F_{Y(k)}
    \big )^2
    \big ]
    \ge  
    \E
    \big [
    \big (
      \1_{ \{Y(k) \le \cdot\} }      
      -
      G_{Y(k)|\bm{X}}
    \big )^2
    \big ]
  $. 
  
  \vspace{0.5cm}
   \noindent
  \textbf{Part (b)}
    Under the correct specification, 
   an application of the law of iterated expectation yields that 
   $
      \Cov
    \big (
      \1_{ \{Y(k) \leq \cdot \} },
      G_{Y(k)|\bm{X}}
    \big )
    =
    \Var \big (    \E[  \1_{ \{Y(k) \leq \cdot \} }|\bm{X}]     \big ),
  $ 
  which together with (\ref{eq:bb2}) shows  
   \begin{align*}
    n
    \big \{
    \Var \big(
    \widehat{\F}_{Y(k)}^{simple}
    \big)
    -
    \Var \big(
      \widetilde{\F}_{Y(k)}
    \big) 
      \big \}
    &= 
    \frac{1-\pi_{k}}{\pi_{k}}
    \Var
      \big (
      G_{Y(k)|\bm{X}}       
      \big ) + o(1). 
  \end{align*}
  Since $\pi_{k} \in (0, 1)$ and
  $    \Var
      \big (
      G_{Y(k)|\bm{X}}      
      \big )  \geq 0$, 
  it follows that 
  $
    \Var \big(
    \widehat{\F}_{Y(k)}^{simple}
    \big)
    \geq 
      \Var \big(
      \widetilde{\F}_{Y(k)}
    \big) + o(n^{-1}).
  $

    Now, we shall show that, for any $k, \ell \in \K$, 
    \begin{align}
    \label{eq:cov-A}
        n
        \Cov \big(
        \widehat{\F}_{Y(k)}, \widetilde{\F}_{Y(\ell)} 
        \big)
        =
        \Cov \big(
        G_{Y(k)|\bm{X}},
        G_{Y(\ell)|\bm{X}}
        \big). 
    \end{align}
    Fix any two distinct treatment statuses $k, \ell \in \K$. 
    We can write 
  $
  \widetilde{\F}_{Y(k)}
  = 
  \big (\widehat{\F}_{Y(k)}^{simple}
  - \widehat{\P}_{\bm{X}}^{(k)} 
  G_{Y(k)|\bm{X}}
  \big ) 
  +
  \widehat{\P}_{\bm{X}} 
  G_{Y(k)|\bm{X}} 
  $
  and also  
  $\widehat{\P}_{\bm{X}}  G_{Y(k)|\bm{X}} = \sum_{\ell \in \K} 
  \hat{\pi}_{\ell} \widehat{\P}_{\bm{X}}^{(\ell)}  
  G_{Y(k)|\bm{X}}
  $.   
  Given random sample and the bi-linear property of the covariance function, 
  we can show that 
  \begin{align*}
        \Cov \Big(
        \widetilde{\F}_{Y(k)}, 
        \widetilde{\F}_{Y(\ell)} 
        \Big)
        =&
        \Cov
        \Big(
        \widehat{\F}_{Y(k)}^{simple}
        - \widehat{\P}_{\bm{X}}^{(k)} 
        G_{Y(k)|\bm{X}},
        \hat{\pi}_{\ell}\widehat{\P}_{\bm{X}}^{(\ell)}  
        G_{Y(\ell)|\bm{X}}
        \Big ) 
        \\ 
        & + \Cov
        \big(
        \hat{\pi}_{k}\widehat{\P}_{\bm{X}}^{(k)}  
        G_{Y(k)|\bm{X}}, 
        \widehat{\F}_{Y(\ell)}^{simple}
        - \widehat{\P}_{\bm{X}}^{(\ell)} 
        G_{Y(\ell)|\bm{X}}
        \big ) \\
        & + \Cov
        \big(
        \widehat{\P}_{\bm{X}}  G_{Y(k)|\bm{X}}, 
        \widehat{\P}_{\bm{X}}  G_{Y(\ell)|\bm{X}}
        \big ) 
    , 
    \end{align*}
    where it can be shown that the first and second terms on the right-hand side are equal zero, due to the fact that
    $     \E \big [   \widehat{\F}_{Y(k)}^{simple}(y)
        - \widehat{\P}_{\bm{X}}^{(k)}
        G_{Y(k)|\bm{X}}|\bm{X}_{1}, \dots, \bm{X}_{n}
        \big ] = 0
    $. 
    Furthermore, under the random sample assumption, we have that 
    $ \Cov
        \big(
        \widehat{\P}_{\bm{X}}  G_{Y(k)|\bm{X}}, 
        \widehat{\P}_{\bm{X}}  G_{Y(\ell)|\bm{X}}
        \big ) =
        n^{-1}
      \Cov
        \big(
         G_{Y(k)|\bm{X}}, 
         G_{Y(\ell)|\bm{X}}
        \big )   $. 
        Thus, we can prove the equality in (\ref{eq:cov-A}).

    Next, we compare the variance-covariance matrices of the simple and regression-adjusted estimators.
    Define 
    $\sigma_{k, \ell}:= 
     \Cov
        \big(
         G_{Y(k)|\bm{X}}, 
         G_{Y(\ell)|\bm{X}}
     \big ) 
    $
    for $k, \ell \in \K$. 
    By applying the result from part (a) of this theorem and the one in (\ref{eq:cov-A}), we are able to show that
    \begin{align*}
     n 
    \big \{
     \Var \big (
    \widehat{\gamma}^{simple}
    \big ) 
    -
    \Var \big (
    \widetilde{\gamma}
    \big )  
    \big \} 
    & = 
    \left [
    \begin{array}{cccc}
     \frac{1-\pi_{1}}{\pi_{1}}
    \sigma_{11},
    & 
    - \sigma_{12},   
        & 
        \dots,  
        & 
        -\sigma_{1K}  \\
        -\sigma_{21},   
        &     
        \frac{1-\pi_{2}}{\pi_{2}}
    \sigma_{22}, 
        & \dots,  
        & 
        -\sigma_{2K} 
        \\ 
          \vdots &\vdots&\ddots & \vdots\\ 
        -\sigma_{K1},   
        &     
        -\sigma_{K2}, 
        & \dots,  
        & 
        \frac{1-\pi_{K}}{\pi_{K}}
    \sigma_{KK}
    \end{array}
    \right ]   + o(1) \\
    & = 
    \E 
    \Big [
     \big (\gamma(\bm{X}) - \E[\gamma(\bm{X})]\big) 
     A
     \big(\gamma(\bm{X}) - \E[\gamma(\bm{X})]\big)^{\top} 
    \Big ] + o(1),
\end{align*}
where 
\begin{align*}
    \gamma(\bm{X}) = [G_{Y(1)|\bm{X}}, \dots, G_{Y(K)|\bm{X}}]^{\top}
    \ \ \ \mathrm{and} \ \ \ 
    A:= 
        \left [
    \begin{array}{cccc}
    {\pi}_{1}^{-1} - 1,
    & 
    -1,   
        & 
        \dots,  
        & 
        -1 \\
        -1,   
        &     
    {\pi}_{2}^{-1} - 1,
        & \dots,  
        & 
        -1 
        \\ 
          \vdots &\vdots&\ddots & \vdots\\ 
        -1,   
        &     
        -1, 
        & \dots,  
        & 
    {\pi}_{K}^{-1} - 1
    \end{array}
    \right ] .
    \end{align*}
    Lemma \ref{lemma:L-ind}
    with $\sum_{k=1}^{K} \pi_{k} =1$
    shows that, 
    for an arbitrary vector $v:=(v_{1}, \dots, v_{K})^{\top} \in \R^{k}$,  
    \begin{align*}
     & v^{\top}
     \big (\gamma(\bm{X}) - \E[\gamma(\bm{X})]\big) 
     A
     \big(\gamma(\bm{X}) - \E[\gamma(\bm{X})]\big)^{\top} 
     v \\ 
      & \ \ = 
     \sum_{k \in \K}
     \frac{
     v_{k}^{2}
     \big ( 
     G_{Y(k)|\bm{X}}
     -
     \E[G_{Y(k)|\bm{X}}]
     \big )^{2}
     }{
        \pi_{k}
     } 
      - 
     \bigg (
     \sum_{k \in \K}
     v_{k}
     \big ( 
     G_{Y(k)|\bm{X}}
     -
     \E[G_{Y(k)|\bm{X}}]
     \big )
     \bigg )^{2} \\ 
     & \ \ \  = 
     \frac{1}{2}
    \sum_{k \in \K} 
    \sum_{\substack{\ell \in \K \\ \ell \neq k}}
     \frac{
     \big \{ 
         v_{k}
         \big ( 
         G_{Y(k)|\bm{X}}
         -
         \E[ G_{Y(k)|\bm{X}} ]
         \big )
         {\pi}_{\ell}
         -
         v_{\ell}
         \big ( 
         G_{Y(\ell)|\bm{X}}
         -
         \E[
         G_{Y(\ell)|\bm{X}}
         ]
         \big )
         {\pi}_{k}     \big \} ^{2}
     }{
     {\pi}_{k}
     {\pi}_{\ell}
     } . 
    \end{align*}
    It follows that 
    \begin{align*} 
       v^{\top}
    \big \{
     \Var \big (
    \widehat{\gamma}^{simple}
    \big ) 
    -
    \Var \big (
    \widetilde{\gamma}
    \big )  
    \big \}
    v
    = 
    \frac{1}{2}
    \sum_{k \in \K} 
    \sum_{\substack{\ell \in \K \\ \ell \neq k}}
     \frac{
     \Var 
     \Big (
         v_{k}
         G_{Y(k)|\bm{X}}
         {\pi}_{\ell}
         -
         v_{\ell}
         G_{Y(\ell)|\bm{X}}
         {\pi}_{k}     
         \Big )
     }{
     {\pi}_{k}
     {\pi}_{\ell}
     } 
    + o(n^{-1}). 
    \end{align*}
    Because 
    $ \Var 
     \big (
         v_{k}
         G_{Y(k)|\bm{X}}
         {\pi}_{\ell}
         -
         v_{\ell}
         G_{Y(\ell)|\bm{X}}
         {\pi}_{k}     
         \big )
         \geq  0$
         for any 
         $k, \ell \in \K$ with $ k \neq \ell$,
  the above equality implies the desired positive semi-definiteness result.  

  Furthermore, the positive definite result holds when 
  $ \Var 
     \big (
         v_{k}
         G_{Y(k)|\bm{X}}
         {\pi}_{\ell}
         -
         v_{\ell}
         G_{Y(\ell)|\bm{X}}
         {\pi}_{k}     
         \big )
         >  0$
    for any $v \in \R^{K}$ with $v \neq 0$
    and 
    for any 
    $k, \ell \in \K$ with $ k \neq \ell$. 
    Because 
    $v \in \R^{K}$ is chosen arbitrarily except  $v \neq 0$
    and $\pi_{k} \in (0,1)$ for all $k \in \K$, 
    the condition for the positive definiteness can be written as  
    $
    \Var 
         \big (
         G_{Y(k)|\bm{X}}
         -
         r
         \cdot 
         G_{Y(\ell)|\bm{X}}
         \big )
         > 0
         $
         for any  $r \in \R$
         and 
         for any 
         $k, \ell \in \K$ with $ w \neq \ell$. 
\end{proof}

\subsection{Distributional Regression}
\label{subsection:DR}

In this subsection, we present
theoretical results of the distributional regression estimator
and the conditional distribution estimator. 
To obtain the asymptotic result in the subsequent section, 
establishing the convergence rate of the conditional
distribution estimator is sufficient. However, for completeness, 
we also provide their limit distributions. 
We first obtain a preliminary lemma, which establishes
the Donskerness of the first derivative of
the log likelihood.

\vspace{0.5cm}
\begin{lemma}
  \label{lemma:donsker}
  Suppose that Assumptions \ref{as:as1} and \ref{as:as2} hold. Then,  
  the function classes 
  $\{ \nabla \widehat{\ell}_{k}(\beta_{k}; y): (\beta_{k}, y)
  \in \mathcal{B} {\times} \mathcal{Y}\}$
  and 
  $\{G(y|\bm{X}): y \in \mathcal{Y} \}$
  are Donsker
  with a square-integrable envelope.    
\end{lemma}
\begin{proof}
  We define
  the function classes, for each $k \in \K$,
  \begin{eqnarray*}
    \mathcal{F}_{k}:=\{ 
    \bm{x} \mapsto T(\bm{x})^{\top} \beta_{k}: \beta_{k} \in \mathcal{B}\}
    \ \ \mathrm{and} \ \
    \mathcal{G}_{k}:=\{y \mapsto \1_{\{y \le v\}}: v \in \mathcal{Y}\}.
  \end{eqnarray*}
  Lemma 2.6.15 of
  \cite{van1996weak}
  shows that 
  each of 
  $\{\mathcal{F}_{k} \}_{k=1}^{K}$
  and 
  $\{ \mathcal{G}_{k} \}_{k=1}^{K}$
  are VC-subgraph classes.
  We can view 
  each element
  of the regressors 
  $T(\bm{x})\equiv[T^{(1)}(\bm{x}), \dots, T^{(d)}(\bm{x})]^{\top} $
  as a function,
  and then
  the function classes  
  $\{T(\bm{x})^{(j)}: j=1, \dots, d\}$
  are also 
  VC-subgraph classes.
  Let
  \begin{eqnarray*}
    \mathcal{H}:=
    \cup_{k=1}^{K}
    \big \{
    [\Lambda(\mathcal{F}_{k}) - \mathcal{G}_{k}]R(\mathcal{F}_{k})T^{(j)}(\bm{X})
    \big \}_{j=1}^{p}
  \end{eqnarray*}
  with 
  $\Lambda(\cdot)$
  being the inverse link function
  and
  $R(\cdot)=\lambda(\cdot)/\{\Lambda(\cdot)[1-\Lambda(\cdot)]\}$.
  The class $\mathcal{H}$ consists of a Lipschitz transformation
  of 
  VC-subgraph classes
  with Lipschitz coefficients 
  and 
  bounded from above by $\| T(\bm{X}) \|$
  up to some constant factor.
  Also, an envelop function for $\mathcal{H}$
  is 
  bounded from above by $\| T(\bm{X}) \|$ 
  up to some constant factor
  and thus is square integrable under Assumption A1.
  A Lipschitz composition of a Donsker class is a Donsker class.
  See 19.20 of \cite{van2000asymptotic}.
  Thus, the desired result for the first class is obtained.
  A simpler argument prove the desired conclusion
  for the other class. 
\end{proof}
\vspace{0.5cm}

The following lemma obtains
an asymptotic linear representation
for the estimator of parameters in
the distributional regression.

\vspace{0.5cm}
\begin{lemma}
  \label{lemma:argmin}
  Suppose that Assumptions \ref{as:as1} -\ref{as:as2} hold. Then,
  for each $k \in \mathcal{K}$,
  we have, uniformly in $y \in \mathcal{Y}$,
  \begin{eqnarray*}
    \sqrt{n_{k}}
    \big (
    \widehat{\beta}_{k}(y)
    -
    \beta_{k}(y)
    \big )
    =
    -
    H_{k}(y)^{-1}
    n_{k}^{-1/2}
    \sum_{i=1}^{n}
    W_{i, k}
    \nabla \ell_{i}
    \big (
    \beta_{1}(y), y  
    \big)
    + o_p(1).
  \end{eqnarray*} 
  
\end{lemma}

\begin{proof}
  Let $k \in \K$ be fixed.
  To simplify notation,
  we define a localized objective function,
  \begin{eqnarray*}
    \widehat{Q}(\delta;y)
    :=
    \frac{1}{n_{k}}
    \sum_{i=1}^{n}
    W_{i, k}
    \big \{ 
    \ell_{i, y}
    \big (\beta_{k}(y) + n_{k}^{-1/2}\delta\big)
    -
    \ell_{i, y}
    \big (\beta_{k}(y)\big)
    \big \}.
  \end{eqnarray*}
  Then, we can write the estimator 
  $\widehat{\delta}(y) := \sqrt{n_{k}}
  \big (
  \widehat{\beta}_{k}(y) - \beta_{k}(y)
  \big)$
  as the solution for 
  $\max_{\delta \in \R^{d}}
  \widehat{Q}
  (\delta;y)
  $.

  Let
  $M$ be a finite positive constant.
  Under Assumption A2,
  the objective function 
  $\delta \mapsto \widehat{Q}(\delta; y)$
  is twice continuously differentiable,
  we can show that,
  for each fixed $\delta$
  with $\|\delta \| \le M$,
  uniformly
  in $y \in \mathcal{Y}$,
  \begin{eqnarray*}
    n_{k}
    \widehat{Q}(\delta; y)
    =
    \delta^{\top}
    \sqrt{n_{k}}
    \nabla
    \widehat{Q}(0; y)
    +
    \frac{1}{2}
    \delta^{\top}
    \nabla^{2}
    \widehat{Q}(0; y)
    \delta
    +
    o(n_{k}^{-1}\|\delta\|^2).
  \end{eqnarray*}
  Also,
  an application of the uniform law of large numbers
  for the Hessian matrix
  shows that 
  $ 
  \nabla^{2}
  \widehat{Q}(0; y)
  \to^p
  H_{k}(y)
  $
  uniformly in $y \in \mathcal{Y}$.
  Thus, for each $\delta$
  with $\|\delta \| \le M$,
  we can show that 
  $
  \sup_{y \in \mathcal{Y}}
  \big |
  \widehat{Q}(\delta;y)
  -
  \widetilde{Q}(\delta;y)
  \big |
  =
  o_p(n^{-1})
  $,
  where
  \begin{eqnarray*}
    n_{k}
    \widetilde{Q}_{}(\delta; y)
    =
    \delta^{\top}
    \sqrt{n_{k}}
    \nabla
    \widehat{Q}(0; y)
    +
    \frac{1}{2}
    \delta^{\top}
    H_{k}(y)
    \delta.
  \end{eqnarray*} 
  The convexity lemma
  \citep[see][]{pollard1991asymptotics, kato2009asymptotics}
  extends
  the point-wise convergence 
  with respect to $\delta$
  to the uniform converges and thus,
  under Assumption A2,
  \begin{eqnarray}
    \label{eq:Q2}
    \sup_{y \in \mathcal{Y}}
    \sup_{\delta : \|\delta\| \le M}
    \big |
    \widehat{Q}(\delta;y)
    -
    \widetilde{Q}(\delta;y)
    \big |
    =
    o_p(n_{k}^{-1}).
  \end{eqnarray}

  Because 
  $\widetilde{Q}(\delta;y)$
  is maximized at 
  $\tilde{\delta}(y)
  :=
  -
  H_{k}(y)^{-1}
  \sqrt{n_{k}}
  \nabla
  \widehat{Q}(0; y)
  $,
  simple algebra can show that,
  for any $\delta$ and for some constant $c>0$,
  \begin{align}
    \widetilde{Q}\big(\tilde{\delta}(y);y \big)
    -
    \widetilde{Q}\big (\delta ;y \big)
    &=
    -
    \frac{1}{2 n_{k}}
    \big (
    \tilde{\delta}(y) - \delta
    \big )^{\top}
    H_{k}(y)
    \big (
    \tilde{\delta}(y) - \delta
      \big )
    \notag \\
      & \ge 
    \frac{c}{2 n_{k}}
    \| 
    \tilde{\delta}(y) - \delta
    \|^2,
    \label{eq:lbd1}
  \end{align}
  where
  the last inequality is
  due to that
  $H_{k}(y)$
  is negative definite
  under Assumption A4.  
  For any subset $D$ including $\tilde{\delta}(y)$,
  an application of the triangle inequality obtains
  \begin{align}
    2 
    \sup_{\delta \in D}
    |
    \widehat{Q}
    (\delta;y)
    -
    \widetilde{Q}
    (\delta;y)
    |
    & \ge
    \sup_{\delta \in D}
    \big \{
    \widetilde{Q}\big(\tilde{\delta}(y);y \big)
    -
    \widetilde{Q}\big (\delta;y \big)
      \big \}
    \notag \\
      & \ \ \     - 
    \sup_{\delta \in D}
    \big \{
    \widehat{Q}
    \big (\tilde{\delta}(y); y\big)
    -
    \widehat{Q}
    (\delta; y)
    \}. 
    \label{eq:tri5}
  \end{align}
  
  Let $\eta>0$ be an arbitrary constant.
  Because of the concavity under Assumption A2,
  difference quotients
  satisfy that,
  for any $\lambda > \eta$
  and
  for any $v \in S^{k}$
  with the unit sphere  
  $S^{k} \subseteq \R^{k}$,
  \begin{eqnarray*}
    \frac{
    \widehat{Q}
    \big (
    \tilde{\delta}(y)
    +
    \eta v; y
    \big)
    -
    \widehat{Q}
    \big(
    \tilde{\delta}(y); y
    \big)
    }{
    \eta
    }
    \ge
    \frac{
    \widehat{Q}
    \big(
    \tilde{\delta}(y)
    +
    \lambda v ; y
    \big)
    -
    \widehat{Q}
    \big(
    \tilde{\delta}(y) ;y
    \big)
    }{
    \lambda
    }.
  \end{eqnarray*}
  This inequality with
  a set $D_{n, \eta}(y):=
  \{\delta \in \R^{k}: \|\delta - \tilde{\delta}(y) \| \le \eta \}$
  implies
  that, 
  given the event 
  $\{
  \sup_{y \in \mathcal{Y}}\| \hat{\delta}(y) - \tilde{\delta}(y) \|
  \ge \eta
  \}$,
  we have,
  for any $y \in \mathcal{Y}$,
  \begin{eqnarray}
    \label{eq:Q1}
    \sup_{\delta \in D_{n_{k}, \eta}(y)}
    \widehat{Q}
    (\delta; y)
    -
    \widehat{Q}
    \big (\tilde{\delta}(y); y\big) \ge 0,
  \end{eqnarray}
  where
  the last inequality is due to that 
  $\widehat{Q}
  \big (\hat{\delta}(y); y\big)
  -
  \widehat{Q}
  \big ( \tilde{\delta}(y); y\big)
  \ge  0
  $,
  by definition of
  $\hat{\delta}(y)$.
  It follows from
  (\ref{eq:lbd1})-(\ref{eq:Q1}) that,
  given the event 
  $\{
  \sup_{y \in \mathcal{Y}}\| \hat{\delta}(y) - \tilde{\delta}(y) \|
  \ge \eta
  \}$, we have 
  \begin{eqnarray*}
    \sup_{\delta \in D_{n_{k}, \eta}(y)}
    |
    \widehat{Q}
    (\delta;y)
    -
    \widetilde{Q}
    (\delta;y)
    |
    \ge
    \frac{c}{4n_{k}} \eta^{2}. 
  \end{eqnarray*}
  Lemma \ref{lemma:donsker}
  shows that 
  $
  \nabla
  \widehat{Q}
  (0;y)
  =
  \nabla
  \widehat{\ell}_{k}
  \big (\beta_{k}(y); y\big)
  $  
  is Donsker. 
  Applying
  the Donsker property
  for 
  $\tilde{\delta}(y)
  :=
  -
  H_{k}(y)^{-1}
  \sqrt{n_{k}}
  \nabla
  \widehat{Q}(0; y)
  $,
  we can show that,
  for any $\xi>0$,
  there exists a constant $C$
  such that 
  $\Pr(\sup_{y \in \mathcal{Y}}\| \tilde{\delta}(y) \| \ge C)
  \le \xi
  $
  for sufficiently large $n_{k}$. 
  Thus, the above display implies that 
  \begin{align*}
         \Pr
    \bigg (
    \sup_{y \in \mathcal{Y}}\| \hat{\delta}(y) - \tilde{\delta}(y) \|
    \ge \eta 
    \bigg ) 
    \le
    \Pr 
    \bigg (
    \sup_{y \in \mathcal{Y}}
    \sup_{\delta: \|\delta\| \le \eta + C}
    \big |
    \widehat{Q}
    (\delta)
    -
    \widetilde{Q}
    (\delta)
    \big | 
       >
    \frac{c }{4n_{k}} \eta^2
    \bigg )
    +
    \xi,
  \end{align*}
  for sufficiently large $n_{k}$.
  It follows from (\ref{eq:Q2}) that
  the first term on the right side of the above equation
  converges to 0 as $n_{k} \to \infty$.
  Thus, we obtain the desired conclusion.
\end{proof}
\vspace{0.5cm}

The theorem below present the asymptotic distribution
of the estimator of the parameters of the DR model.

\vspace{0.5cm}
\begin{proposition}
  \label{pro:est}

  Under Assumptions \ref{as:as1} and \ref{as:as2},
  we have, for each $k \in \mathcal{K}$,
  \begin{eqnarray*}
    \sqrt{n_{k}}
    \big (
    \widehat{\beta}_{k}(\cdot)
    -
    \beta_{k}(\cdot)
    \big )
    \rightsquigarrow
    \mathbb{V}_{k}(\cdot)
    \ \ \mathrm{in} \ \
    \ell^{\infty}(\mathcal{Y})^{p},
  \end{eqnarray*}
  where
  $\mathbb{V}_{k}(y)$ is
  a mean-zero Gaussian process
  over  
  $\mathcal{Y}$,
  and
  its covariance function is
  given by
  $
  H_{k}(y)^{-1}
  \Sigma_{k}(y, y')
  H_{k}(y')^{-1} 
  $
  for $y, y'\in \mathcal{Y}$
  with
  $
  \Sigma_{k}(y,y'):=
  \E[ 
  W_{i,k}
  \nabla_{\beta} \ell_{i} \big( \beta_{k}(y), y \big)
  \nabla_{\beta} \ell_{i} \big( \beta_{k}(y), y \big)^{\top}
  ]$.
\end{proposition}
\begin{proof}
  Fix $k \in \mathcal{K}$. 
  In this proof, we consider the case where
  outcome variable is a continuous random variable,
  while a simpler proof can apply for the case of discrete outcome variables.
  Lemma \ref{lemma:argmin} implies that,
  uniformly in $y \in \mathcal{Y}$,
  \begin{eqnarray*}
    \sqrt{n_{k}}
    \big (
    \hat{\beta}_{k}(y)
    -
    \beta_{k}(y)
    \big )
    =
    -
    H_{k}(y)^{-1}
    \sqrt{n_{k}}
    \nabla
    \widehat{\ell}_{k}
    \big(\beta_{k}(y) ,y \big)
    + o_p(1).
  \end{eqnarray*}
  By the implicit function theorem, we can show that
  $y \mapsto \beta_{k}(y)$
  is differentiable uniformly over $y \in \mathcal{Y}$.
  Thus, the empirical process
  $\nabla    \widehat{\ell}_{k} \big(\beta_{k}(y) ,y \big)$
  is 
  stochastically equicontinuous over $\mathcal{Y}$.
  Given iid observations under Assumption A1,
  the finite dimensional convergence follows from 
  a multivariate central limit theorem.
  This with
  the stochastic equicontinuity 
  implies that 
  $\nabla    \widehat{\ell}_{k} \big(\beta_{k}(y) ,y \big)
  \rightsquigarrow
  \mathbb{G}_{k}(\cdot)
  $
  in
  $\ell^{\infty}(\mathcal{Y})^{k}$,
  where $\mathbb{G}(\cdot)$ is
  a zero-mean Gaussian process with covariance function
  $\Sigma_{k}(\cdot,\cdot)$. 
  The desired result follows,
  after taking 
  $\mathbb{V}_{k}(\cdot) =  - H_{k}(\cdot)^{-1}  \GG_{k}(\cdot) $.
\end{proof}
\vspace{0.5cm}

The distribution function estimator, $\widehat{G}_{Y(k)|\bm{X}}$
is a transformation of the estimator $\hat{\beta}_{k}(\cdot)$.
More precisely,
let 
$\mathbb{D}:=
\ell^{\infty}(\mathcal{Y})^{d} $
and 
define 
the map 
$\phi: \mathbb{D}_{\phi}\subset \mathbb{D}\mapsto \mathbb{S}_{\phi}$,
given by
\begin{eqnarray*}
  \phi(\beta)(\bm{x}, y) :=
  \Lambda\big(T(\bm{x})^{\top} \beta(y)\big).
\end{eqnarray*}
Then, we can write
\begin{align*}
  \widehat{G}_{Y(k)|\bm{X}}
  =
  \phi( 
  \widehat{\beta}_{k}
  ).
\end{align*}
It can be shown that 
the map $\phi$
is 
Hadamard differentiable
at $\beta_{k} \in \mathbb{D}_{\phi}$
tangentially to $\mathbb{D}$
with
the Hadamard derivative 
$ b_{k}
\mapsto 
\phi_{\beta_{k}(\cdot)}'(b_{k})$,
given by 
\begin{eqnarray*}
  \phi_{\beta_{k}(\cdot)}'(b_{k})(\bm{x},y)
  :=
  \lambda
  \big ( T(\bm{x})^{\top} \beta_{k}(y)\big)
  T(\bm{x})^{\top} b_{k}(y) . 
\end{eqnarray*}
We can obtain
the joint asymptotic distribution of the distribution function
estimators,
applying the functional delta method
with the Hadamard derivative in the above display.
The result is formally stated in the theorem below.

\vspace{0.5cm}
\begin{proposition}
  \label{pro:CLT}
  Under Assumptions \ref{as:as1} and \ref{as:as2},
  we have,
  \begin{eqnarray*}
    \sqrt{n}
    \left (
    \begin{array}{c}
      \widehat{G}_{Y(1)|\bm{X}}
      -
      G_{Y(1)|\bm{X}} \\
      \vdots\\ 
      \widehat{G}_{Y(K)|\bm{X}}
      -
      G_{Y(K)|\bm{X}} 
    \end{array}
    \right )
    \rightsquigarrow
    \left [
    \begin{array}{c}
      \pi_{1}^{-1/2}
      \phi_{\beta_{1}(\cdot)}'
      \big (
      \mathbb{V}_{1}
      \big ) \\
      \vdots \\ 
      \pi_{K}^{-1/2}
      \phi_{\beta_{K}(\cdot)}'
      \big (
      \mathbb{V}_{K}
      \big ) \\
    \end{array} 
    \right ]
    \ \ \mathrm{in} \ 
    \ell^{\infty}(\mathcal{X}{\times}\mathcal{Y})^{2},
  \end{eqnarray*}
  where
  $\{\mathbb{V}_{k} \}_{k=1}^{K}$
    are the mean-zero Gaussian process defined in Theorem \ref{pro:est}.
\end{proposition}
\begin{proof}

  Consider the map
  $\phi:
  \mathbb{D}_{\phi}
  \subset
  \mathbb{D}
  \mapsto 
  \mathbb{S}_{\phi}
  $,
  where 
  $ \beta \mapsto \phi(\beta)$,
  given by 
  $
  \phi(\beta)(\bm{x},y)
  =
  \Lambda
  \big (T(\bm{x})^{\top}\beta(y) \big)
  $.
  Then,
  we can write 
  $
  \widehat{G}_{Y(k)|\bm{X}}
  =
  \phi
  \big(\hat{\beta}_{k}(\cdot)
    \big)
  $
  and
  $
  G_{Y(k)|\bm{X}}
  =
  \phi\big(\beta_{k}(\cdot) \big)
  $
  for $k \in \mathcal{K}$.
  Under Assumption \ref{as:as2}, 
  the map $\phi(\cdot)$
  is shown to be Hadamard differentiable
  at
  $
  \beta_{k}(\cdot)
  $
  tangentially to
  $\mathbb{D}$
  with the derivative map
  $b
  \mapsto 
  \phi_{\beta_{k}(\cdot)}'(b)$
  for every $k \in \mathcal{K}$,
  given by 
  \begin{eqnarray*}
    \phi_{\beta_{k}(\cdot)}^{\prime}(b)(\bm{x},y)
    =
    \lambda
    \big ( T(\bm{x})^{\top}\beta_{k}(y) \big)
    T(\bm{x})^{\top} b(y).
  \end{eqnarray*}
  Applying the functional delta method
  with
  the result in Theorem \ref{pro:est},  
  we can show that, for every $k \in \mathcal{K}$,
  \begin{eqnarray*}
    \sqrt{n_{k}}
    \big (
      \widehat{G}_{Y(k)|\bm{X}} - G_{Y(k)|\bm{X}} 
    \big )
    \rightsquigarrow
    \phi_{\beta_{k}(\cdot)}'
    (
    \mathbb{B}_{k}
    )
    \ \ \mathrm{in} \ \
    \ell^{\infty}(\mathcal{X}{\times}\mathcal{Y}).
  \end{eqnarray*}
  The Gaussian processes
  $\{\B_{k}(\cdot) \}_{k=1}^{K}$
  are mutually independent
  under the random sample assumption.
  Also, we have 
  $n_{k}/n = \pi_{k} + o(1)$.
  Thus, the desired conclusion follows. 
\end{proof}
\vspace{0.5cm}

\subsection{Proof of Theorem \ref{theorem:ate-ra-fclt} }

For treatment $k \in \mathcal{K}$, 
define the empirical processes  
\begin{align*}
  \widehat{\nu}_{\bm{X}}^{(k)}
  :=
  \sqrt{n_{k}}
  \big (
  \widehat{\P}_{\bm{X}}^{(k)}
  - 
  P_{\bm{X}}^{(k)}
  \big )
  \ \ \ \mathrm{and} \ \ \ 
  \widehat{\nu}_{Y(k)}^{simple}
  :=
  \sqrt{n_{k}}
  \big (
  \widehat{\F}_{Y(k)}^{simple} -   F_{Y(k)}
  \big ),
\end{align*}
In the below lemma, we derive the joint limit of the empirical process.

\vspace{0.3cm}
\begin{lemma}
  \label{lemma:donsker2}
  Suppose that Assumption \ref{as:as1} hold.
  Then, 
  we have, 
  \begin{eqnarray*}
    \big \{
    \big (
    \widehat{\nu}_{\bm{X}}^{(k)},
    \widehat{\nu}_{Y(k)}^{simple}
    \big) 
    \big \}_{k=1}^{K}
    \rightsquigarrow
    \big \{
    \big (
    \Z_{\bm{X}}^{(k)},
    \Z_{Y(k)} 
    \big)
    \big \}_{k=1}^{K}
    \ \ \ \mathrm{in} \ \
    \ell^{\infty}(\mathcal{X}\times \mathcal{Y})^{K}.
  \end{eqnarray*}
  Here,
  $\Z_{\bm{X}}^{(k)}$ 
  and 
  $\Z_{Y(k)}^{(k)}$ 
  are 
  mean-zero Gaussian processes
  for $k \in \mathcal{K}$.
  Given an arbitrary function
  $m(y,x) \in \ell^{\infty}(\mathcal{Y}{\times}\mathcal{X})$
  having the Donsker property,
  the 
  covariance function of
  the pair 
  $\big (\Z_{\bm{X}}^{(k)}m( \cdot, \bm{X}) , \Z_{Y(k)}, \big )$
  is given by, for any $y, y' \in \mathcal{Y}$,
  \begin{align*}
    \Cov
    \Big (
    \big (
    \1_{ \{Y(k) \le y\}},
    m(y,\bm{X})
    \big )
    ,
    \big (
    \1_{ \{Y(k) \le y'\}},
    m(y',\bm{X})
    \big )
    \Big ).
  \end{align*}
  Moreover,
  the pairs
  $ \{ (\Z_{\bm{X}}^{(k)}, \Z_{Y(k)} ) \}_{k=1}^{K}$
  are statistically independent across treatments. 
\end{lemma}
\begin{proof}[\textbf{Proof}]

  Fix $k \in \mathcal{K}$. 
  By definition, we can write 
  $
  \widehat{\nu}_{Y(k)}^{simple}
  = \sqrt{n_{k}}
  \sum_{i=1}^{n}
  W_{ik}\cdot 
  \big (
  \1_{ \{Y(k) \le y\}} - F_{Y(k)}(y)
  \big )
  $
  and 
  it can be shown that 
  the class 
  $\{ \1_{ \{ Y(k) \le y \}}: y \in \mathcal{Y} \}$
  is Donsker.  
  The Donskerness of 
  $\{ \1_{ \{ Y(k) \le y \}}: y \in \mathcal{Y} \}$
  and 
  $\{ m(y, X): y \in \mathcal{Y} \}$
  implies the stochastic equicontinuity.  
  Thus, it remains to show the finite dimensional convergence,
  which
  follows from
  a simple application of the multivariate central limit theorem,
  given the iid observations under Assumption \ref{as:as1}.

  Under Assumption \ref{as:as1}, the observations
  of the treatment groups are mutually independent,
  thereby yielding
  the pairwise independence of the limit processes. 
\end{proof}

\vspace{0.5cm}
\begin{proof}[\textbf{Proof of Theorem \ref{theorem:ate-ra-fclt}}]

 \textbf{(a)}
  For treatment $k \in \mathcal{K}$,
  we have that 
  $\widehat{\F}_{Y(k)}^{simple} =
  \widehat{\P}_{\bm{X}}^{(k)}
  \widehat{G}_{Y(k)|\bm{X}}
  $, 
  provided that the link function is a canonical link function,
  and 
  we can write 
  \begin{align*}
    \widehat{\F}_{Y(k)}
    =
    \widehat{\F}_{Y(k)}^{simple}
    +
    \big(
    \widehat{\P}_{\bm{X}}
    -
    \widehat{\P}_{\bm{X}}^{(k)}
    \big)
    \widehat{G}_{Y(k)|\bm{X}}.
  \end{align*}
  This yields 
  \begin{align*}
    \sqrt{n}
    \big (
    \widehat{\F}_{Y(k)}
    -
    F_{Y(k)}
    \big )
     =
    \frac{
    \sqrt{n}
    }{
    \sqrt{n_{k} }
    }
    \sqrt{n_{k}}
    \big (
    \widehat{\F}_{Y(k)}^{simple}
    -
    F_{Y(k)}
    \big ) 
     +
    \sqrt{n}
    \big(
    \widehat{\P}_{\bm{X}}
    -
    \widehat{\P}_{\bm{X}}^{(k)}
    \big)
    \widehat{G}_{Y(k)|\bm{X}}.
  \end{align*}
  The pre-treatment variable $\bm{X}$ shares
  the population distribution $P_{\bm{X}}$
  across treatment groups
  and also  
  $
    \widehat{\P}_{\bm{X}}
    =
    \sum_{\ell=1}^{K}
    (n_{\ell}/n)
    \widehat{\P}_{\bm{X}}^{(\ell)}
  $. 
  Thus, we can write
  \begin{align*}
    \widehat{\P}_{\bm{X}}
    - 
    \widehat{\P}_{\bm{X}}^{(k)}
    &=
    \sum_{\ell =1}^{K}
    \frac{n_{\ell}}{n}
    (
    \widehat{\P}_{\bm{X}}^{(\ell)}
    -
    P_{\bm{X}}
    )
    -
    (
    \widehat{\P}_{\bm{X}}^{(k)}
    -
    P_{\bm{X}}
    ) \\ 
    & = 
    \frac{1}{\sqrt{n}}
    \bigg (
    \sum_{\ell =1}^{K}
    \frac{ \sqrt{n_{\ell}}}{ \sqrt{n}}
    \widehat{\nu}_{\bm{X}}^{(\ell)}    
    -
    \frac{ \sqrt{n}}{ \sqrt{n_{k}}}
    \widehat{\nu}_{\bm{X}}^{(k)}
    \bigg )
  \end{align*}
  The empirical processes 
  $ \{\widehat{\nu}_{ \bm{X}}^{(\ell)}(\cdot) \}_{\ell \in \mathcal{K}}$
  are stochastically equicontinuous. Thus, it can be shown that 
  $\widehat{\nu}_{\bm{X}}^{(\ell)}
  \big (\widehat{G}_{Y(k)|\bm{X}} - G_{Y(k)|\bm{X}} \big )
  = o_p(1)
  $
  for every $\ell \in \mathcal{K}$,
  given the result that 
  $\sqrt{n} \big (\widehat{G}_{Y(k)|\bm{X}} - G_{Y(k)|\bm{X}}\big ) = O_p(1)$
  by Theorem  \ref{pro:CLT}.
  Thus, we have  
  \begin{align*}
    \sqrt{n}
    (
    \widehat{\P}_{\bm{X}}
    - 
    \widehat{\P}_{\bm{X}}^{(k)}
    )
    \widehat{G}_{Y(k)|\bm{X}}
    =
    \bigg (
    \sum_{\ell =1}^{K}
     \sqrt{\pi_{\ell}}
    \widehat{\nu}_{\bm{X}}^{(\ell)}    
    -
    \frac{ 1}{ \sqrt{\pi_{k}}}
    \widehat{\nu}_{\bm{X}}^{(k)}
    \bigg )
    G_{Y(k)|\bm{X}}
    + 
    o_{p}(1),
  \end{align*}
  which leads to 
  \begin{align*}
    \sqrt{n}
    \big (
    \widehat{\F}_{Y(k)}
    -
    F_{Y(k)}
    \big ) 
    =
    \frac{
    1
    }{
    \sqrt{\pi_{k}}
    }
    \widehat{\nu}_{Y(k)}^{simple}
    +
    \bigg (
    \sum_{\ell =1}^{K}
    \sqrt{\pi_{\ell}}
    \widehat{\nu}_{\bm{X}}^{(\ell)}    
    -
    \frac{ 1}{ \sqrt{\pi_{k}}}
    \widehat{\nu}_{\bm{X}}^{(k)}
    \bigg )
    G_{Y(k)|\bm{X}}
    + 
    o_p(1).
  \end{align*}
  It follows from  
  Lemma \ref{lemma:donsker} that 
    \begin{eqnarray}
    \label{eq:B-def}
    \sqrt{n}
    \left (
    \begin{array}{c}
      \widehat{F}_{Y(1)}
      -
      F_{Y(1)} \\
      \vdots\\ 
      \widehat{F}_{Y(K)}
      -
      F_{Y(K)} 
    \end{array}
    \right )
    \rightsquigarrow
    \mathbb{B}
    :=
    \left [
    \begin{array}{c}
    \B_{1} \\
      \vdots \\ 
      \B_{K} \\
    \end{array} 
    \right ]
    \ \ \mathrm{in} \ 
    \ell^{\infty}(\mathcal{X}{\times}\mathcal{Y})^{2},
  \end{eqnarray}
  with 
  \begin{align*}
    \Z_{k} :=
    & 
      \frac{
      1
      }{
      \sqrt{\pi_{k}}
      }
      \Z_{Y(k)}
      +
      \bigg (
      \sum_{\ell =1}^{K}
      \sqrt{\pi_{\ell}}
      \Z_{\bm{X}}^{(\ell)}    
      -
      \frac{ 1}{ \sqrt{\pi_{k}}}
      \Z_{\bm{X}}^{(k)}
      \bigg )
      G_{Y(k)|\bm{X}}. 
  \end{align*}
  As $\Psi^{DTE}_{k,k'}(\cdot)$ and $\Psi^{PTE}_{k,k',h}(\cdot)$ are linear transformations, the conclusion follows.

    \textbf{(b)}
    The QTE is obtained through a difference between two inverse maps: $\Psi^{QTE}_{k,k'}(\gamma) = F{Y_{k}}^{-1}(\gamma) - F_{Y_{k'}}^{-1}(\gamma)$ for 
    $\gamma:=(F_{Y_{1}}, \dots, F_{Y_{K}})^{\top}$. 
    Given that $F_{Y_k}$ and $F_{Y_{k'}}$ are continuously differentiable with strictly positive densities $f_{Y_k}$ and $f_{Y_{k'}}$ on their respective supports, 
    we can show that 
    the functional map $\Psi^{QTE}_{k,k'}$ is Hadamard differentiable, 
    and 
    its Hadamard derivative 
    $\psi^{QTE}_{k,k', \gamma}$ at $\gamma$ is defined in the statement of this theorem 
    \citep[see Lemma 3.9.20,][]{van1996weak}.
    Applying the function delta method \citep[Theorem 3.9.4,][]{van1996weak}, 
    we can obtain the desired result. 
  \end{proof}
\vspace{0.5cm}

\subsection{Proof of Theorem \ref{theorem:ate-ra-fclt-bootstrap}}
\label{subsec:appendix-bootstrap}

For the validity of the resampling method,
we introduce the notion
of conditional weak convergence in probability,
following \cite{van1996weak}.

\vspace{0.3cm}
\begin{definition}
For some normed space $\mathbb{Q}$,
let $BL_{1}(\mathbb{Q})$
denote the space of all 
Lipschitz continuous functions
from $\mathbb{Q}$ to $[-1,1]$.
Given the original sample 
$\{(\bm{W}_{i}, \bm{X}_{i}, Y_{i})\}_{i=1}^{n}$,
consider a random element 
$D_{n}^{\ast}:=
g(\{(\bm{W}_{i}, \bm{X}_{i}, Y_{i})\}_{i=1}^{n}, \{S_{i}\}_{i=1}^{n})$
as  
a function of 
the original sample 
and 
the random weight vector 
generating 
the bootstrap draw. 
The bootstrap law of $B_{n}^{\ast}$
is said to 
consistently estimate the law of some tight random element 
$D$ or 
$D_{n}^{\ast}
\overset{p}{\rightsquigarrow} D$ 
if 
\begin{eqnarray*}
  \sup_{\kappa \in BL_{1}(\mathbb{Q})}
  \big |
  \E_{n}[\kappa(D_{n}^{\ast})]
  -
  \E[\kappa(D)]
  \big |
  \overset{p}{\to} 0,
\end{eqnarray*}
where 
$\E_{n}$ is the expectation with respect to
$\{S_{i}\}_{i=1}^{n}$
conditional on the original sample. 
\end{definition}
\vspace{0.3cm}

Using the bootstrap observations, 
we define the empirical processes, for treatment $k \in \K$,  
\begin{align*}
  \widehat{\nu}_{\bm{X}}^{(k) \ast}
  :=
  \sqrt{n_{k}}
  \big (
  \widehat{\P}_{\bm{X}}^{(k) \ast}
  - 
  \widehat{\P}_{\bm{X}}^{(k) }
  \big )
  \ \ \mathrm{and} \ \ 
  \widehat{\nu}_{Y(k)}^{simple \ast}
  :=
  \sqrt{n_{k}}
  \big (
  \widehat{\F}_{Y(k)}^{simple\ast}
  -
  \widehat{\F}_{ Y(k)}^{simple}
  \big ),
\end{align*}
where 
$  \widehat{\P}_{\bm{X}}^{(k)\ast} :=   
  n_{k}^{-1}
  \sum_{i=1}^{n}
  S_{i}
  \cdot  W_{i,k} \cdot  \delta_{\bm{X}_i}.
$

\vspace{0.5cm}
\begin{proof}
  [\textbf{Proof of Theorem \ref{theorem:ate-ra-fclt-bootstrap}}]

  \textbf{(a)}
  Applying a similar argument in 
  the proof of Theorem \ref{theorem:ate-ra-fclt},
  we can show that
  \begin{align*}
    \sqrt{n}
    \big (
    \widehat{\F}_{Y(k)}^{\ast}
    -
    \widehat{\F}_{Y(k)}
    \big ) 
      =
    \frac{
    1
    }{
    \sqrt{\pi_{k}}
    }
    \widehat{\nu}_{Y(k)}^{simple \ast}
      +
    \bigg (
    \sum_{\ell=1}^{K}
    \sqrt{\pi_{\ell}}
    \widehat{\nu}_{\bm{X}}^{(\ell) \ast}
    -
    \frac{1}{\sqrt{\pi}}
    \widehat{\nu}_{\bm{X}}^{(k) \ast}
    \bigg )
    G_{Y(k)|\bm{X}}
    +
    o_p(1),
  \end{align*}
  in $\ell^{\infty}(\mathcal{Y})$,
  for every $k \in \K$.  
  Because 
  Lemma \ref{lemma:donsker} shows that the function classes considered in the above equation are Donsker,
  Theorem 3.6.13 of \cite{van1996weak}
  with the functional delta method
  yields 
   \begin{align*}
    \sqrt{n}
    \big (
    \widehat{\F}_{Y(k)}^{\ast}
    -
    \widehat{\F}_{Y(k)}
    \big ) 
    \overset{p}{\rightsquigarrow}
    \frac{
    1
    }{
    \sqrt{\pi_{k}}
    }
    \B_{Y(k)}
      +
    \bigg (
    \sum_{\ell=1}^{K}
    \sqrt{\pi_{\ell}}
    \B_{\bm{X}}^{(\ell) }
    -
    \frac{1}{\sqrt{\pi}}
    \B_{\bm{X}}^{(k) }
    \bigg )
    G_{Y(k)|\bm{X}},
  \end{align*}
  in $\ell^{\infty}(\mathcal{Y})$.
  Furthermore, 
  the transformations $\Psi^{DTE}_{k,k'}(\cdot)$ and 
  $\Psi^{PTE}_{k,k',h}(\cdot)$ are linear in the regression-adjusted estimator,  the desired conclusion follows. 

  \textbf{(b)}
  The desired result follows by combining the argument presented in the proof of Theorem \ref{theorem:ate-ra-fclt} with the result established in part (a) above.
\end{proof}
\vspace{0.5cm}

\clearpage 
\subsection{Generalized Linear Models and Canonical Link Functions}\label{sec:glm-canonical}

This subsection provides a brief overview of generalized linear models (GLMs) with particular emphasis on the canonical link function. We begin with the fundamental structure of GLMs, discuss their connection to exponential families, and then focus on the properties of canonical link functions. To maintain generality, we use notation distinct from that in the rest of this paper. For a comprehensive treatment of GLMs, readers are referred to \cite{mccullagh1989binary} among others.

We consider a scalar random variable $Y$ and a vector of covariates $X$. The essential structure of GLM for modeling the relationship between $Y$ and $X$ is
\begin{eqnarray*}
g(\E[Y|X]) = X^{\top} \beta,
\end{eqnarray*}
where $g: \R \to \R$ is the link function and $\beta$ is a $k \times 1$ vector of unknown parameters. Note that the linear regression model $\E[Y|X] = X^{\top} \beta$ is a special case where $g(\cdot)$ is the identity map $g(\mu)=\mu$.

GLMs are fundamentally connected to exponential family distributions, which encompass commonly-used probability distributions through a unified representation of their density:
\begin{equation*}
f(y;\theta,\phi) = \exp\left(\frac{y\theta - b(\theta)}{a(\phi)} + c(y,\phi)\right).
\end{equation*}
Here, $\theta$ is the canonical parameter and captures the essential relationship between the response and its mean. The function $b(\theta)$ is the cumulant function, taking specific forms for different distributions in the family. The dispersion parameter $\phi$ enters through the scale function $a(\phi)$, while $c(y,\phi)$ ensures proper normalization of the density.
Under the exponential family, we have
\begin{align*}
\E[Y|X] = b'(\theta)
\ \ \mathrm{and} \ \
Var(Y|X) = b''(\theta)a(\phi),
\end{align*}
where $b'(\cdot)$ and $b''(\cdot)$ denote the first and second derivatives of $b(\cdot)$, respectively.

There are various choices for the link function $g(\cdot)$ in GLMs, with the canonical link function being a popular choice that relates the linear predictor $X^{\top} \beta$ to the canonical parameter $\theta$.
Specifically, let $(b')^{-1}: \R \to \R$ be the inverse function of $b'(\cdot)$
and then the math display above imples that $(b')^{-1}(\E[Y|X]) = \theta$ under the exponential family.
The link function $g(\cdot)$ is called canonical when $g(\cdot) = (b')^{-1}(\cdot)$. This canonical link function yields the relationship between the canonical parameter and the linear predictor: $\theta = X^{\top} \beta$.
The form of the canonical link varies across different distributions in the exponential family. For the normal distribution, it is the identity function $g(\mu) = \mu$, which yields the standard linear regression model. For the Bernoulli distribution, it is the logit function $g(\mu) = \log(\mu/(1-\mu))$, mapping from the probability space $(0,1)$ to the real line. For the Poisson distribution, it is the logarithm function $g(\mu) = \log(\mu)$, which ensures the predicted mean remains positive.

We consider a random sample $\{(Y_i, X_i) \in \R \times \R^{k}\}_{i=1}^{n}$ with sample size $n$ as independent copies of $(Y, X)$. Under the canonical link function, the log-likelihood is
\begin{align*}
\ell(\beta, \phi) :=
\frac{1}{n}
\sum_{i=1}^n \left(\frac{Y_i(X_i^\top\beta) - b(X_i^\top\beta)}{a(\phi)} + c(Y_i, \phi)\right).
\end{align*}
The canonical link function ensures uniqueness of the maximum likelihood estimates (when they exist) and yields simpler expressions for the score statistic and information matrix. A key property of the canonical link function is its relationship to the score function. The score function for $\beta$ is
\begin{align*}
\frac{\partial \ell(\beta, \phi)}{\partial \beta}
&=
\frac{1}{a(\phi)}
\frac{1}{n}
\sum_{i=1}^n
\big (Y_i - b'(X_i^\top\beta) \big)X_i   \\
&=
\frac{1}{a(\phi)}
\frac{1}{n}
\sum_{i=1}^n
\big (Y_i - g^{-1}(X_i^\top\beta) \big)X_i. 
\end{align*}
When $X$ includes a constant term, the maximum likelihood estimator $\hat{\beta}$ satisfying $\partial \ell(\hat{\beta}, \hat{\phi})/ \partial \beta =0$ implies
$
n^{-1}
\sum_{i=1}^n
\big (Y_i - g^{-1}(X_i^\top \hat{\beta}) \big) = 0
$
and thus
\begin{align*}
\frac{1}{n}\sum_{i=1}^n
Y_i =
\frac{1}{n}\sum_{i=1}^n
g^{-1}(X_i^\top \hat{\beta}).
\end{align*}
Consequently, the canonical link function leads to an unbiased estimator of the unconditional mean.

\newpage
\subsection{Semiparametric Efficiency Bound}
\label{subsection:Semiparametric Efficiency Bound}
In this section, we derive the semiparametric efficiency bound of the DTE at a given $y\in\mathcal Y$, denoted by $\Delta^{DTE}_{k, k'}(y)$.\footnote{We thank an anonymous referee for raising this interesting question.} The approach is analogous to that of \citet{hahn1998role}. 

First, we characterize the tangent space. To that end, the joint density of the observed variables $(Y,\bm{W}, \bm{X})$ can be written as: 
\[f(y,w,x) =  f(y|w,x)f(w|x)f(x) 
=  \prod_{k=1}^{K}\{f_k(y|x)\pi_k\}^{W_k} f(x),
\]
where $f_k(y|x) := P(Y=y|W_k=1, X=x)$ and $\pi_k = P(W_k=1|X=x)$ for all $x\in\mathcal X$. Notice that $f(w|x)$ does not depend on $x$ in our setup.

Consider a regular parametric submodel indexed by $\theta$: 
\[f(y,w,x; \theta) := \prod_{k=1}^{K} \{f_k(y|x; \theta) \pi_k(\theta)\}^{W_k} f(x; \theta),
\]
which equals $f(y,w,x)$ when $\theta = \theta_0$.

The corresponding score of $f(y,w,x; \theta)$ is given by
\begin{align*}
    s(y,w,x; \theta) := & \frac{\partial \ln f(y,w,x;\theta) }{\partial\theta} \\
    = & \sum_{k=1}^{K} \Big(W_k \dot{f}_k(y|x; \theta) + W_k\dot{\pi}_k(\theta) \Big) + \dot{f}(x; \theta),
\end{align*}
where $\dot{f}$ denotes a derivative of the log, i.e., $\dot{f}_k(y|x; \theta) = \frac{\partial \ln f_k(y|x; \theta)}{\partial \theta}$, $\dot{\pi}_k(\theta) = \frac{\partial \ln\pi_k(\theta)}{\partial \theta}$ and $\dot{f}(x; \theta) = \frac{\partial \ln f(x; \theta)}{\partial \theta}$.

At the true value, the expectation of the score equals zero. The tangent space of the model is the set of functions that are mean zero and satisfy the additive structure of the score:
\begin{align}\label{tangent-space}
  \mathcal{T} = \Big\{ \sum_{k=1}^{K} \Big(W_k s_k(y|x) + W_k s_{\pi}\Big) + s_x(x)  \Big\},
\end{align}
where $s_k(y|x)$ and $s_x(x)$ are mean-zero functions and $s_\pi$ is a real number.

The semiparametric variance bound of $\Delta^{DTE}_{k, k'}(y)$ is the variance of the projection onto the tangent space $\mathcal T$ of a mean-zero function $\psi(Y,\bm{W},\bm{X})$ with finite second order moments, which satisfies for all regular parametric submodels
\begin{equation} \label{semiparametric-bound}
\frac{\partial \Delta^{DTE}_{k,k'}(y; F_\theta)}{\partial \theta}\Big|_{\theta=\theta_0} = E[\psi(Y, \bm{W}, \bm{X}) \cdot s(Y,\bm{W},\bm{X})]\Big|_{\theta=\theta_0}.
\end{equation}
If $\psi$ itself already lies in the tangent space, the variance bound is given by $E[\psi^2].$
Now, the DTE is
\[\Delta^{DTE}_{k,k'}(y; F_\theta) = 
\iint \1_{\{Y \leq y\}} f_k(y|x; \theta) f(x; \theta)dy dx - 
\iint \1_{\{Y \leq y\}} f_{k'}(y|x; \theta) f(x; \theta)dy dx.
\]
Letting $\gamma_k^y(\bm{X}): = E[\1_{\{Y(k)\leq y\}}|\bm{X}]$ and $\gamma_{k,k'}^y(\bm{X}):=\gamma_k^y(\bm{X}) - \gamma_{k'}^y(\bm{X})$, we thus have
\begin{align*}
\frac{\partial \Delta^{DTE}_{k,k'}(y; F_\theta)}{\partial \theta} \Big|_{\theta=\theta_0} = & \iint \1_{\{Y \leq y\}} \dot{f}_k(y|x; \theta) f_k(y|x)f(x)dydx + \int \gamma^y_k(x)\dot{f}(x; \theta) f(x) dx \\
& - \iint \1_{\{Y \leq y\}} \dot{f}_k(y|x; \theta) f_k(y|x)f(x)dydx - \int \gamma^y_{k'}(x) \dot{f}(x; \theta) f(x) dx
\end{align*}
We choose $\psi(Y, \bm{W}, \bm{X})$ as 
\begin{align}\label{psi-definition}
\psi(Y,\bm{W},\bm{X}) = \frac{W_k}{\pi_k}(\1_{\{Y\leq y\}}-\gamma^y_k(\bm{X})) - \frac{W_{k'}}{\pi_{k'}}(\1_{\{Y\leq y\}}- \gamma^y_{k'}(\bm{X})) + \gamma_{k, k'}^y(\bm{X}) - \Delta^{DTE}_{k, k'}(y).
\end{align}
Notice that $\psi$ satisfies \eqref{semiparametric-bound} and that $\psi$ lies in the tangent space $\mathcal{T}$ given in \eqref{tangent-space}.
Since $\psi$ lies in the tangent space, the variance bound is given by the expected square of $\psi$:
\begin{align*}\label{efficiency-bound}
    \E[\psi(Y,\bm{W},\bm{X})^2] = & \E\bigg[ \Big(\frac{W_k}{\pi_k}(\1_{\{Y\leq y\}}-\gamma^y_k(\bm{X})) - \frac{W_{k'}}{\pi_{k'}}(\1_{\{Y\leq y\}}- \gamma^y_{k'}(\bm{X})) + \gamma_{k, k'}^y(\bm{X}) - \Delta^{DTE}_{k, k'}(y)\Big) ^2 \bigg]. 
\end{align*}

When $G_{Y(k)|\bm{X}}$ is correctly specified for $k, k'\in\mathcal K$, the asymptotic variance of the regression-adjusted DTE estimator coincides with the efficiency bound in the above equation.

Furthermore, the semiparametric efficiency of the DTE estimator naturally extends to QTE and PTE estimators under suitable regularity conditions. Since both QTE and PTE can be obtained through Hadamard differentiable transformations of the DTE estimator, they achieve their respective semiparametric efficiency bounds as shown in Theorem 25.47 of \cite{van2000asymptotic}. 

\clearpage
\section{Additional Simulation Results}\label{app:simulation}
The remaining simulation results are presented in Figures \ref{fig:dgp-1-rho0.3}-\ref{fig:dgp-4-rho0.5}. The key insights discussed in Section \ref{sec:simulation} remain consistent across these figures.

\begin{figure}[!h]
\vskip 0.2in
\begin{center}
\caption{Performance metrics of simple and regression-adjusted DTE estimators}
(DGP1, continuous outcome, $\pi_{1}=0.3$)
\includegraphics[width=0.9\columnwidth]{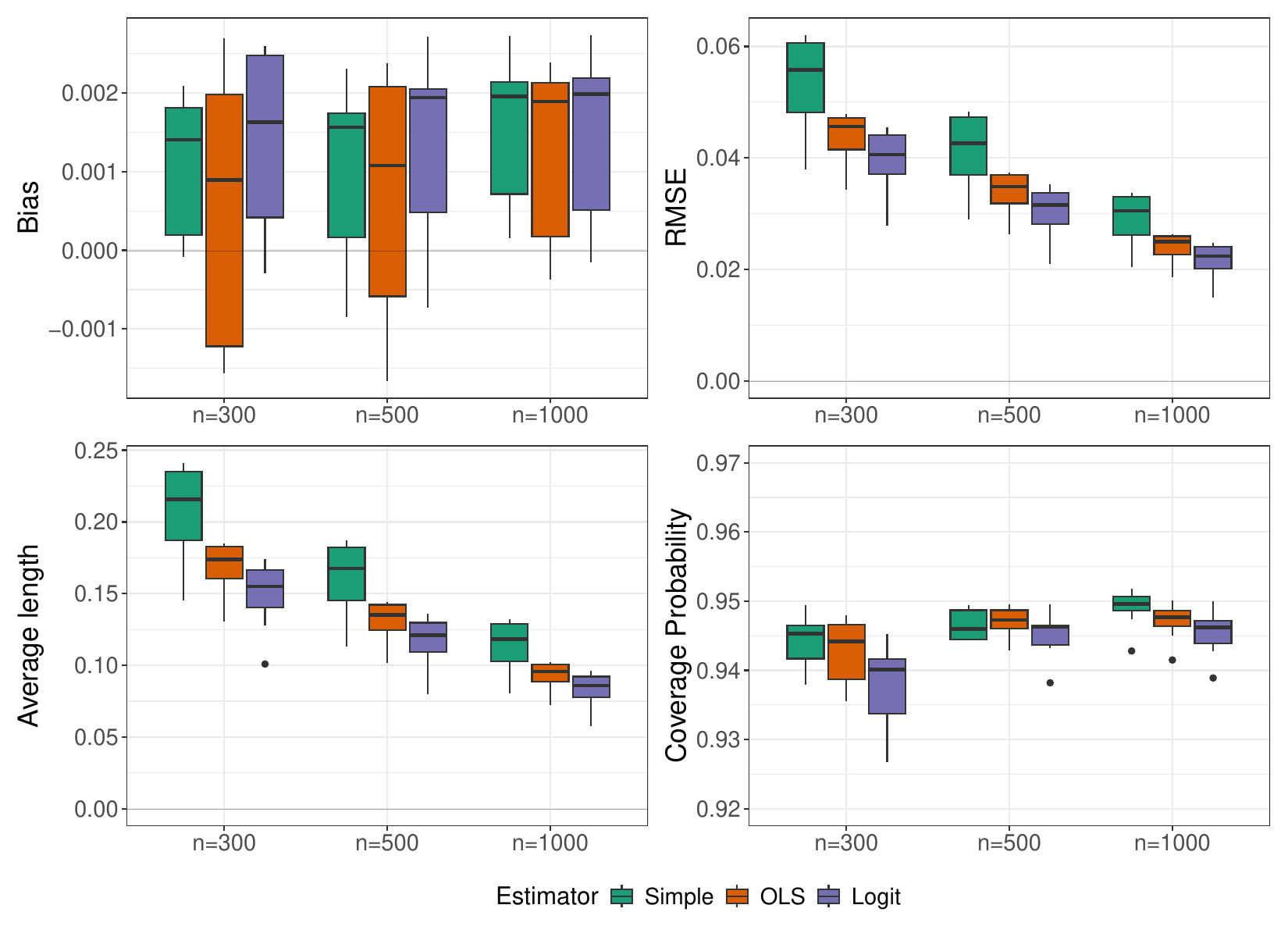}
    \begin{minipage}{0.90\textwidth}
      \small 
      \textit{Note}:
       Bias, RMSE, 95\% CI length and coverage probability calculated over 10,000 simulations. Each boxplot represents the distribution across locations $y$ for a specific sample size.
      \end{minipage}    
\label{fig:dgp-1-rho0.3}
\end{center}
\vskip -0.2in
\end{figure} 

\begin{figure}[!h]
\vskip 0.2in
\begin{center}
\caption{Performance metrics of simple and regression-adjusted DTE estimators}
(DGP2, continuous outcome, $\pi_{1}=0.3$)
\includegraphics[width=0.9\columnwidth]{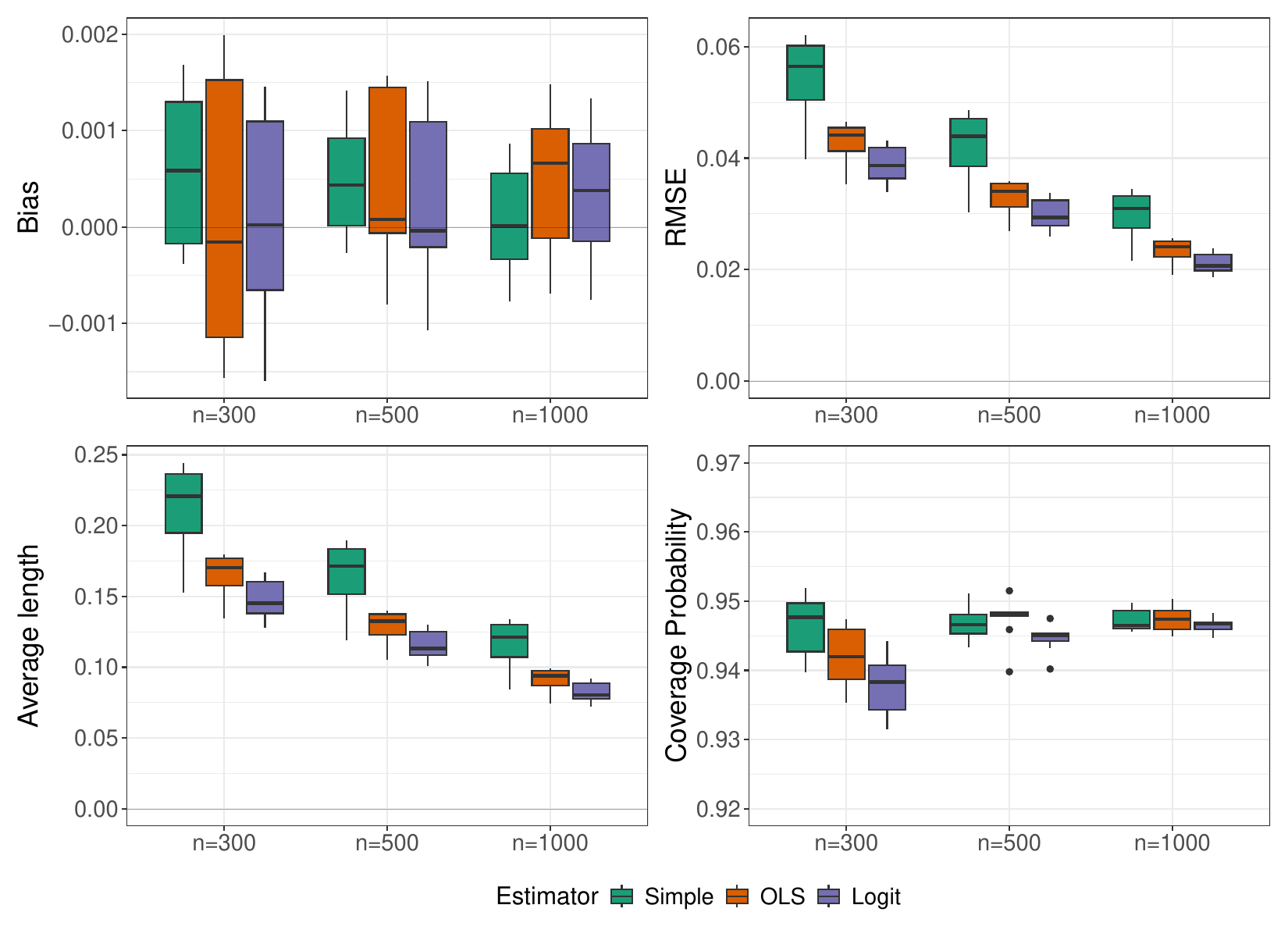}
 \begin{minipage}{0.90\textwidth}
      \small 
      \textit{Note}:
       Bias, RMSE, 95\% CI length and coverage probability calculated over 10,000 simulations. Each boxplot represents the distribution across locations $y$ for a specific sample size.
      \end{minipage}    

\label{fig:dgp-2-rho0.3}
\end{center}
\vskip -0.2in
\end{figure} 

\begin{figure}[!h]
\vskip 0.2in
\begin{center}
\caption{Performance metrics of simple and regression-adjusted DTE estimators}
(DGP2, continuous outcome, $\pi_{1}=0.5$)
\includegraphics[width=0.9\columnwidth]{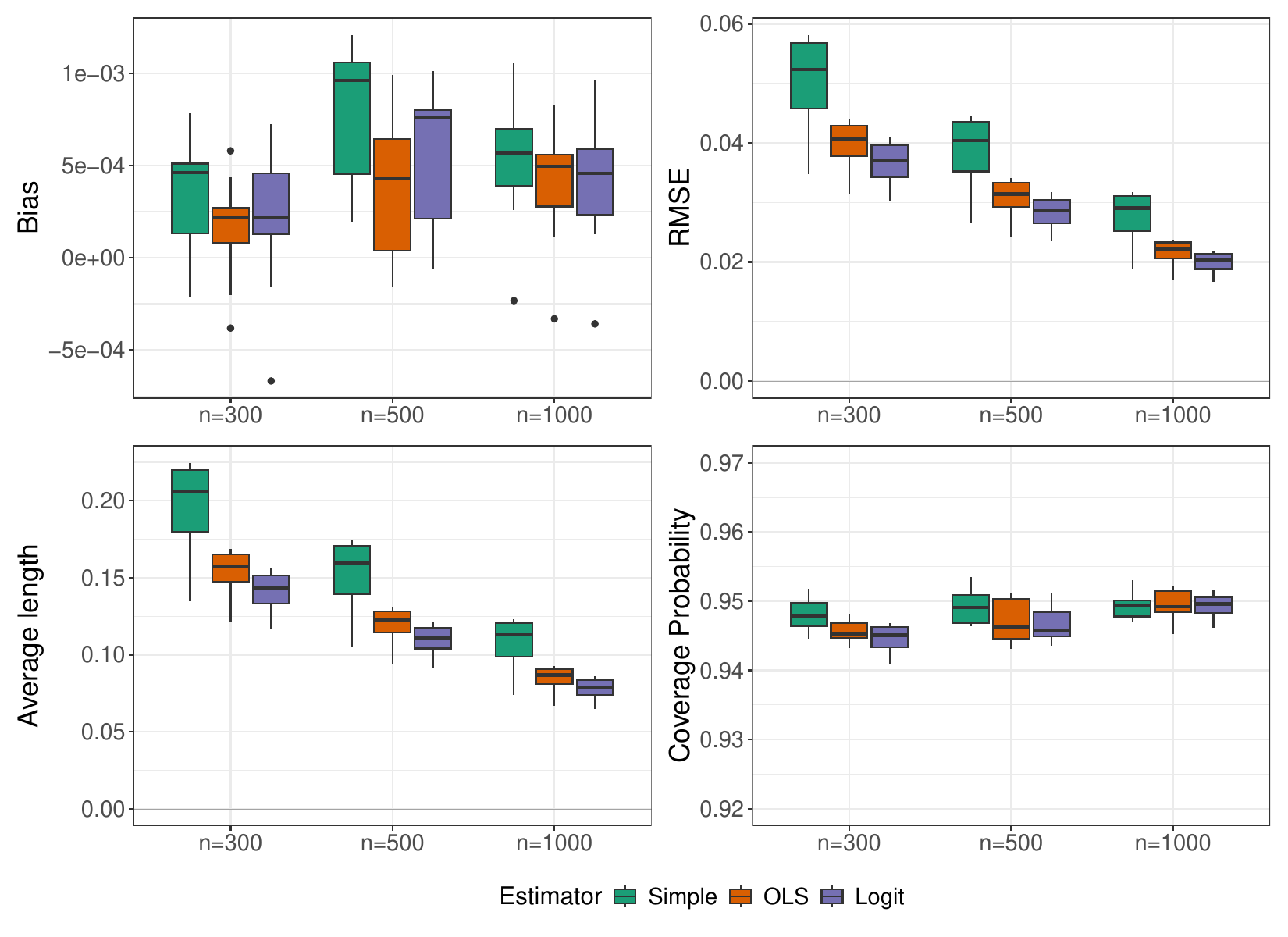}
 \begin{minipage}{0.90\textwidth}
      \small 
      \textit{Note}:
       Bias, RMSE, 95\% CI length and coverage probability calculated over 10,000 simulations. Each boxplot represents the distribution across locations $y$ for a specific sample size.
      \end{minipage}    
\label{fig:dgp-2-rho0.5}
\end{center}
\vskip -0.2in
\end{figure} 

\begin{figure}[!h]
\vskip 0.2in
\begin{center}
\caption{Performance metrics of simple and regression-adjusted DTE estimators}
(DGP3, discrete outcome, $\pi_{1}=0.3$)
\includegraphics[width=0.9\columnwidth]{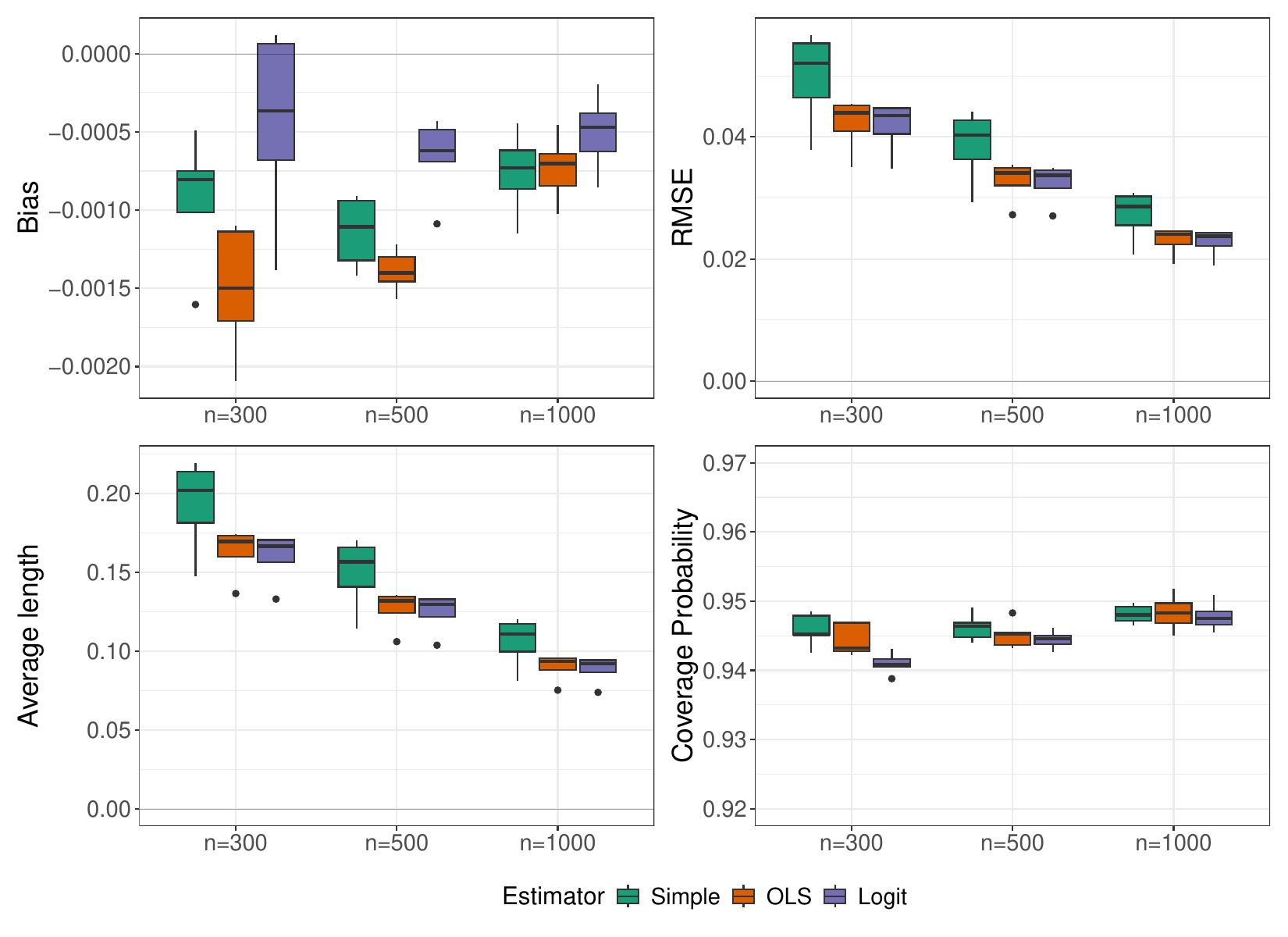}
\begin{minipage}{0.90\textwidth}
      \small 
      \textit{Note}:
       Bias, RMSE, 95\% CI length and coverage probability calculated over 10,000 simulations. Each boxplot represents the distribution across locations $y$ for a specific sample size.
      \end{minipage}    
\label{fig:dgp-3-rho0.3}
\end{center}
\vskip -0.2in
\end{figure} 

\begin{figure}[!h]
\vskip 0.2in
\begin{center}
\caption{Performance metrics of simple and regression-adjusted DTE estimators}
(DGP4, discrete outcome, $\pi_{1}=0.3$)
\includegraphics[width=0.9\columnwidth]{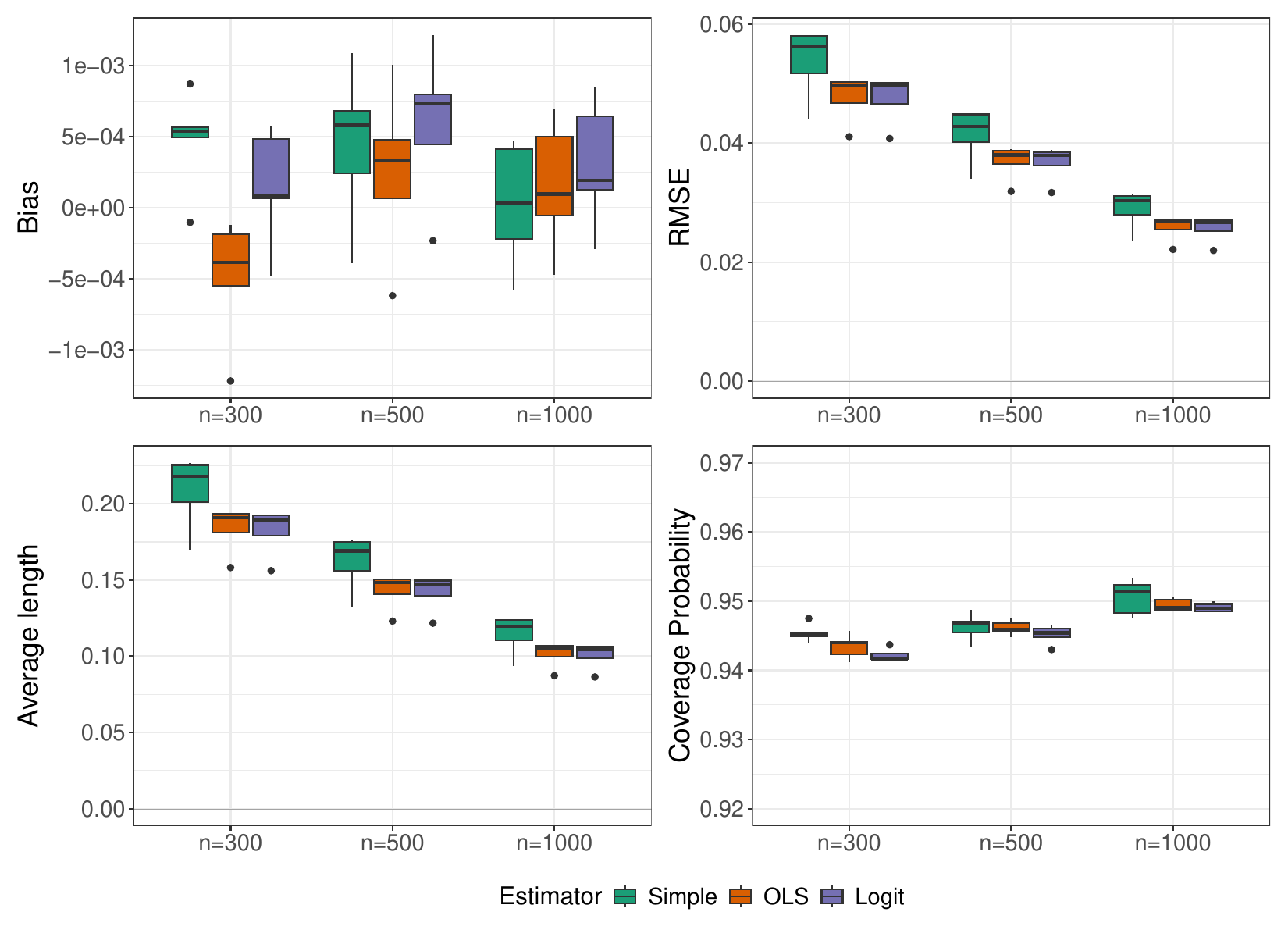}
\begin{minipage}{0.90\textwidth}
      \small 
      \textit{Note}:
       Bias, RMSE, 95\% CI length and coverage probability calculated over 10,000 simulations. Each boxplot represents the distribution across locations $y$ for a specific sample size.
      \end{minipage}    
\label{fig:dgp-4-rho0.3}
\end{center}
\vskip -0.2in
\end{figure} 

\begin{figure}[!h]
\vskip 0.2in
\begin{center}
\caption{Performance metrics of simple and regression-adjusted DTE estimators}
(DGP4, discrete outcome, $\pi_{1}=0.5$)
\includegraphics[width=0.9\columnwidth]{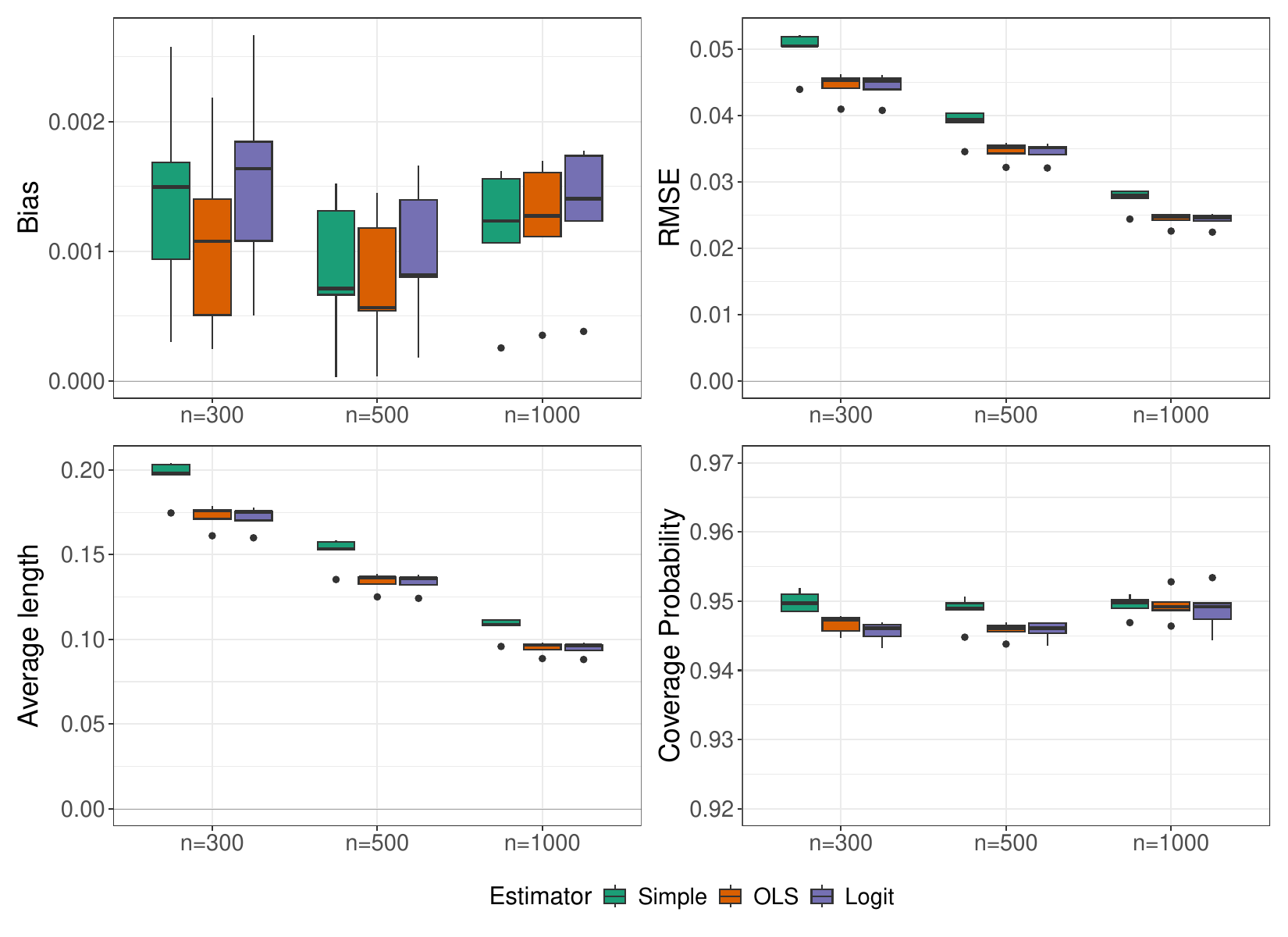}
\begin{minipage}{0.90\textwidth}
      \small 
      \textit{Note}:
       Bias, RMSE, 95\% CI length and coverage probability calculated over 10,000 simulations. Each boxplot represents the distribution across locations $y$ for a specific sample size.
      \end{minipage}    
\label{fig:dgp-4-rho0.5}
\end{center}
\vskip -0.2in
\end{figure} 

\clearpage
\subsection{Bootstrap confidence intervals} \label{app:bootstrap}
We evaluate the finite sample performance of empirical bootstrap confidence intervals outlined in Algorithm \ref{algorithm:bootstrap} in Section \ref{sec:asymptotic} by examining their average length and coverage probability across 1,000 simulations. Figure \ref{fig:dgp-1-rho0.5-boot} illustrates these results for DGP 1 (continuous outcome) with $\pi_1 = 0.5$. Standard errors are computed using two methods: bootstrap standard deviations (SD) and bootstrap interquartile ranges (IQR). The average lengths of the confidence intervals for both methods are comparable to those obtained using analytic standard errors shown in Figure \ref{fig:dgp-1-rho0.5}. Coverage probabilities remain close to 0.95, though slight under-coverage (around 0.93) is observed at certain locations when $n \in \{300, 500\}$.

\begin{figure}[!h]
\vskip 0.2in
\begin{center}
\caption{Bootstrap confidence intervals}
(DGP1, continuous outcome, $\pi_{1}=0.5$)
\includegraphics[width=0.9\columnwidth]{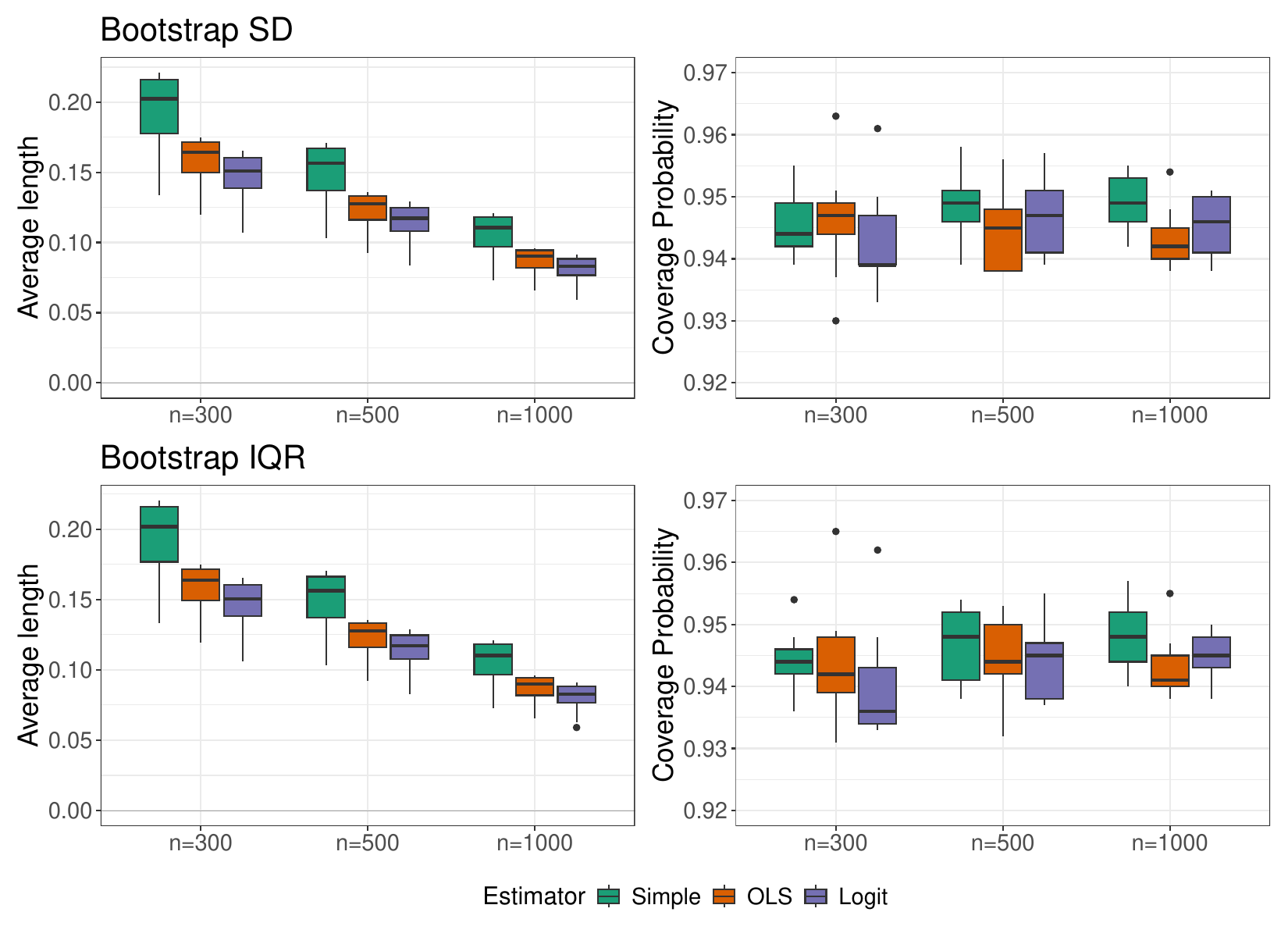}
\begin{minipage}{0.90\textwidth}
      \small 
      \textit{Note}:
       Average length and coverage probability of 95\% bootstrap confidence intervals (500 bootstrap repetitions, 1,000 simulations). Each boxplot shows the distribution across locations $y$ for a given sample size. The first row uses bootstrap standard deviations (SD) for standard errors; the second row uses bootstrap interquartile ranges (IQR) rescaled  by the normal distribution.
      \end{minipage}    
\label{fig:dgp-1-rho0.5-boot}
\end{center}
\vskip -0.2in
\end{figure} 

Figure \ref{fig:dgp-3-rho0.5-boot} displays the results for DGP3 (discrete outcome) with $\pi_1=0.5$. Again, the average lengths of the confidence intervals for both bootstrap methods are comparable to those obtained using analytic standard errors shown in Figure \ref{fig:dgp-3-rho0.5}. Coverage probabilities remain close to 0.95 in this case as well.

\begin{figure}[!h]
\vskip 0.2in
\begin{center}
\caption{Bootstrap confidence intervals}
(DGP3, discrete outcome, $\pi_{1}=0.5$)
\includegraphics[width=0.9\columnwidth]{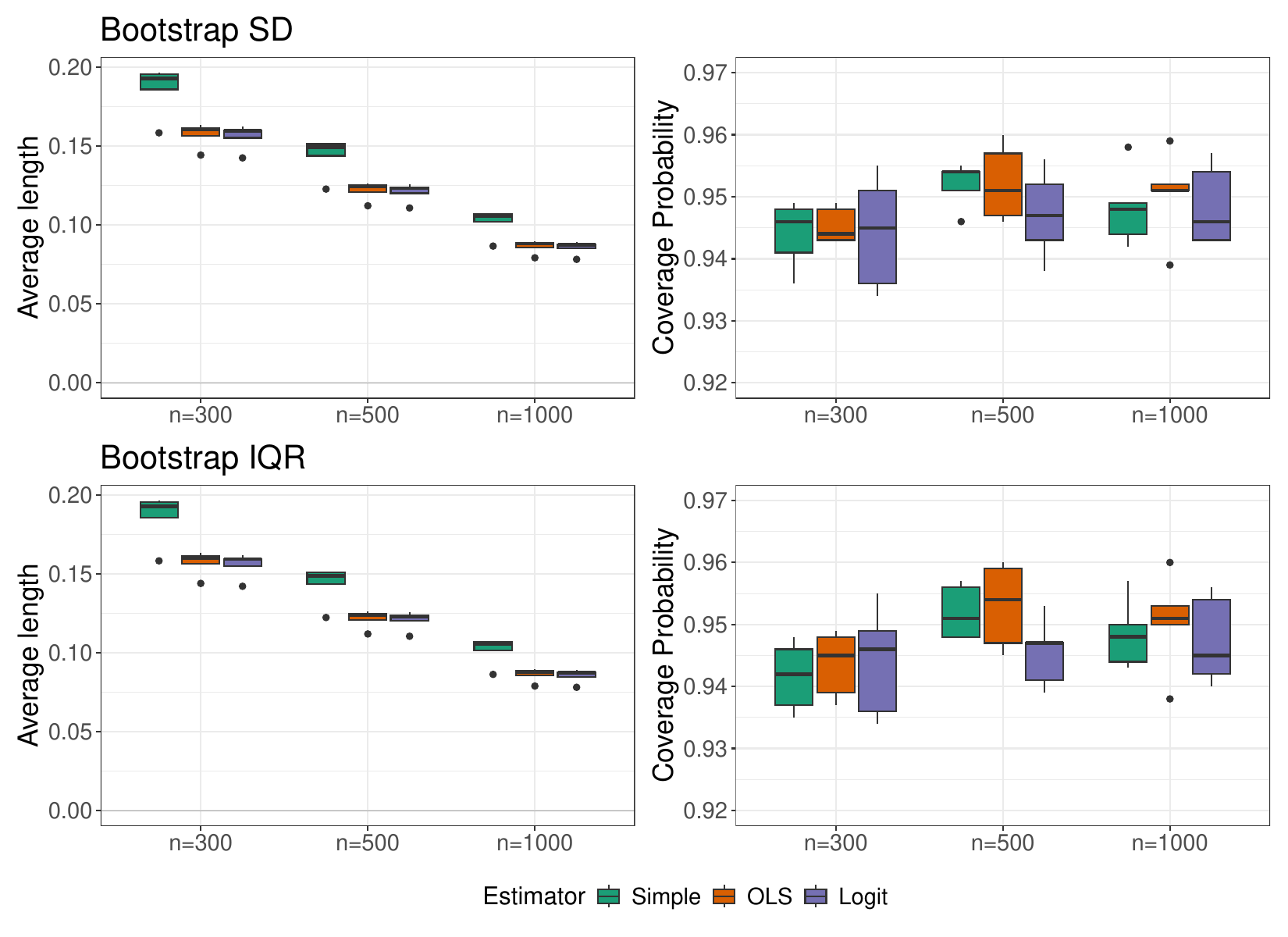}
\begin{minipage}{0.90\textwidth}
      \small 
      \textit{Note}:
       Average length and coverage probability of 95\% bootstrap confidence intervals (500 bootstrap repetitions, 1,000 simulations). Each boxplot shows the distribution across locations $y$ for a given sample size. The first row uses bootstrap standard deviations (SD) for standard errors; the second row uses bootstrap interquartile ranges (IQR) rescaled  by the normal distribution.
      \end{minipage}    
\label{fig:dgp-3-rho0.5-boot}
\end{center}
\vskip -0.2in
\end{figure} 

\clearpage
\subsection{Covariate transformations} \label{app:cov_transformation}
In the simulation study, we examine linear regression and logistic regression models without covariate transformations, i.e., with $T(X) =X$. In this section, we explore the impact of polynomial transformations on the results. Figures \ref{fig:dgp-1-rho0.5-poly2} and \ref{fig:dgp-1-rho0.5-poly3} present the simulation outcomes for DGP1 under $\pi=0.5$, where regressors are transformed using polynomials of degree 2 and 3, respectively. 

\begin{figure}[!h]
\vskip 0.2in
\begin{center}
\caption{Performance metrics of simple and regression-adjusted DTE estimators: polynomial covariate transformation (degree 2)}
(DGP1, continuous outcome, $\pi_{1}=0.5$)
\includegraphics[width=0.9\columnwidth]{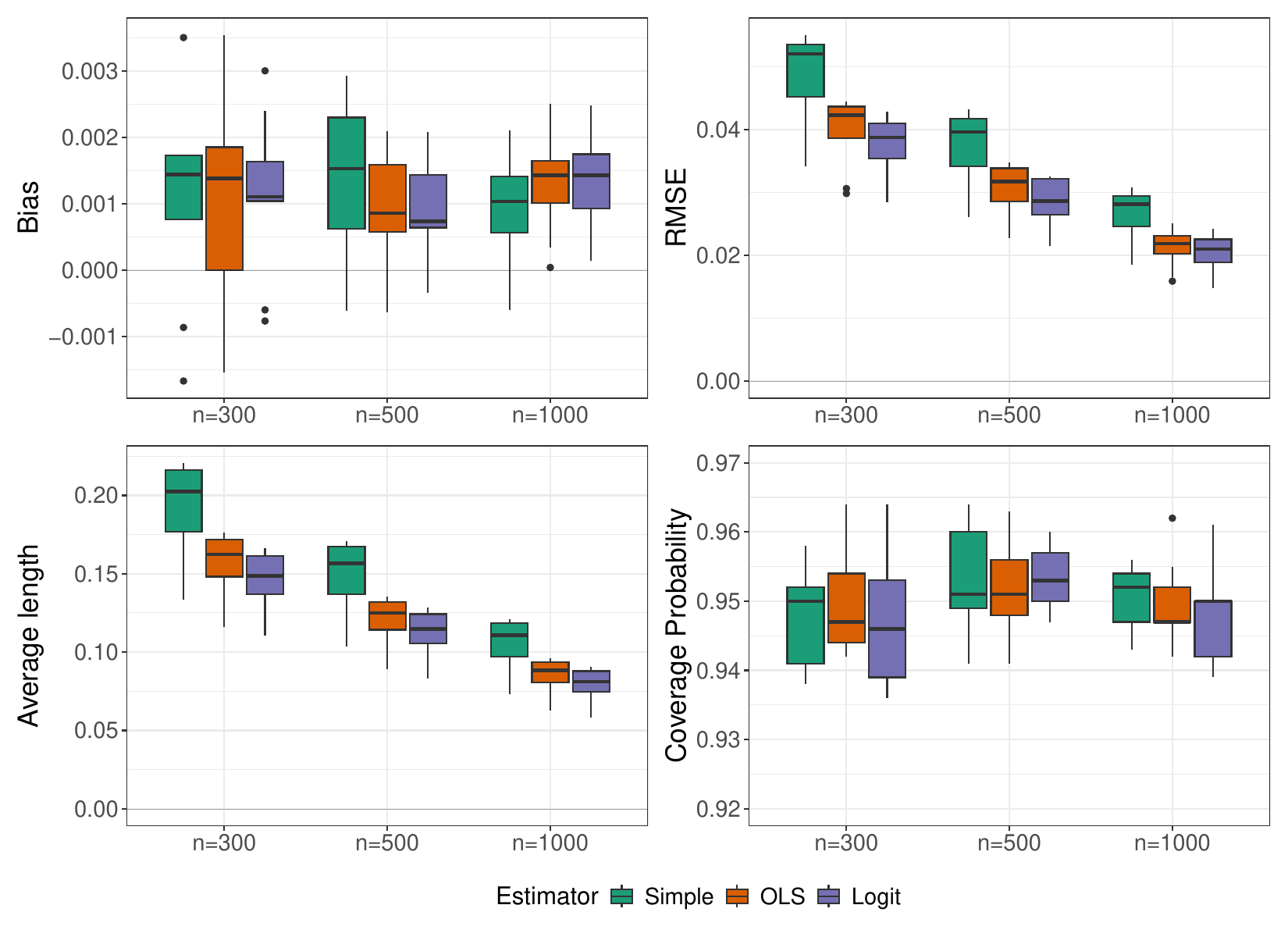}
\begin{minipage}{0.90\textwidth}
      \small 
      \textit{Note}:
       Bias, RMSE, 95\% bootstrap CI length and coverage probability calculated over 1,000 simulations. Standard errors are calculated as bootstrap standard deviation with 500 repetitions. Each boxplot represents the distribution across locations $y$ for a specific sample size. Regression adjustment is based on linear regression and logit model with covariates transformed with polynomial of degree 2. 
      \end{minipage}    
\label{fig:dgp-1-rho0.5-poly2}
\end{center}
\vskip -0.2in
\end{figure} 

For both linear regression and logit model, incorporating polynomial transformations produces results that closely resemble those obtained without transformations. Specifically, for linear adjustment, a polynomial of degree 2 achieves a 12–22\% reduction in RMSE, while a polynomial of degree 3 yields a 9–24\% reduction across all sample sizes. As for logit model, a polynomial of degree 2 achieves a 11–28\% reduction in RMSE, while a polynomial of degree 3 yields a 10-29\% reduction across all sample sizes.

\begin{figure}[!h]
\vskip 0.2in
\begin{center}
\caption{Performance metrics of simple and regression-adjusted DTE estimators: polynomial covariate transformation (degree 3)}
(DGP1, continuous outcome, $\pi_{1}=0.5$)
\includegraphics[width=0.9\columnwidth]{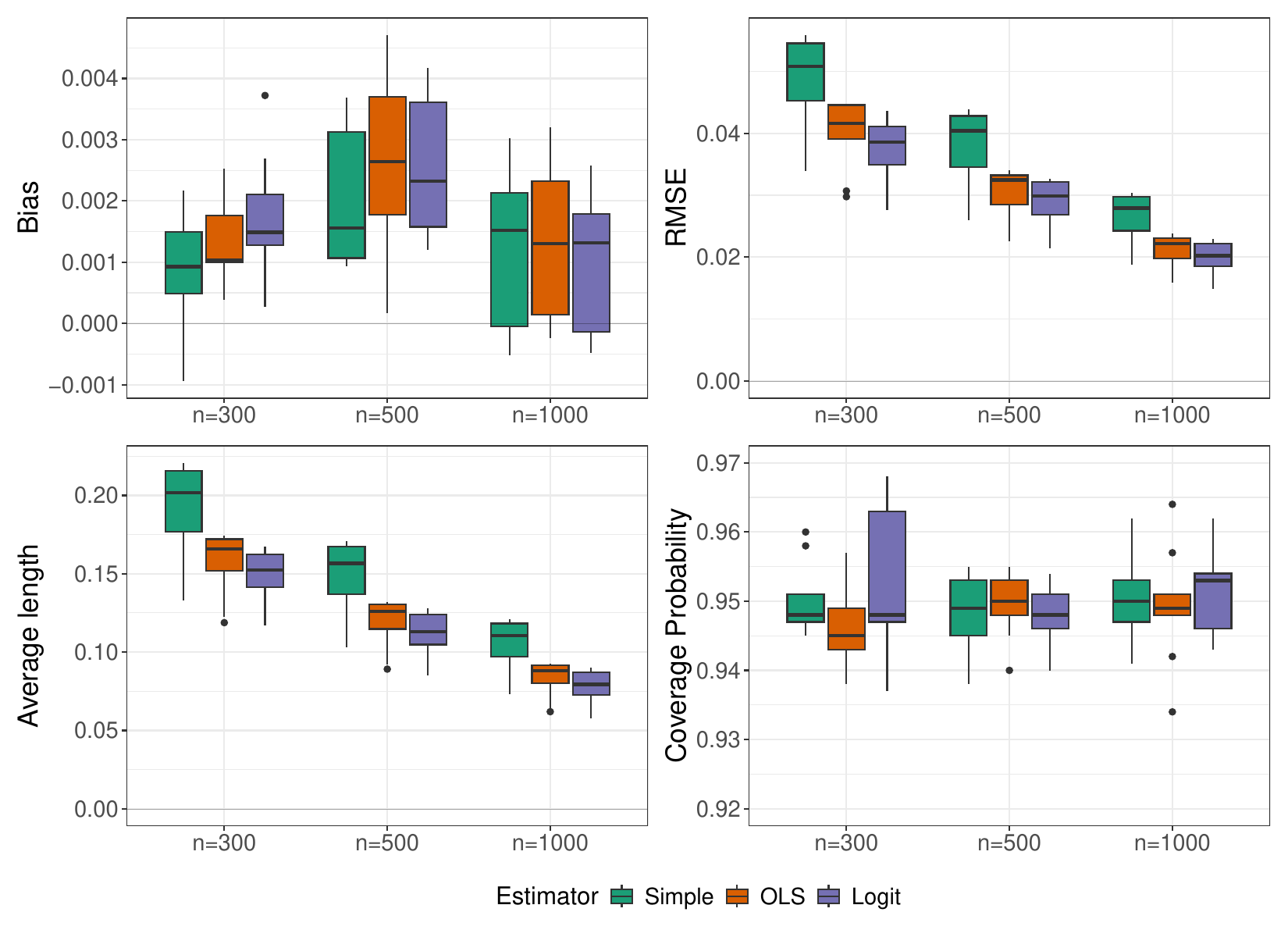}
\begin{minipage}{0.90\textwidth}
      \small 
      \textit{Note}:
      Bias, RMSE, 95\% bootstrap CI length and coverage probability calculated over 1,000 simulations. Standard errors are calculated as bootstrap standard deviation with 500 repetitions. Each boxplot represents the distribution across locations $y$ for a specific sample size. Regression adjustment is based on linear regression and logit model with covariates transformed with polynomial of degree 3. 
      \end{minipage}    
\label{fig:dgp-1-rho0.5-poly3}
\end{center}
\vskip -0.2in
\end{figure}


\end{document}